\documentclass[11pt,letterpaper]{article}

\usepackage[margin=1in]{geometry}
\usepackage{cite}
\usepackage{hyperref}
\usepackage{colortbl}
\usepackage{paralist}
\usepackage{pdfpages}
\usepackage{enumitem}
\usepackage{float}
\usepackage{booktabs}
\usepackage[override]{cmtt}
\usepackage{color,xcolor}
\usepackage{graphicx,eso-pic}
\usepackage{boxedminipage}
\usepackage{url}
\usepackage{float}
\usepackage{caption}
\usepackage{subcaption}
\usepackage{xspace}
\usepackage{multirow}
\usepackage{amsmath,amsthm,amstext,amssymb,amsfonts,latexsym}
\usepackage{wrapfig}
\usepackage[linesnumbered,ruled,vlined]{algorithm2e}
\usepackage[noend]{algpseudocode}
\usepackage{algorithmicx}
\usepackage{epstopdf}
\usepackage{mdframed}
\usepackage{cases}
\usepackage{mathrsfs}
\usepackage[capitalize]{cleveref}
\usepackage[compact]{titlesec}
\usepackage{tikz}
\usepackage{graphicx}
\usepackage{color}
\usetikzlibrary{positioning, shapes.geometric}
\usetikzlibrary{positioning,chains,fit,shapes,calc}

\usepackage{diagbox}

\usepackage{scalerel}

\newcommand\sdag{{\stretchrel*{\dag}{X}}}

\newcounter{note}[section]
\newcommand{\hubert}[1]{\refstepcounter{note}$\ll${\sf Hubert's
        Comment~\thenote:} {\sf \textcolor{blue}{#1}}$\gg$\marginpar{\tiny\bf HC~\thenote}}

\newcommand{\quan}[1]{\refstepcounter{note}$\ll${\sf Quan's
		Comment~\thenote:} {\sf \textcolor{red}{#1}}$\gg$\marginpar{\tiny\bf Quan~\thenote}}

\newcommand{\yuetodo}[1]{{\large\color{green}[Yue todo: #1]}}

\definecolor{yue}{rgb}{0.7, 0, 0}

\renewcommand{\yuetodo}[1]{}

\newcommand{\msf}[1]{\ensuremath{{\mathsf {#1}}}}
\newcommand{\mf}[1]{\ensuremath{\mathfrak {#1}}}
\newcommand{\mcal}[1]{\ensuremath{\mathcal {#1}}}

\definecolor{darkgreen}{rgb}{0,0.5,0}
\definecolor{lightblue}{RGB}{0,176,240}
\definecolor{darkblue}{RGB}{0,112,192}
\definecolor{lightpurple}{RGB}{124, 66, 168}
\definecolor{grey}{RGB}{139, 137, 137}
\definecolor{maroon}{RGB}{178, 34, 34}
\definecolor{green}{RGB}{34, 139, 34}
\definecolor{types}{RGB}{72, 61, 139}
\definecolor{gold}{rgb}{0.8, 0.33, 0.0}

\definecolor{darkgray}{gray}{0.3}

\definecolor{darkred}{rgb}{0.5, 0, 0}
\definecolor{darkgreen}{rgb}{0, 0.5, 0}
\definecolor{darkblue}{rgb}{0,0,0.5}

\newcommand\markx[2]{}

\newcommand{\R}{\mathbb{R}}

\newcommand{\B}{\mathbb{B}}

\newcommand{\side}{\iota}
\newcommand{\Ib}{\mcal{I}^{(\side)}}
\newcommand{\Iob}{\mcal{I}^{(\overline{\side})}}
\newcommand{\ob}{\overline{\side}}

\newcommand{\Izero}{\mcal{I}^{(0)}}
\newcommand{\Ione}{\mcal{I}^{(1)}}

\newcommand{\LP}{\mathsf{LP}}
\newcommand{\DP}{\mathsf{DP}}

\newcommand{\ignore}[1]{}

\newcounter{task}

\newtheorem{theorem}{Theorem}[section]

\newtheorem{corollary}[theorem]{Corollary}
\newtheorem{fact}[theorem]{Fact}
\newtheorem{lemma}[theorem]{Lemma}

{
\theoremstyle{definition}
\newtheorem{definition}[theorem]{Definition}
\newtheorem{remark}[theorem]{Remark}
\newtheorem{goal}[theorem]{Goal}
}

\newenvironment{proofof}[1]{\noindent \emph{Proof of }#1.}{\hfill \qed}

\newcommand{\Pow}{\mathsf{Pow}}

\newcounter{cnt:challenge}

\newcommand*\samethanks[1][\value{footnote}]{\footnotemark[#1]}

\newcommand{\renyi}{R\'{e}nyi\xspace}

\makeatletter
\def\moverlay{\mathpalette\mov@rlay}
\def\mov@rlay#1#2{\leavevmode\vtop{%
   \baselineskip\z@skip \lineskiplimit-\maxdimen
   \ialign{\hfil$\m@th#1##$\hfil\cr#2\crcr}}}
\newcommand{\charfusion}[3][\mathord]{
    #1{\ifx#1\mathop\vphantom{#2}\fi
        \mathpalette\mov@rlay{#2\cr#3}
      }
    \ifx#1\mathop\expandafter\displaylimits\fi}
\makeatother

\newcommand{\cupdot}{\charfusion[\mathbin]{\cup}{\cdot}}

\SetKwComment{Comment}{//}{}

\begin{document}
\title{Symmetric Splendor: Unraveling Universally Closest Refinements and Fisher Market Equilibrium through Density-Friendly Decomposition}

\author{}
\author{T-H. Hubert Chan\thanks{Department of Computer Science, the University of Hong Kong.} \and Quan Xue\samethanks}

\date{}

\begin{titlepage}

\maketitle

\begin{abstract}

We investigate the closest distribution refinements problem, which involves a vertex-weighted bipartite graph as input, where the vertex weights on each side sum to 1 and represent a probability distribution. A refinement of one side's distribution is an edge distribution that corresponds to distributing the weight of each vertex from that side to its incident edges. The objective is to identify a pair of distribution refinements for both sides of the bipartite graph such that the two edge distributions are as close as possible with respect to a specific divergence notion. This problem is a generalization of transportation, in which the special case occurs when the two closest distributions are identical. The problem has recently emerged in the context of composing differentially oblivious mechanisms.

Our main result demonstrates that a universal refinement pair exists, which is simultaneously closest under all divergence notions that satisfy the data processing inequality. Since differential obliviousness can be examined using various divergence notions, such a universally closest refinement pair offers a powerful tool in relation to such applications.

We discover that this pair can be achieved via locally verifiable optimality conditions. Specifically, we observe that it is equivalent to the following problems, which have been traditionally studied in distinct research communities: (1) hypergraph density decomposition, and (2) symmetric Fisher Market equilibrium.

We adopt a symmetric perspective of hypergraph density decomposition, in which hyperedges and nodes play equivalent roles. This symmetric decomposition serves as a tool for deriving precise characterizations of optimal solutions for other problems and enables the application of algorithms from one problem to another. This connection allows existing algorithms for computing or approximating the Fisher market equilibrium to be adapted for all the aforementioned problems. For example, this approach allows the well-known iterative proportional response process to provide approximations for the corresponding problems with multiplicative error in distributed settings, whereas previously, only absolute error had been achieved in these contexts. Our study contributes to the understanding of various problems within a unified framework, which may serve as a foundation for connecting other problems in the future.

%
%
%
%

\end{abstract}

\thispagestyle{empty}
\end{titlepage}

\section{Introduction}
\label{sec:intro}

In this work, we explore several existing and new problems 
under the same unifying framework.  By unraveling 
previously unknown connections between these seemingly unrelated problems,
we offer deeper insights into them.
Moreover, existing solutions for one problem can readily
offer improvements to other problems.  Input instances
for these problems all have the following form.

\noindent \emph{Input Instance.} An instance
is a vertex-weighted bipartite graph
with vertex bipartition $\mcal{I} = \Izero \cupdot \Ione$,
edge set $\mcal{F}$ and positive\footnote{Vertices with
zero weights are removed and identified with the empty set~$\emptyset$.} vertex weights $w: \mcal{I} \to \R_{> 0}$.
For some problems, it is convenient to normalize the vertex weights
into distributions, i.e., the vertex weights on each side sum to 1;
in this case, we call this a \emph{distribution} instance.
We use $\B := \{0,1\}$ to index the two sides.
Given a side~$\side \in \B$, we use $\ob := \side \oplus 1$ to denote the other side.

\noindent \emph{Solution Concept.}  We shall see that in each
problem, \emph{refinement} is the underlying solution concept.
Given a side $\side \in \B$, a refinement\footnote{Since vertex weights on one
side can be interpreted as a probability distribution,
allocating the weight of each vertex to its incident edges
can be viewed as a refinement of the original distribution.} $\alpha^{(\side)}$
on vertex weights on side~$\side$ indicates how each vertex~$i \in
\Ib$ distributes its weight $w(i)$ among its incident edges in $\mcal{F}$.
In other words, $\alpha^{(\side)}$ is a vector of edge weights in $\mcal{F}$
such that for each~$i \in \Ib$, the sum of weights of edges incident on~$i$
equals $w(i)$.  While some problems may require refinements of vertex weights
from one side only, it will still be insightful to consider
a \emph{refinement pair} $\vec{\alpha} = (\alpha^{(0)}, \alpha^{(1)})$
of vertex weights on both sides.

The main problem we focus on is the closest distribution refinement problem,
which we show is intricately related to other problems.
In probability theory, a \emph{divergence} notion quantifies how close two distributions are to each other,
where a smaller quantity indicates that they are closer.  Common examples
of divergences include total variation distance and KL-divergence.
Observe that a divergence notion~$\msf{D}$ needs not be symmetric, i.e.
for two distributions $P$ and $Q$ on the same sample space, it is possible that $\msf{D}(P \| Q) \neq \msf{D}(Q \| P)$.

\begin{definition}[Closest Distribution Refinement Problem]
\label{defn:refine_problem}
Given a distribution instance~$(\Izero, \Ione; \mcal{F}; w)$
and a divergence notion $\msf{D}$,
the goal is to find a distribution refinement pair $\vec{\alpha} = (\alpha^{(0)}, \alpha^{(1)})$ such that
the distributions $\alpha^{(0)}$ and 
$\alpha^{(1)}$ on the edges $\mcal{F}$ are as close to each other as possible
under~$\msf{D}$, i.e., 
$\msf{D}(\alpha^{(0)} \| \alpha^{(1)})$ is minimized.
\end{definition}

\noindent \emph{Universally Closest Refinement Pair.}
It is not clear \emph{a priori} if a refinement pair achieving the minimum
for one divergence notion would also achieve the minimum for another divergence notion.
For a non-symmetric divergence $\msf{D}$,
a pair $\vec{\alpha} = (\alpha^{(0)}, \alpha^{(1)})$ minimizing $\msf{D}(\alpha^{(0)} \| \alpha^{(1)})$
may not necessarily achieve the minimum $\msf{D}(\alpha^{(1)} \| \alpha^{(0)})$.
The special case when $\vec{\alpha}$ consists of a pair of identical refinements is also known as $w$-transportation
or perfect fractional 
matching~\cite{brualdi2006combinatorial,bhcombinatorial,schrijver2003combinatorial,behrend2013fractional},
in which case $\vec{\alpha}$ is obviously the minimizer for all divergences.

In general, it would seem ambitious to request a \emph{universally closest} refinement
pair that simultaneously minimizes all ``reasonable'' divergence notions.
In the literature, all commonly used divergence notions satisfy the \emph{data processing inequality},
which intuitively says that the divergence of two distributions cannot increase after post-processing.
Formally, for any deterministic function $\varphi: \Omega \to \widehat{\Omega}$, the induced
distribution $\varphi(P)$ and $\varphi(Q)$ on $\widehat{\Omega}$
satisfy the
 inequality: $\msf{D}(\varphi(P) \| \varphi(Q)) \leq \msf{D}(P \| Q)$.
Surprisingly, we show that such a universally closest pair exists.

\begin{theorem}[Universally Closest Refinement Pair]
\label{th:main_universal}
Given a distribution instance~$(\Izero, \Ione; \mcal{F}; w)$,
there exists
a distribution refinement pair $\vec{\alpha} = (\alpha^{(0)}, \alpha^{(1)})$  
that minimizes $\msf{D}(\alpha^{(0)} \| \alpha^{(1)})$ simultaneously
for all divergence notions $\msf{D}$ that satisfy the data processing inequality.

Moreover, this pair can be returned using
the near-linear time algorithm of Chen et al.~\cite{DBLP:conf/focs/ChenKLPGS22}
for the minimum quadratic-cost flow problem.
\end{theorem}

Since Definition~\ref{defn:refine_problem} gives
a clean mathematical problem, we first describe the technical contributions
and approaches
of this work, and we describe its motivation and application
in Section~\ref{sec:universal_motivation}.

While the technical proofs are given in Sections~\ref{sec:univeral_matching}
and~\ref{sec:dist_refine},
we give some motivation for why the problem is studied in the first place and sketch the key ideas for our proofs.

\subsection{Our Technical Contributions and Approaches}

In addition to the problem in Definition~~\ref{defn:refine_problem},
we shall see that the following solution properties
are also the key to solving other problems as well.

\noindent \emph{Desirable Solution Properties.}  Two desirable
solution properties are \emph{local maximin} and \emph{proportional
response}.

\begin{compactitem}

\item \textbf{Local Maximin.} 
A refinement~$\alpha^{(\side)}$ on vertex weights
on side~$\side$ can also be interpreted as how each vertex in~$\Ib$
allocates its weight to its neighbors on the other side~$\Iob$ via the edges in $\mcal{F}$.
Specifically, for $i \in \Ib$ and $j \in \Iob$,
$\alpha^{(\side)}(ij)$ is the \emph{payload} received by $j$ from $i$, where
we rely on the superscript in $\alpha^{(\side)}$
to indicate that the payload is sent from side~$\side$ to side~$\ob$ along an edge in $\mcal{F}$.
Consequently, $\alpha^{(\side)}$ induces
a payload vector $p \in \R^{\Iob}$, where $p(j) = \sum_i \alpha^{(\side)}(ij)$ is the total payload
received by~$j$ in the allocation~$\alpha^{(\side)}$.
Recalling that vertex~$j$ also has a weight, the payload density
(resulting from $\alpha^{(\side)}$) of $j$ is $\rho(j) := \frac{p(j)}{w(j)}$.

The refinement~$\alpha^{(\side)}$ is \emph{locally maximin}
if for each $i \in \Ib$, vertex~$i$ distributes a non-zero payload $\alpha^{(\side)}(ij) > 0$
to vertex~$j$ only if $j$ is among the neighbors of~$i$ with minimum payload densities (resulting from $\alpha^{(\side)}$).

The intuition of the term ``maximin'' is that vertex~$i$ tries to maximize the minimum payload density among
its neighbors by giving its weight to only neighbors that achieve the minimum payload density.

\item \textbf{Proportional Response.} Given a refinement~$\alpha^{(\side)}$ on vertex weights
on side~$\side$, its proportional response~$\alpha^{(\ob)}$
is a refinement on vertex weights from the other side~$\ob$.
Specifically, vertex $j \in \Iob$ distributes its weight $w(j)$ proportionately
according to its received payloads: $\alpha^{(\ob)}(ij) = \frac{\alpha^{(\side)}(ij)}{p(j)} \cdot w(j)$.

\end{compactitem}

Several problems in the literature studied
in different contexts can be expressed with the same goal.

\begin{goal}[High-Level Goal]
\label{goal}
Given an instance $(\Izero, \Ione; \mcal{F}; w)$,
find a refinement pair $\vec{\alpha} = (\alpha^{(0)}, \alpha^{(1)})$
such that both refinements are locally maximin and proportional response to each other.
\end{goal}

This abstract problem have been studied extensively in the literature under
different contexts,
and various algorithms to exactly achieve or \emph{approximate} (whose
meaning we will explain) Goal~\ref{goal} have been designed.  One contribution
of this work is to unify all these problems under our abstract framework.
Consequently, we readily offer new approximation results for some problems that are unknown before this connection is discovered.

We next describe how our abstract framework is related to
each considered problem and what new insights or results are achieved,
which also provide the tools to solve or approximate 
the problem in Definition~\ref{defn:refine_problem}.
 

\subsubsection{Hypergraph Density Decomposition}
\label{sec:hypergraph}

An instance $(\Izero, \Ione; \mcal{F}; w)$ is interpreted as a \emph{hypergraph}~$H$.
Each element in $i \in \Izero$ is a hyperedge and each element in $j \in \Ione$ is a node,
where $\{i, j\} \in \mcal{F}$ means that node~$j$ is contained
in hyperedge~$i$.  Observe that in this hypergraph, both the hyperedges and the nodes are weighted.

\noindent \textbf{Densest Subset Problem.}
Given a subset $S \subseteq \Ione$ of nodes in a (weighted) hypergraph,
its density is $\rho(S) := \frac{w(H[S])}{w(S)}$,
where $H[S]$ is the collection of hyperedges containing only nodes inside $S$.
The densest subset problem requests a subset $S$ with the maximum density.

\noindent \textbf{Density Decomposition.} The decomposition
can be described by an iterative process that generates
a density vector $\rho_* \in \R^{\Ione}$ for the nodes in the hypergraph by repeated peeling off densest subsets.
First, the (unique) maximal densest subset~$S_1$ is obtained,
and each $j \in S_1$ will be assigned the density $\rho_*(j) = \frac{w(H[S_1])}{w(S_1)}$
of~$S_1$.  Then, the hyperedges $H[S_1]$ and the subset $S_1$ are removed
from the instance; note that any hyperedge containing a node outside $S_1$ will
remain. The process is applied to the remaining instance
until eventually every node in $\Ione$ is removed and assigned its density.
The collection of subsets found in the iterative process
forms the density decomposition $\Ione = \cupdot_{k \geq 1} S_k$.

As mentioned in~\cite{harb2022faster},
the earliest variant of this decomposition was introduced by
Fujishige~\cite{fujishige1980lexicographically} in the context of polymatroids.
We describe the related work in more detail in Section~\ref{sec:related}.

\noindent \textbf{Connection with the Local Maximin Condition.}
Charikar~\cite{10.1007/3-540-44436-X_10} gave an LP (that can be easily
generalized to weighted hypergraphs) which solves the
densest subset problem exactly,
and a refinement~$\alpha^{(0)}$ on the weights of hyperedges $\Izero$
corresponds exactly to a feasible solution of the dual LP.
Replacing the linear function in this dual LP with a quadratic function
while keeping the same feasible variable~$\alpha^{(0)}$,
Danisch et al.~\cite{DBLP:conf/www/DanischCS17}
considered a quadratic program (that has also 
appeared in~\cite{fujishige1980lexicographically}).
They showed that a refinement 
corresponds to an optimal solution to the quadratic program
\emph{iff}
it is locally maximin\footnote{Instead of using the term ``maximin'',
the same concept is called ``stability'' in~\cite{DBLP:conf/www/DanischCS17}.}.
Moreover,
a locally maximin refinement also induces the aforementioned density vector 
$\rho_*$.

As observed in~\cite{harb2022faster}, the quadratic program can be expressed
as a minimum quadratic-cost flow problem, which can be solved optimally
using the nearly-linear time algorithm of Chen et al.~\cite{DBLP:conf/focs/ChenKLPGS22}.
As we shall show that our considered problems
are all equivalent to Goal~\ref{goal},
this implies the algorithmic result in Theorem~\ref{th:main_universal}.

\noindent \textbf{Symmetry of Density Decomposition.}  
We observe that the density decomposition
is symmetric between the hyperedges and the nodes.  Recall 
that in iteration~$k$ of the density decomposition,
a subset~$S_k$ is produced, together with its corresponding
collection $\widehat{S}_k$ of hyperedges that determines
the density $\frac{w(\widehat{S}_k)}{w(S_k)}$.
Therefore, in addition to a partition of nodes,
there is a corresponding decomposition of the hyperedges $\Izero = \cupdot_{k \geq 1} \widehat{S}_k$.

Given an abstract instance $(\Izero, \Ione; \mcal{F}; w)$,
one could have interpreted it differently with $\Izero$ as nodes and $\Ione$ as hyperedges.
When the density decomposition is constructed with the roles
reversed, the exact sequence $\{(\widehat{S}_k, S_k)\}_k$ of pairs is obtained,
but in reversed order.  Intuitively, $(\widehat{S}_1, S_1)$ is the densest from the perspective of $S_1 \subseteq \Ione$, but it is the least dense from the perspective of~$\widehat{S}_1 \subseteq \Izero$.

The symmetric nature of the density decomposition
has been mentioned recently by Harb et al.~\cite[Theorem 13]{DBLP:conf/esa/HarbQC23},
but we shall explain later this fact may be inferred when different
earlier works are viewed together.
Hence, this observation may be considered as part of the folklore.
Since we need some crucial properties of the symmetric density decomposition
to prove Theorem~\ref{th:main_universal},  for the sake of completeness,
the connection of Goal~\ref{goal} and this folklore result is described in 
Section~\ref{sec:sym_decomp}.

\ignore{
Essentially we show
that Goal~\ref{goal} can be achieved by first obtaining a 
locally maximin refinement~$\alpha^{(0)}$ on vertex weights on side $\Izero$
using the aforementioned quadratic program and then consider
its proportional response~$\alpha^{(1)}$ to achieve
the desired refinement pair $\vec{\alpha} = (\alpha^{(0)}, \alpha^{(1)})$.
}

\noindent \textbf{Approximation Notions.}
Even though the exact density decomposition can be achieved
via LP or maximum flow, it can be computationally expensive in practice.
In the literature, efficient approximation algorithms have been considered.
While there are different attempts to define approximation notions
for density decompositions, it is less clear what they mean when they are
applied to different problems.

Observe that a refinement pair $\vec{\alpha} = (\alpha^{(0)}, \alpha^{(1)})$
achieving Goal~\ref{goal} may not be unique, but
the induced payload vectors $(p^{(0)}_*, p^{(1)}_*)$
and density vectors $(\rho^{(0)}_*, \rho^{(1)}_*)$ are unique.
When we say \emph{approximating} Goal~\ref{goal},
we mean finding a refinement pair such that the induced
payload and density vectors have some approximation guarantees.
We shall see that approximation guarantees\footnote{See Definition~\ref{defn:error} for approximation notions
of \emph{muliplicative} vs \emph{absolute} errors on vectors.} on these vectors
can be translated readily to approximation notions for various problems.

\subsubsection{Unraveling Universally Closest Distribution Refinements via Symmetric Density Decomposition}
\label{sec:intro_universal}

\noindent \textbf{Restricted Class of
Divergence.} As a first step towards
proving Theorem~\ref{th:main_universal},
we consider a restricted class of divergences known
as \emph{hockey-stick} divergence.  For each $\gamma \geq 0$,

$$\mathsf{D}_{\gamma}(P \| Q) := \sup_{S \subseteq \Omega} Q(S) - \gamma \cdot P(S).$$

The hockey-stick divergence is related to the
well-known $(\epsilon, \delta)$-differential privacy inequality.
Specifically, $\mathsf{D}_{e^\epsilon}(P \| Q) \leq \delta$ \emph{iff}
for all subsets $S \subseteq \Omega$,

$$Q(S) \leq e^\epsilon \cdot P(S) + \delta.$$

We first consider the universal closest distribution refinements
with respect to the class of hockey-stick divergences $\mcal{D}_{\textrm{HS}} := \{\mathsf{D}_\gamma: \gamma \geq 0\}$.
In Lemma~\ref{lemma:hs_matching},
we show that the restricted variant of
Theorem~\ref{th:main_universal} to the class
$\mcal{D}_{\textrm{HS}}$ is equivalent
to the following problem.

\noindent \textbf{Universal Refinement Matching Problem.}
Given an instance $G = (\Izero, \Ione; \mcal{F}; w)$
and two non-negative weights $\vec{c} = (c^{(0)}, c^{(1)}) \in \R^{\B}$, 
we use $G^{(\vec{c})}$ to denote the instance
of (fractional) maximum matching on the bipartite graph
$(\Izero, \Ione; \mcal{F})$ with vertex capacity constraints.
For side~$\side$ and $i \in \Ib$,
the capacity constraint at vertex~$i$ is $c^{(\side)} \cdot w(i)$.
A fractional solution consists of assigning non-negative weights
to the edges in $\mcal{F}$ such that the weighted degree of every vertex
does not exceed its capacity constraint.
Observe that for any refinement pair $\vec{\alpha} = (\alpha^{(0)}, \alpha^{(1)})$,
the fractional matching $\vec{c} \bullet \vec{\alpha}$ defined below
is a feasible solution to $G^{(\vec{c})}$:

$\vec{c} \bullet \vec{\alpha}(ij) := \min \{c^{(0)} \cdot \alpha^{(0)}(ij), 
c^{(1)} \cdot \alpha^{(1)}(ij)\}$
for all $\{i, j\} \in \mcal{F}$.

The goal is to find
a refinement pair $\vec{\alpha} = (\alpha^{(0)}, \alpha^{(1)})$
such that for any weight pair
$\vec{c}$,
the edge weights $\vec{c} \bullet \vec{\alpha}$
form an optimal (fractional) solution to
 the maximum matching instance $G^{(\vec{c})}$.

The structure of the symmetric density decomposition described in Section~\ref{sec:hypergraph}
makes it convenient to apply a primal-dual analysis
to achieve the
following in Section~\ref{sec:univeral_matching}.

\begin{theorem}[Universal Maximum Refinement Matching]
\label{th:main_matching}
Given an instance~$(\Izero, \Ione; \mcal{F}; w)$,
suppose a distribution refinement pair $\vec{\alpha} = (\alpha^{(0)}, \alpha^{(1)})$  satisfies Goal~\ref{goal}.
Then, for any 
$\vec{c} = (c^{(0)}, c^{(1)})$,
the (fractional) solution
$\vec{c} \bullet \vec{\alpha}$ is optimal for the maximum matching instance $G^{(\vec{c})}$.
\end{theorem}

\noindent \textbf{Extension to General Case.} To extend Theorem~\ref{th:main_universal}
beyond the restricted divergence class
$\mcal{D}_{\textrm{HS}}$,
we consider the closest distribution refinement
problem under a more general probability distance
notion.

\noindent \emph{Capturing All Data Processing Divergences via Power Functions.}
As opposed to using a single number by a divergence notion
to measure the closeness of two distributions,
the \emph{power function}
$\Pow(P\|Q): [0,1] \to [0,1]$ between distributions
$P$ and $Q$ captures enough information to recover all divergence notions that satisfy
the data processing inequality.  
We will give a formal explanation in Section~\ref{sec:dist_refine}.
A main result is given in Theorem~\ref{th:local_power}, which states that
any distribution refinement pair satisfying Goal~\ref{goal}
achieves the minimal power function,
from which Theorem~\ref{th:main_universal} follows immediately.

\noindent \textbf{Approximation Results.}  As mentioned in Section~\ref{sec:hypergraph},
we also show that a refinement that induces a payload with multiplicative error
can lead to approximation guarantees for other problems.
For the universal refinement matching problem, 
a natural approximation notion is based on the standard approximation ratio
for the maximum matching problem.

On the other hand, achieving a meaningful approximation notion
for the universally closest distribution refinement problem in terms of divergence values is more difficult.
The reason is that any strictly increasing function on a data processing divergence
still satisfies the data processing inequality.  Therefore, any non-zero deviation
from the correct answer may be magnified arbitrarily.
Instead of giving approximation guarantees directly on the divergence values,
we consider a technical notion of power function stretching,
and the approximation quality of a refinement pair is quantified
by how much the optimal power function is stretched.  The formal statement
is given in Theorem~\ref{th:min_power}.

\subsubsection{Achieving Multiplicative Error via Proportional Response in Fisher Markets}

We next interpret an instance~$(\Izero, \Ione; \mcal{F}; w)$
as a special symmetric case of linear Fisher markets~\cite{fisher1892},
which is itself a special case of the more general Arrow-Debreu markets~\cite{arrow1954existence}.  This allows us to
transfer any algorithmic results for Fisher/Arrow-Debreu market equilibrium
to counterparts for Goal~\ref{goal}.  Specifically,
we can use a quick iterative procedure \emph{proportional response}~\cite{wu2007proportional,DBLP:conf/icalp/Zhang09}
to approximate the densities of vertices in the density decomposition
in Section~\ref{sec:hypergraph} with multiplicative error,
where only absolute error has been achieved 
before~\cite{DBLP:conf/www/DanischCS17,harb2022faster}.
In addition to being interesting in its own right,
multiplicative error is also used to 
achieve the approximation results in Section~\ref{sec:intro_universal}.

\noindent \textbf{Market Interpretation.}
Even though in an instance, the roles of the two sides 
are symmetric, it will be intuitive to think of
$\Izero$ as buyers and $\Ione$ as sellers.
For buyer $i \in \Izero$, the weight $w(i)$ is its \emph{budget}.
In the symmetric variant, the utilities of buyers have a special form.
Each seller~$j \in \Ione$ has one divisible unit of good~$j$
for which a buyer~$i$ has a positive utility exactly when  
$\{i, j\} \in \mcal{F}$, in which case the utility per unit of good~$j$
is $w(j)$ (that is common for all buyers that have a positive utility for~$j$).
For a general Fisher market, a buyer~$i$ can have
an arbitrary utility $w_i(j)$ for one unit of good~$j$.

\noindent \emph{Market Allocation and Equilibrium.}
For the symmetric variant, an allocation 
can be represented by a refinement pair $\vec{\alpha} = (\alpha^{(0)}, \alpha^{(1)})$,
in which $\alpha^{(1)}$ must be a proportional response to $\alpha^{(0)}$.
For the refinement $\alpha^{(0)}$ on the buyers' budgets,
the induced payload vector $p^{(1)}$ on the sellers corresponds to the \emph{price} vector
on goods.
For the refinement $\alpha^{(1)}$ on the sellers' weights,
the induced payload vector $p^{(0)}$ on the buyers corresponds
to the utilities received by the buyers.
It is well-known (explained in Fact~\ref{fact:fisher_local})
that
a market equilibrium is characterized by a locally maximin
refinement~$\alpha^{(0)}$ of buyers' budgets
together with its proportional response~$\alpha^{(1)}$.
From our theory of symmetric density decomposition in Section~\ref{sec:hypergraph},
this is equivalent to Goal~\ref{goal}.

\noindent \textbf{Alternative Characterization.}
The fact that the local maximin
condition on the buyers' budget allocation
is equivalent to Fisher market equilibrium
readily allows us to borrow any algorithmic results
for market equilibrium to solve or approximate
Goal~\ref{goal}.  However, our framework allows us to
make the following observations.

\begin{compactitem}

\item In linear Fisher markets (not necessarily symmetric), we show in Theorem~\ref{th:fisher_local_seller}
that
the local maximin condition on sellers' goods allocation to buyers
can be used to characterize a market equilibrium.
This is 
an alternative characterization of linear Fisher market equilibrium
that involves only goods allocation and buyers' utilities, but not
involving good prices.

\item On the other hand, we show in Theorem~\ref{th:counterexample} that
there exists an Arrow-Debreu market such that
there is a goods allocation satisfying the local maximin condition,
but it is not an equilibrium.
\end{compactitem}

\ignore{
demonstrated that computing the decomposition can be formulated as solving a quadratic program that extends the LP of Charikar~\cite{10.1007/3-540-44436-X_10}, and they applied the Frank-Wolfe algorithm. Their algorithm requires $O(\frac{|E|\Delta(G)}{\epsilon^2})$ iterations (each taking $O(|E|)$ time) to converge to an $\epsilon$-approximate solution, where $\Delta(G)$ is the degree of graph $G$. Their work demonstrates that the Frank-Wolfe algorithm provides a good approximation to the densest subgraph with a relatively small number of iterations based on experimental results.

Subsequently, Elfarouk et al.~\cite{harb2022faster} proposed an algorithm based on FISTA that takes $O(\frac{\sqrt{|E|\Delta(G)}}{\epsilon})$ iterations, each taking $O(|E|)$ time, to converge to an $\epsilon$-approximate solution. 
}

\subsection{Motivation and Application for Universally Closest Distribution Refinements}
\label{sec:universal_motivation}

The problem in Definition~\ref{defn:closest_refine} arises in the
study of the composition of differentially oblivious 
mechanisms~\cite{DBLP:conf/eurocrypt/ZhouSCM23,DBLP:conf/innovations/ZhouZCS24},
which is a generalization of the more well-known concept
of differential privacy (DP)~\cite{DBLP:conf/icalp/Dwork06}.

\ignore{
Since differentially oblivious mechanisms form a wide research area~\cite{DBLP:journals/jacm/ChanCMS22}
and the background explanation does not affect the technical understanding of this work,
we will only highlight the aspects that are the most relevant to the universally closest
refinement problem, which may be regarded  
as an interesting mathematical problem in its own right.
Readers interested in learning more about the background
may study the relevant references in more detail.
}

\noindent \textbf{A Quick Introduction to Differential Privacy/Obliviousness.}
Consider a (randomized) mechanism $\mcal{M}$ that takes an input from space $\mcal{X}$
and returns an output from space~$\mcal{Y}$ that will be passed to a second mechanism.
During the computation of $\mcal{M}$, an adversary can observe some \emph{view} from space~$\mcal{V}$ that might not necessarily include the output in general.
Suppose at first we just consider the privacy/obliviousness of $\mcal{M}$ on its own, which measures how much information the adversary can learn about the input from its view.
Intuitively, it should be difficult for the adversary to distinguish between two similar inputs
from their corresponding views produced by the mechanism.
Formally, similar inputs in $\mcal{X}$ are captured by a \emph{neighboring} symmetric relation, and the privacy/obliviousness of the mechanism $\mcal{M}$
requires that the two distributions of views produced by
neighboring inputs are ``close'', where closeness can be quantified
by some notion of divergence.  As mentioned in Section~\ref{sec:intro_universal},
the popular notion of $(\epsilon, \delta)$-DP is captured
using hockey-stick divergences.

\noindent \textbf{Composition of Differentially Oblivious Mechanisms.}
Suppose the output of mechanism~$\mcal{M}$ from space $\mcal{Y}$ is passed as an input
to a second mechanism~$\mcal{M}'$, which is differentially oblivious with respect
to some neighboring relation on $\mcal{Y}$.
Hence, given input~$x$ to $\mcal{M}$, one should consider
the joint distribution $\mcal{M}(x)$ on the space
$\mcal{V} \times \mcal{Y}$ of views and outputs.

The notion of neighbor-preserving differential obliviousness (NPDO)
has been introduced~\cite{DBLP:conf/eurocrypt/ZhouSCM23,DBLP:conf/innovations/ZhouZCS24}
to give sufficient conditions on $\mcal{M}$ that allows composition with 
such a mechanism $\mcal{M}'$ so that the overall composition
preserves differential obliviousness.


Specifically, given neighboring inputs~$x_0 \sim x_1$,
conditions are specified on the two distributions $\mcal{M}(x_0)$
and $\mcal{M}(x_1)$.  Unlike the case for differential privacy/obliviousness, 
it is not appropriate to just require that the two distributions are close.
This is because
the adversary can only observe the view component
of $\mcal{V} \times \mcal{Y}$.  Moreover, one also needs to incorporate
the neighboring relation on the output space~$\mcal{Y}$, which
is needed for the privacy guarantee for the second mechanism~$\mcal{M}'$.

The definition of NPDO is more conveniently described
via the terminology offered by our abstract instance $(\Izero, \Ione; \mcal{F}; w)$.
Here, the sets $\Izero$ and $\Ione$ are two copies of $\mcal{V} \times \mcal{Y}$.
For each $\side \in \B$, the weights $w$ on $\Ib$ corresponds
to the distribution $\mcal{M}(x_\side)$.
Finally, $(v_0, y_0) \in \Izero$ and $(v_1, y_1) \in \Ione$ are neighbors in $\mcal{F}$
\emph{iff} $v_0 = v_1$ and $y_0 \sim y_1$ are neighboring in $\mcal{Y}$.

In the original definition~\cite{DBLP:conf/innovations/ZhouZCS24}, 
the mechanism~$\mcal{M}$ is said to be $d$-NPDO with respect to
divergence~$\mathsf{D}$ if 
for the constructed instance $(\Izero, \Ione; \mcal{F}; w)$,
there exists a distribution refinement pair $\vec{\alpha} = (\alpha^{(0)}, \alpha^{(1)})$,
which may depend on the divergence $\mathsf{D}$,
such that $\mathsf{D}(\alpha^{(0)} \| \alpha^{(1)}) \leq d$.
Examples of considered divergences include the standard hockey-stick divergence and \renyi divergence~\cite{DBLP:conf/csfw/Mironov17}.

Since it only makes sense to consider divergences satisfying 
the data processing inequality in the privacy context,
the universally closest refinement pair guaranteed
in Theorem~\ref{th:main_universal} makes the notion of NPDO conceptually simpler,
because the closest refinement pair does not need to
specifically depend on the considered divergence.

We believe that this simplification will be a very useful tool.
In fact, to achieve advanced composition in~\cite{DBLP:conf/innovations/ZhouZCS24},
a very complicated argument has been used to 
prove the following fact, which is just an easy corollary of
Theorem~\ref{th:main_universal}.

\begin{fact}[Lemma~B.1 in~\cite{DBLP:conf/innovations/ZhouZCS24}]
Given a distribution instance $(\Izero, \Ione; \mcal{F}; w)$,
suppose there exist two refinement pairs
$\vec{\alpha}_s =  (\alpha^{(0)}_s, \alpha^{(1)}_s)$ and 
$\vec{\alpha}_t = (\alpha^{(0)}_t, \alpha^{(1)}_t)$ such that
the hockey-stick divergences $\mathsf{D}_\gamma(\alpha^{(0)}_s \| \alpha^{(1)}_s)$
and $\mathsf{D}_\gamma(\alpha^{(1)}_t \| \alpha^{(0)}_t)$ are both at most~$\delta$.

Then, there exists a single refinement pair
$\vec{\alpha} = (\alpha^{(0)}, \alpha^{(1)})$
such that both
$\msf{D}_\gamma(\alpha^{(0)} \| \alpha^{(1)})$
and $\msf{D}_\gamma(\alpha^{(1)} \| \alpha^{(0)})$ are at most $\delta$.
\end{fact}

\subsection{Related Work}
\label{sec:related}

The densest subset problem is a classical problem in combinatorial optimization with numerous real-world applications in data mining, network analysis, and machine learning \cite{goldberg1984finding, 10.1007/3-540-44436-X_10, khuller2009finding, ma2020efficient, angel2012dense, shin2016corescope, li2020flowscope}, and it is the basic subroutine in defining the density decomposition.

\noindent \textbf{Density Decomposition.}
As mentioned in Section~\ref{sec:hypergraph},
the earliest reference to the density decomposition and its
associated quadratic program were introduced by Fujishige~\cite{fujishige1980lexicographically}.
Even though this decomposition was defined for the more general submodular or supermodular case,
the procedure has been rediscovered independently many times in the literature
for the special bipartite graph case (or equivalently its hypergraph interpretation).

Earlier we described the decomposition procedure by repeatedly
peeling off maximal densest subsets in a hypergraph.
This was the approach by Tatti and Gionis~\cite{tatti2015density, tatti2019density}
and Danisch et al.~\cite{DBLP:conf/www/DanischCS17} .
Moreover, the same procedure has been used to define hypergraph Laplacians
in the context of the spectral properties of hypergraphs~\cite{DBLP:journals/jacm/ChanLTZ18}.

In the literature, there is another line of research that defines
the decomposition by repeatedly peeling off least densest subsets instead.
However, given a subset~$S \subseteq V$ of nodes in a
weighted hypergraph $H = (V,E)$,
a different notion of density is considered: $\mathfrak{r}(S) := \frac{w(E(S))}{w(S)}$,
where $E(S)$ is the collection of hyperedges that have
non-empty intersection with $S$.
The difference is that before we consider $E[S]$ that
is the collection of hyperedges totally contained inside~$S$.
Then, an alternative procedure can be defined by peeling off a
maximal least densest subset under $\mathfrak{r}$ and performing
recursion on the remaining graph.

This was the approach used by Goel et al.~\cite{DBLP:conf/soda/GoelKK12}
to consider the communication and streaming complexity of maximum
bipartite matching.
The same procedure has been used by Bernstein et al.~\cite{DBLP:journals/jacm/BernsteinHR19}
to analyze the replacement cost of online maximum bipartite matching maintenance.
Lee and Singla~\cite{DBLP:journals/talg/LeeS20}
also used this decomposition to investigate the
the batch arrival model of online bipartite matching.
Bansal and Cohen~\cite{DBLP:conf/waoa/BansalC21} have independently discovered this 
decomposition procedure to consider maximin allocation of hyperedge weights to incident nodes.

\noindent \textbf{Symmetric Nature of Density Decomposition.} At this point, the reader may ask how we know that the two procedures will give the same decomposition (but in reversed order).  The answer is that exactly the same quadratic program
has been used in~\cite{DBLP:conf/www/DanischCS17}
and~\cite{DBLP:journals/jacm/BernsteinHR19} to analyze the corresponding approaches.
This observation has also been made recently by Harb et al.~\cite[Theorem 13]{DBLP:conf/esa/HarbQC23}.  It is not too difficult to see that the least densest subset
using the modified density $\mathfrak{r}$ is equivalent
to the densest subset with the roles of vertices and hyperedges reversed.
Therefore, Theorem~\ref{th:folklore} can also be obtained indirectly through these observations.

\noindent \textbf{Connection between
Density Decomposition and Fisher Market Equilibrium.}
Even before this work in which we highlight
the local maximin condition,
density decomposition has already surreptitiously
appeared in algorithms computing Fisher market equilibrium.
Jain and Vazirani~\cite{DBLP:journals/geb/JainV10}
considered submodular utility allocation (SUA) markets,
which generalize linear Fisher markets.  Their
procedure~\cite[Algorithm 9]{DBLP:journals/geb/JainV10}
for computing an SUA market equilibrium 
can be viewed as a masqueraded variant of the density decomposition by Fujishige~\cite{fujishige1980lexicographically}.
Buchbinder et al.~\cite{DBLP:conf/soda/BuchbinderGHKS24}
used a variant of SUA market equilibrium to analyze 
online maintenance of matroid intersections,
where equilibrium prices are also obtained by the density
decomposition procedure.

In this work, we observe conversely that
algorithms for approximating market equilibrium
can also be directly used for approximating density decompositions.
For the simple special case in which the input graph contains a perfect fractional matching, Motwani et al.~\cite{DBLP:conf/approx/MotwaniPX06} showed that several randomized balls-into-bins
procedures attempting to maintain the locally maximin condition
can achieve $(1+\epsilon)$-approximation.


\noindent \textbf{Distributed Iterative Approximation Algorithms.}
Even though efficient exact algorithms
are known to achieve Goal~\ref{goal},
distributed iterative approximation algorithms have been considered for practical scenarios.

\noindent \emph{Density Decomposition.}
By applying the first-order iterative
approach of Frank-Wolfe to the quadratic program,
Danisch et al. \cite{DBLP:conf/www/DanischCS17} 
showed that $T$ iterations, each of which takes $O(|\mcal{F}|)$ work and may be parallelized, can achieve an absolute
error
of $O(\frac{1}{\sqrt{T}})$ on the density vector~$\rho_*$.
Harb et al.~\cite{DBLP:conf/esa/HarbQC23} showed that the method~$\textsc{Greedy++}$
can be viewed as a noisy variant of Frank-Wolfe
and hence has a similar convergence result.

By using a first-order method augmented with Nesterov momentum~\cite{Nesterov1983AMF}
such as accelerated FISTA~\cite{beck2009fast},
Elfarouk et al.~\cite{harb2022faster} readily concluded
that $T$ iterations can achieve an absolute error of $O(\frac{1}{T})$.
However, observe that if there is a vertex~$i$ such that $\rho_*(i) \leq \epsilon$,
an absolute error of $\epsilon$ may produce an estimate $\rho(i)$ that is arbitrarily small.

Inspired by estimates of $\rho_*$,
several notions of approximate density decompositions have been considered in previous works~\cite{tatti2015density,DBLP:conf/www/DanischCS17,harb2022faster}.
However, it is not clear if there 
is any meaningful interpretations of these approximate
decompositions in the context of
closest distribution refinements or market equilibrium.

\noindent \emph{Fisher Market Equilibrium.}
Market dynamics have been considered to
investigate how agents' mutual interaction can lead
to convergence to a market equilibrium.
The proportional response (PR, \emph{aka} tit-for-tat) protocol
has been used by Wu and Zhang~\cite{wu2007proportional}
to analyze bandwidth trading in a peer-to-peer network $G = (V,E)$
with vertex weights~$w$.  This can be seen as a special case of
our instance
$(\Izero, \Ione; \mcal{F}; w)$, where $\Izero$ and $\Ione$ are two copies
of $V$, and $i \in \Izero$ and $j \in \Ione$ are neighbors
in $\mcal{F}$ \emph{iff} the corresponding pair $\{i, j\} \in E$
are neighbors in $G$.
Interestingly, 
they consider a \emph{bottleneck} decomposition
of the graph~$G$ which can be recovered from our symmetric density decomposition.
They showed that the PR protocol converges to some equilibrium under this special model.

Subsequently, 
Zhang~\cite{DBLP:conf/icalp/Zhang09} showed that the PR protocol
can be applied in Fisher markets between buyers and sellers,
and analyzed the convergence rate to the market equilibrium
in terms of multiplicative error.
Birnbaum et al.~\cite{birnbaum2011distributed} later
improved the analysis of the PR protocol by demonstrating that it can be viewed as a generalized gradient descent algorithm.

For completeness, in Section~\ref{sec:mult_error}, we summarize
the analysis of how the PR protocol can estimate payloads from Goal~\ref{goal} with multiplicative error.

\ignore{
\color{blue}
\noindent \textbf{Fractional Perfect b-Matching. }Given a vertex weighted graph $G$, a fractional perfect b-matching of $G$ is a non-negative assignment of edge weights such that for each vertex $u$, the sum of weights on its incident edges is exactly equals to the vertex weight $b_u$. 
A refinement pair $\vec{\alpha} = (\alpha^{(0)}, \alpha^{(1)})$ is a generalized version of fractional perfect b-matching when $G$ is bipartite. Actually, when $\alpha^{(0)} = \alpha^{(1)}$, then either of them is a fractional perfect b-matching. 
The fractional perfect b-matching polytope of $G$ is the polytope of all fractional perfect b-matching of $G$. Certain fractional perfect $b$-matching polytopes have been studied and used in the contexts of combinatorial matrix classes (see, for example, the book by Brualdi \cite{brualdi2006combinatorial}), and combinatorial optimization (see, for example, the books by Korte and Vygen \cite{bhcombinatorial}, or Schrijver \cite{schrijver2003combinatorial}). Theorems concerning the properties of such polytope can be found in \cite{behrend2013fractional}. 
In some special cases, an approximate fractional perfect b-matching can be found. For example, when each vertex in $G$ has weight $1$, a near-perfect fractional matching where each right node has load within $1\pm \epsilon$ can be computed in $O(\frac{m\log n}{\epsilon})$ time.\cite{DBLP:conf/approx/MotwaniPX06}.

\color{black}
}

\subsection{Potential Impacts and Paper Organization}
\label{sec:future}

This work unifies diverse problems under a single framework, revealing connections and providing new approximation results. Future research directions are abundant. Examining the implications of various market models or equilibrium concepts on the densest subset problem and density decomposition opens the door for further exploration of the relationships between economic markets and combinatorial optimization problems. This line of inquiry could also lead to new insights into economics, particularly regarding market participants' behavior and the efficiency of different market mechanisms. 

\noindent \emph{Paper Organization.}  
\ignore{
To emphasize our technical contributions, we have included the innovative aspects of this work in the main body, while background materials and fresh insights on existing results are located in the appendix. However, to facilitate comprehension, we suggest the following reading order.
}  
In Section~\ref{sec:prelim}, we introduce the precise notation that will be used throughout the paper.  In Section~\ref{sec:sym_decomp}, we summarize the structural results regarding
the symmetric density decomposition, which is used as a technical tool to analyze other problems.
As mentioned in the introduction,
Section~\ref{sec:univeral_matching} investigates the
universal refinement matching problem which is the precursor
for tackling the more challenging universally
closest distribution refinements problem in Section~\ref{sec:dist_refine}.
In Section~\ref{sec:hockey-stick}, we first consider
the special case of hockey-stick divergences; by considering power functions,
we prove Theorem~\ref{th:main_universal} in Section~\ref{sec:dp_div};
in Section~\ref{sec:approx_power},
we present approximation results for the universal closest refinements problem
by defining a stretching operation on power functions.
In Section~\ref{sec:market}, 
we recap the connection
between the symmetric density decomposition and 
the symmetric Fisher market equilibrium;
moreover, we give an alternative characterization of the Fisher market equilibrium
in terms of the local maximin condition on sellers' goods allocation to buyers.  In Section~\ref{sec:mult_error},
we give details on the convergence rates of 
how the iterative proportional response protocol
can achieve an approximation with multiplicative error
for Goal~\ref{goal}.  
For completeness,
in Section~\ref{sec:additive_error},
we give a review of how absolute error
for Goal~\ref{goal} is achieved
using techniques in previous works.

\section{Preliminaries}
\label{sec:prelim}

We shall consider several problems, whose inputs
are in the form of vertex-weighted bipartite graphs.

\begin{definition}[Input Instance]
\label{defn:input}
An instance $(\Izero, \Ione; \mcal{F}; w)$ of the input consists of the following:
\begin{compactitem}
\item Two disjoint vertex sets 
$\Izero$ and $\Ione$.

For the side~$\side \in \B := \{0,1\}$, we use $\ob := \side \oplus 1$ to indicate the other side different from $\side$.

\item A collection $\mcal{F}$ of unordered pairs
of vertices that is interpreted as a
bipartite graph between vertex sets $\Izero$ and $\Ione$;
in other words, if $f = \{i,j\} \in \mcal{F}$,
then exactly one element of $f$ is in $\Izero$ and the
other is in $\Ione$. 
In this case, we also denote $i \sim_{\mcal{F}} j$
or just $i \sim j$.

A vertex is isolated if it has no neighbor.


\item Positive\footnote{Vertices with zero weight do not
play any role and may be excluded.} vertex weights $w: \Izero \cup \Ione \rightarrow \mathbb{R}_{>0}$;
for a vertex~$i$, we use $w_i$ or $w(i)$ to denote its weight.

For $\side \in \B$,
we also use $w^{(\side)} \in \R^{\Ib}$ 
to denote the weight vector for each side.
\end{compactitem}

This notation highlights the symmetry between the two sets
$\Izero$ and $\Ione$.
\end{definition}

\begin{remark}
In this work, we focus on the case that $\Izero$ and $\Ione$ are finite.
Even though it is possible to extend the results
to continuous and uncountable sets (which are relevant
when those sets are interpreted as continuous probability spaces),
the mathematical notations and tools involved are
outside the scope of the typical computer science community.
\end{remark}

\begin{definition}[Allocation Refinement and Payload]
\label{defn:refinement}
Given an input instance as in Definition~\ref{defn:input},
for side $\side \in \B$, an allocation refinement (or simply refinement) $\alpha^{(\side)} \in \R^{\mcal{F}}_{\geq 0}$
of $w^{(\side)} \in \R^{\Ib}$ can be interpreted as edge weights
on~$\mcal{F}$ such that for each non-isolated vertex $i \in \Ib$, 

$$w^{(\side)}(i) = \sum_{f \in \mcal{F}: i \in f} \alpha^{(\side)}(f).$$

In other words, if a vertex $i$ has at least one neighbor, it distributes its weight $w_i$
among its neighboring edges $\{i,j\}$ such that $i \sim j$.
These edge weights due to the refinement $\alpha^{(\side)}$
can be interpreted as received by vertices on the other side $\Iob$;
this induces a payload vector 
$p^{(\ob)} \in \R^{\Iob}$ defined by:

$$\forall j \in \Iob, 
p^{(\ob)}(j) := \sum_{f: j \in f} \alpha^{(\side)}(f).
$$

With respect to the weights of vertices in $\Iob$,
the (payload) density $\rho^{(\ob)} \in \R^{\Iob}$
is defined by:

$$\rho^{(\ob)}(j) := \frac{p^{(\ob)}(j)}{w(j)}.$$

When there is no risk of ambiguity,
the superscripts $(\side)$ and $(\ob)$ may be omitted.
\end{definition}

\begin{definition}[Proportional Response]
\label{def:PR-refinement}
Suppose for side~$\side \in \B$,
$\alpha^{(\side)}$ is a refinement of the vertex weights on $\Ib$.
Then, a proportional response to
$\alpha^{(\side)}$ is a refinement 
$\alpha^{(\ob)}$ of the vertex weights on the other side $\Iob$ satisfying
the following:

\begin{compactitem}
\item If $j \in \Iob$ receives a positive payload $p^{(\ob)}(j) > 0$
from $\alpha^{(\side)}$,
then for each $i \sim j$, 

$\alpha^{(\ob)}(ij) = \frac{\alpha^{(\side)}(ij)}{p^{(\ob)}(j)} \cdot w(j)
=\frac{\alpha^{(\side)}(ij)}{\rho^{(\ob)}(j)}$.

\item If a non-isolated vertex $j \in \Iob$
receives zero payload 
from $\alpha^{(\side)}$, then
in $\alpha^{(\ob)}$, vertex $j$ may distributes its weight $w(j)$ arbitrarily
among its incident edges in $\mcal{F}$.

\end{compactitem}

\end{definition}

\begin{definition}[Local Maximin Condition]
\label{defn:maximin}
Given an input instance as in Definition~\ref{defn:input},
for $\side \in \B$, suppose a refinement $\alpha^{(\side)} \in \R^{\mcal{F}}_+$
of $w^{(\side)} \in \R^{\Ib}$
induces a payload density vector $\rho^{(\ob)} \in \R^{\Iob}$
as in Definition~\ref{defn:refinement}.
Then,
$\alpha^{(\side)}$  is locally maximin if,

$$\forall i \in  \Ib, \quad
\alpha^{(\side)}(ij) > 0  \implies j \in \arg\min_{\ell \in \Iob: \ell \sim i} \rho^{(\ob)}(\ell).$$

In other words, 
the local maximin condition states that
each vertex $i \in \Ib$
attempts to send a positive weight to a neighbor~$j \sim i$ only if
$j$ achieves the minimum $\rho^{(\ob)}(\ell)$ among the neighbors~$\ell$ of $i$.
\end{definition}

\begin{remark}
\label{remark:prop_resp}
Observe that if $\alpha^{(\side)}$ satisfies the local maximin condition,
then every non-isolated vertex in $\Iob$ must receive a positive payload.
This means that its proportional response
$\alpha^{(\ob)}$ is uniquely determined.
\end{remark}

\begin{definition}[Notions of Error]
\label{defn:error}
Given a vector $\rho$, we consider two notions of error
for an approximation vector $\widehat{\rho}$.

\begin{compactitem}

\item \emph{Absolute error~$\epsilon$ with respect
to norm $\| \cdot \|$.}
We have: $\|\rho - \widehat{\rho}\| \leq \epsilon$.

A common example is the standard $\ell_2$-norm.

\item \emph{Multiplicative error~$\epsilon$.}
For all coordinates~$i$,
$|\rho(i) - \widehat{\rho}(i)| \leq \epsilon \cdot |\rho(i)|$.

\end{compactitem}

\end{definition}

\section{Folklore: Symmetric Density Decomposition}
\label{sec:sym_decomp}

Given an instance $(\Izero, \Ione; \mathcal{F}; w)$
as in Definition~\ref{defn:input},
instead of interpreting it as a bipartite graph between
$\Izero$ and $\Ione$,
one can interpret it as a hypergraph~$H$.
In this case, we have the vertex set $V = \Ione$
and the hyperedge set $E = \Izero$,
where  $v \in V$ and $e \in E$ are neighbors in $\mcal{F}$ \emph{iff} the hyperedge $e$ contains the vertex $v$.  Moreover, $w$ assigns weights
to both $V$ and $E$ in the hypergraph.

Under the hypergraph interpretation,
the density of a subset $S \subseteq V$
of vertices is $\rho(S) := \frac{w(H[S]))}{w(S)}$,
where $H[S]$ is the subset of edges that are totally contained in $S$.
Inspired by the densest subset problem,
density decomposition 
(which
is applicable to hypergraphs) has been studied.

With the hypergraph interpretation,
one can readily define the density decomposition
for an instance $(\Izero, \Ione; \mathcal{F}; w)$.
However, under this approach, the roles of $\Izero$ 
and $\Ione$ are not symmetric: one of $\Izero$ and $\Ione$
plays the role of vertices and the other plays the role of hyperedges.
We show that it is possible to reinterpret the density
decomposition such that the two sets 
$\Izero$ and $\Ione$ play symmetric roles.
This symmetric interpretation will be a useful tool
when we study the universally closest distribution refinements problem in Section~\ref{sec:dist_refine}.

\noindent \textbf{Notation: Ground and Auxiliary Sets.}
For side~$\side \in \B$,
the \emph{open neighborhood} 
of a non-empty subset $S \subseteq \Ib$
consists of vertices on the other side that
have at least one neighbor in $S$, i.e.,
$\mcal{F}(S) := \{j \in \Iob:
\exists i \in S, i \sim j\}$.
Its \emph{closed neighborhood}
is $\mcal{F}[S] := \{j \in \mcal{F}(S): \mcal{F}(\{j\}) \subseteq S\}$,
i.e., a vertex~$j$ is in $\mcal{F}[S]$
\emph{iff} it has at least one neighbor and all its neighbors
are contained in~$S$.
Its density is defined as $\rho^{(\side)}(S) := \frac{w(\mcal{F}[S])}{w(S)}$,
where the superscript~$\side$ may be omitted when the context is clear.
When we consider the densities of subsets
in $\Ib$, the set $\Ib$ plays the role of the
\emph{ground set} and the set $\Iob$ plays the role of
the \emph{auxiliary set}.
Observe that this definition is consistent
with the hypergraph $(V, H)$ interpretation when $\side = 1$ plays the role of
the vertex set~$V$.

\noindent \emph{Extension to Empty Set.}
In some of our applications, it makes sense to have isolated vertices
in the input instance.
For $\emptyset \subseteq \Ib$, 
we extend $\mcal{F}(\emptyset) = \mcal{F}[\emptyset] :=
\{j \in \Iob: \forall i \in \Ib, 
j \not\sim i\}$.
For the purpose of defining $\rho(\emptyset)$,
we use the convention that $\frac{0}{0} = 0$ and 
$\frac{x}{0} = + \infty$ for $x > 0$.

\begin{definition}[Densest Subset]
\label{defn:densest}
When $\Ib$ plays the role of the ground set,
a subset $S \subseteq \Ib$ is a densest subset,
if it attains $\max_{S \subseteq \Ib} \rho^{(\side)}(S)$.
\end{definition}

The following well-known fact characterizes a property of maximal densest subset
(see~\cite[Lemma 4.1]{DBLP:conf/wsdm/BalalauBCGS15} for a proof).

\begin{fact}
With respect to set inclusion,
the maximal densest subset in $\Ib$ is unique
and contains all densest subsets.
%
\end{fact}


\begin{definition}[Sub-Instance]
Given an instance $(\Izero, \Ione; \mcal{F}; w)$, the
subsets $S^{(0)} \subseteq \Izero$ and $S^{(1)} \subseteq \Ione$
naturally induces a sub-instance
$(S^{(0)}, S^{(1)}; \mcal{F}'; w')$,
where $\mcal{F}'$ and $w'$ are restrictions of $\mcal{F}$ and $w$, respectively,
to $S^{(0)} \cup S^{(1)}$.
\end{definition}


The following definition is essentially the same as the density decomposition of a hypergraph.


\begin{definition}[Density Decomposition with Ground Set]\label{def:decomposition}
	Given an instance $(\Izero, \Ione; \mcal{F}; w)$,
	suppose for $\side \in \B$, we consider $\Ib$ as the ground set.
	The density decomposition with $\Ib$ as the ground set is a sequence of pairs
	$\{(S^{(0)}_\ell, S^{(1)}_\ell)\}_{\ell \geq 1}$, together with 
	a density vector $\rho^{(\side)}_*: \Ib \rightarrow \R_+$, defined by the following iterative process.
	
	Initially, we set $G_0 := (\Izero, \Ione; \mcal{F}; w)$ to be the original given
	instance, and initialize $\ell = 1$.
	
	\begin{enumerate}
	
	\item If $G_{\ell-1}$ is an instance where the vertex sets on both sides are empty, the process stops.
	
	\item Otherwise, let $S^{(\side)}_\ell$ be the maximal densest subset in $G_{\ell-1}$ with side~$\side$ being the ground set, where $\rho_\ell$ is the corresponding density;
	let $S^{(\ob)}_\ell := \mcal{F}[S^{(\side)}_\ell]$ be its closed neighborhood
	in~$G_{\ell-1}$.
	
	For each $u \in S^{(\side)}_\ell$, define $\rho^{(\side)}_*(u) := \rho_\ell = \frac{w(S^{(\ob)}_\ell)}{w(S^{(\side)}_\ell)}$.
	
	\item Let $G_\ell$ be the sub-instance induced by removing
	$S^{(\side)}_\ell$ and	$S^{(\ob)}_\ell$ from $G_{\ell-1}$;
	note that edges in $\mcal{F}$ incident to removed vertices are also removed.
	
	Increment $\ell$ and go to the next iteration.
	
	\end{enumerate}
\end{definition} 

\begin{remark}[Isolated Vertices]
Observe that if each side~$\side$ has isolated vertices $U^{(\side)}$ that
have no neighbors, then they will also appear at the beginning and the end
of the decomposition sequence as $(U^{(0)}, \emptyset)$ and $(\emptyset, U^{(1)})$.

When $\Ib$ is the ground set,
the density $\rho^{(\side)}$ of vertices in $U^{(\side)}$ is zero.
In the rest of this section, we may assume that isolated vertices are removed from the instance.
\end{remark}

\noindent \textbf{Symmetry in Density Decomposition.}
In Definition~\ref{def:decomposition},
one side~$\side \in \B$ plays the role of the ground set.
One of our major discoveries is that if we switch the ground set to the
other side~$\ob$,
the same sequence $\{(S^{(0)}_\ell, S^{(1)}_\ell)\}_{\ell \geq 1}$
will be produced in the density decomposition, but in the reversed order.
To show this, we need a result that relates
the local maximin condition in Definition~\ref{defn:maximin}
with the density decomposition in Definition~\ref{def:decomposition}.
This is already known in a previous work~\cite{DBLP:conf/www/DanischCS17},
but we include the details in Section~\ref{sec:maximin_decomp}.

\subsection{Local Maximin Condition Implies Density Decomposition}
\label{sec:maximin_decomp}

The following fact has been proved in~\cite{DBLP:conf/www/DanischCS17}
when the vertices in the ground set have uniform weights.
For completeness, we are going to extend the approach
where both sides have general positive weights.

\begin{fact}[Local Maximin Condition Gives Density Decomposition]
\label{fact:maximin_decomp}
Given an instance $(\Izero, \Ione; \mcal{F}; w)$,
consider side~$\side \in \B$ as the ground set.
Suppose $\alpha^{(\ob)} \in \R^{\mcal{F}}$
is a refinement of the vertex weights $w^{(\ob)}$ on side~$\ob$,
which induces the payload density vector $\rho^{(\side)}$ on side~$\side$
as in Definition~\ref{defn:refinement}.

Suppose further that $\alpha^{(\ob)}$ satisfies 
the local maximin condition in Definition~\ref{defn:maximin}.
Then, $\rho^{(\side)}$ coincides exactly with the density vector $\rho^{(\side)}_*$
achieved by the density decomposition in Definition~\ref{def:decomposition}
with side~$\side$ being the ground set.

Moreover, 
if $(S^{(\ob)}_\ell, S^{(\side)}_\ell)$ is a pair in the density decomposition,
then $S^{(\ob)}_\ell$ is exactly the collection of vertices
that send positive payload to $S^{(\side)}_\ell$ in $\alpha^{(\ob)}$.
\end{fact}

\noindent \textbf{LP Relaxation of Densest Subset.}
For ease of illustration,
we will consider $\side=1$ as the ground set.
In terms of the hypergraph interpretation,
the reader can treat side~$\ob = 0$
as the hyperedges and side~$\side$ as vertices.
To simplify the notation,
we will omit some superscripts
and denote $\alpha = \alpha^{(0)}$ as the refinement of $w^{(0)}$
and $\rho = \rho^{(1)}$ as the density vector for $\Ione$.
An LP relaxation for the densest subset problem has
been considered~\cite{10.1007/3-540-44436-X_10}
and we rephrase it with $\Ione$ being the ground set
in an instance $G = (\Izero, \Ione; \mcal{F}; w)$ as follows.
It is known that the optimal objective value of $\LP(G)$ is exactly the 
maximum density of a subset.

\begin{align*}
	\LP(G): \quad \max \sum_{i \in \Izero} w(i) \cdot x_i\\
	\text{s.t.} \quad x_i & \leq y_j, & \forall i \in \Izero \sim j \in \Ione \\
	\sum_{j \in \Ione} w(j) \cdot y_j & = 1 & \\
	x_i, y_j & \geq 0, &\forall i \in \Izero, j \in \Ione
\end{align*}

The corresponding dual $\DP$ can be expressed as
follows, where the dual variables $\alpha$ naturally corresponds
to a refinement of $w^{(0)} \in \R^{\Izero}$.
Because $\DP$ is a minimization LP, in an optimal solution,
equality should hold for the first constraint.

\begin{align*}
	\DP(G): \quad \min \quad  r \\
	\text{s.t.} \quad \sum_{j \in \Ione: i\sim j} \alpha(ij) & \geq w(i), &\quad \forall i \in \Izero\\
	 w(j) \cdot r  & \geq p(j) := \sum_{i \in \Izero: i \sim j} \alpha(ij), &\quad \forall j \in \Ione \\
	\alpha(ij) & \geq 0, &\quad \forall i \in \Izero \sim j \in \Ione \\
	r & \in \R &
\end{align*}

\begin{lemma}[Maximin Condition Recovers Maximal Densest Subset]
\label{lemma:maximin_densest}
Suppose $\alpha \in \R^\mcal{F}$ is a refinement
of vertex weights $w^{(0)}$ in $\Izero$, which induces a payload density vector $\rho$
on the vertices~$\Ione$.
Suppose $S := \arg \max_{j \in \Ione} \rho(j)$ are the vertices
with the maximum payload density.  

In addition, if $\alpha$ satisfies
the local maximin condition in Definition~\ref{defn:maximin}, then $S$ is the maximal densest
subset in $\Ione$, whose density equals to the payload density $\rho(j)$ for $j \in S$.
\end{lemma}

\begin{proof}
We first show that $S$ is a densest subset.
Note that the refinement $\alpha$ of $w^{(0)}$ is a feasible
dual solution to $\DP(G)$. 
Because $\alpha$ satisfies the local maximin condition in Definition~\ref{defn:maximin},
it follows that any vertex~$j \in S = \arg \max_{j \in \Ione} \rho(j)$
can receive positive payload $\alpha(ij) > 0$ only from $i$ in the closed neighborhood $\mcal{F}[S]$.  Hence, it follows that for all $j \in S$, $\rho(j) = \rho_G(S) := \frac{w(\mcal{F}[S])}{w(S)}$.

Moreover, since $S$ are the vertices in $\Ione$ with 
the maximum payload density, it follows that we can set
the objective value $r = \rho_G(S)$ for the dual feasible solution
$\alpha$.  On the other hand, $\rho_G(S)$ is the density of the subset~$S \subseteq \Ione$,
which induces a feasible primal solution.
Since we have exhibited a primal-dual pair with the same objective value,
it follows that both solutions are optimal, i.e., $S$ is a densest subset.

We next show that $S$ is the maximal densest subset.
For the sake of contradiction, assume that
there exists a non-empty $B \subseteq \Ione$ disjoint from $S$
such that the density of $S \cup B$ in $G$ is also $\rho_G(S)$, i.e,
we have 

$\rho_G(S \cup B) = \frac{w(\mcal{F}[S]) + w(\mcal{F}[S \cup B] \setminus \mcal{F}[S])}{w(S) + w(B)} = \rho_G(S)$, which implies that
$\frac{w(\mcal{F}[S \cup B] \setminus \mcal{F}[S])}{w(B)} = \rho_G(S)$.

Consider the sub-instance $\widehat{G}$ obtained
by removing $S$ and $\mcal{F}[S]$ from $G$.
Then, we can conclude that the objective
value of $\LP(\widehat{G})$ is at least $\rho_G(S)$, because we have $\rho_{\widehat{G}}(B) = \rho_G(S)$.

Next, consider the dual $\DP(\widehat{G})$.
Because $\alpha$ satisfies the local maximin condition,
this means that for any $i \notin \mcal{F}[S]$ and $j \in S$,
$\alpha(ij) = 0$.  Therefore,
the natural restriction $\widehat{\alpha}$ of $\alpha$ to $\widehat{G}$
is a feasible dual solution to $\LP(\widehat{G})$.
However, observe that since all vertices $S$ with the maximum
payload density has been removed,
the corresponding objective value of the dual solution $\widehat{\alpha}$
will drop strictly below $\rho_G(S)$.  This violates the weak duality
between $\LP(\widehat{G})$ and $\DP(\widehat{G})$.

Hence, we conclude that $S$ is the maximal densest subset in $G$.
\end{proof}

\noindent \textbf{Proof of Fact~\ref{fact:maximin_decomp}.}
We can prove the result by induction on the
number of distinct values appearing in
the payload density vector $\rho$ on $\Ione$ induced
by the refinement $\alpha$ of the vertex weights $w^{(0)}$ on $\Izero$.
For the base case when all vertices~$j \in \Ione$ have the same
payload density $\rho(j)$, Lemma~\ref{lemma:maximin_densest}
implies that the decomposition consists of only one pair $(\Izero, \Ione)$.

For the inductive case, suppose $S = \arg \max_{j \in \Ione} \rho(j)$
as defined in Lemma~\ref{lemma:maximin_densest}.  Then,
the first pair in the decomposition is $(\mcal{F}[S], S)$.
As argued in the proof of Lemma~\ref{lemma:maximin_densest},
the result follows from the induction hypothesis
on the sub-instance $\widehat{G}$ by removing $(\mcal{F}[S], S)$ from the original instance.
The reason this works is because $\alpha$ satisfies
 the local maximin condition,
for $i \notin \mcal{F}[S] \sim j \in S$,
$\alpha(ij) = 0$.  Hence,
 the restriction $\widehat{\alpha}$
of $\alpha$ to the sub-instance $\widehat{G}$ is
also a refinement of the vertex weights on side~$0$ that
satisfies the local maximin condition.
\hfill \qed

Fact~\ref{fact:maximin_decomp} states the desirable
properties achieved by a locally maximin refinement.
It is easy to show that a locally maximin refinement exists.

\begin{fact}[Existence of Locally Maximin Refinement]
\label{fact:exist_maximin}
There exists a refinement $\alpha^{(0)}$
of the vertex weights in $\Izero$ that is locally maximin.
\end{fact}

\begin{proof}
We prove by induction on the number~$\ell$ of pairs
in the density decomposition of the instance~$G$.

Consider an optimal dual solution $\alpha^*$ to 
the dual $\DP(G)$.  By strong LP duality,
we know that the objective value of the dual
equals to the density $\frac{w(\mcal{F}[S])}{w(S)}$
of the maximal densest subset $S \subseteq \Ione$.

Hence, when we consider the restriction
of $\alpha^*$ to the edges between $\mcal{F}[S] \subseteq \Izero$
and $S \subseteq \Ione$,
it must be the case that the payload density
of every vertex in $S$ is the same;
otherwise, some vertex in $S$ will have payload
density strictly larger than $\frac{w(\mcal{F}[S])}{w(S)}$,
which makes the dual objective value also
strictly larger than $\frac{w(\mcal{F}[S])}{w(S)}$.
Therefore, we can use this restriction to define
a refinement $\alpha'$ for the vertex weights in $\mcal{F}[S]$.

For the case $\ell = 1$ where $\mcal{F}[S] = \Izero$ and $S = \Ione$,
the base case is completed.

For $\ell \geq 2$,
we consider the sub-instance $\widehat{G}$
by removing $\mcal{F}[S]$ and $S$ from the original instance~$G$,
and apply the induction hypothesis to $\widehat{G}$
and obtain a locally maximin refinement $\widehat{\alpha}$
of vertex weights in $\Izero \setminus \mcal{F}[S]$.

Observe that we can combine $\alpha'$ and $\widehat{\alpha}$
to get a refinement~$\alpha^{(0)}$ for the vertex weights in $\Izero$.
We check that $\alpha^{(0)}$ is locally maximin,
because the payload density of every vertex in $S$
is strictly larger than that of every vertex in $\Ione \setminus S$,
by the property of density decomposition.
\end{proof}

\subsection{Density Decomposition Is Symmetric}
\label{sec:intro_decomp}

Since we have shown in
Lemma~\ref{lemma:maximin_densest} that a refinement that
satisfies the local maximin condition is related
to the density decomposition in Definition~\ref{def:decomposition},
we will illustrate the symmetry of the density decomposition
by considering the proportional response to
a locally maximin refinement.

\begin{lemma}[Proportional Response to Locally Maximin Refinement]
\label{lem:local_maximin_sym}
Let $\alpha^{(0)}$ be a refinement
of vertex weights in $\Izero$
 that satisfies the local maximin condition. 
By Remark~\ref{remark:prop_resp}, 
suppose the refinement $\alpha^{(1)}$ on vertex weights in $\Ione$
is the unique proportional response to $\alpha^{(0)}$.
Moreover, for each side $\side \in \B$,
the refinement $\alpha^{(\side)}$ induces
the payload density vector $\rho^{(\ob)}$ on vertices on
the other side~$\ob$.
Then, the following holds:

\begin{compactitem}

\item For any $i \in \Izero$ and $j \in \Ione$ such that $i \sim j$ and $\alpha^{(0)}(ij) > 0$, their received payload densities due to
the refinements satisfy $\rho^{(0)}(i) = \frac{1}{\rho^{(1)}(j)}$.

\item  The refinement  $\alpha^{(1)}$ also satisfies the local maximin condition.

\item We get back $\alpha^{(0)}$ as the proportional response to $\alpha^{(1)}$.

\end{compactitem}
	
\end{lemma}

\begin{proof}
By the definition of proportional response,
for any $i \in \Izero$ and $j \in \Ione$,
$\alpha^{(0)}(ij) > 0$ \emph{iff} $\alpha^{(1)}(ij) > 0$.

We prove the first statement.
For $i \in \Izero$, denote $F_i := \{j \in \Ione: \alpha^{(0)}(ij) > 0\}$.
Since $\alpha^{(0)}$ satisfies the local maximin condition,
all $j \in F_i$ has the same payload density
$\rho^{(1)}(j)$.
Now, consider the payload received by $i$ due to the
proportional response $\alpha^{(1)}$:

$p^{(0)}(i) := \sum_{j \in F_i} \alpha^{(1)}(ij)
= \sum_{j \in F_i} \alpha^{(0)}(ij) \cdot \frac{w(j)}{p^{(1)}(j)}
= \frac{w(i)}{\rho^{(1)}(j)}$,

where in the last equality, we use that all $j \in F_i$
have the same payload density $\rho^{(1)}(j)$,
and $w(i) = \sum_{j \in F_i} \alpha^{(0)}(ij)$.
This implies that for any $j \in F_i$,

$\rho^{(0)}(i) := \frac{p^{(0)}(i) }{w(i)} = \frac{1}{\rho^{(1)}(j)}.$

We next prove the second statement.
Fix some $j \in \Ione$ and suppose $\alpha^{(1)}(ij) > 0$,
which is equivalent to $\alpha^{(0)}(ij) > 0$ and $j \in F_i$.  If some vertex $i' \in \Izero$
has payload density $\rho^{(0)}(i') > \rho^{(0)}(i)$,
then our first statement implies that
$\rho^{(0)}(i') \cdot \rho^{(1)}(j) > \rho^{(0)}(i)
\cdot \rho^{(1)}(j) = 1$,
which means that both $\alpha^{(0)}(i'j) = \alpha^{(1)}(i'j) = 0$.

Hence, it suffices to show that
if a vertex $\widehat{i} \in \Izero$
has payload density $\rho^{(0)}(\widehat{i}) < 
\rho^{(0)}(i)$, then $\widehat{i}$ and $j$ are not neighbors in $\mcal{F}$.

Consider some $\widehat{j} \in F_{\widehat{i}}$
that receives positive payload from $\widehat{i}$ in the 
refinement $\alpha^{(0)}$.
By our first statement,
we have the payload density $\rho^{(1)}(\widehat{j}) 
= \frac{1}{\rho^{(0)}(\widehat{i})}
> \frac{1}{\rho^{(0)}(i)} = \rho^{(1)}(j)$.
Since $\alpha^{(0)}$ satisfies the local maximin
condition, it follows that $\widehat{i}$ and $j$
cannot be neighbors in $\mcal{F}$.

Hence, $\alpha^{(1)}$ also satisfies the local maximin condition.

Finally, for $i \in \Izero$ and $j \in \Ione$,
the proportional response to $\alpha^{(1)}$
is $\widehat{\alpha}(ij) = \frac{\alpha^{(1)}(ij)}{\rho^{(0)}(i)}
= \frac{\alpha^{(0)}(ij)}{\rho^{(1)}(j) \cdot \rho^{(0)}(i)} =
\alpha^{(0)}(ij)$.
\end{proof}

Now, we are ready to prove that density decomposition is symmetric.

\begin{theorem}[Folklore: Symmetry of Density Decomposition]
\label{th:folklore}
Given an instance 
$(\Izero, \Ione; \mcal{F}; w)$,
suppose the sequence $\{(S^{(0)}_\ell, S^{(1)}_\ell)\}_{\ell \geq 1}$ of pairs
is the density decomposition using $\Ione$ as the ground set.

Then, the density decomposition using $\Izero$ as the ground set
produces the same sequence of pairs in reversed order.
\end{theorem}

\begin{proof}
By Facts~\ref{fact:maximin_decomp} and~\ref{fact:exist_maximin},
the density decomposition $\{(S^{(0)}_\ell, S^{(1)}_\ell)\}_{\ell \geq 1}$ using $\Ione$ as the ground set
is associated with a locally maximin refinement $\alpha^{(0)}$ of the vertex weights in $\Izero$.
	
We next apply Lemma~\ref{lem:local_maximin_sym} to the proportional response $\alpha^{(1)}$ to $\alpha^{(0)}$, and conclude that the refinement $\alpha^{(1)}$ of vertices weights in $\Ione$	is also locally maximin. 
By Fact~\ref{fact:maximin_decomp},
the payload density vector $\rho^{(0)} \in \R^{\Izero}$
induced by $\alpha^{(1)}$ coincides with
the density vector in the density decomposition
$\{(\widehat{S}^{(0)}_k, \widehat{S}^{(1)}_k)\}_{k \geq 1}$ with $\Izero$ as the ground set.

By Lemma~\ref{lem:local_maximin_sym},
for all $\ell$, for all $i \in S^{(0)}_\ell$
and for all $j \in S^{(1)}_\ell$,
$\rho^{(0)}(i) \cdot \rho^{(1)}(j) = 1$.

Because $\{(S^{(0)}_\ell, S^{(1)}_\ell)\}_{\ell \geq 1}$
is the density decomposition with $\Ione$ being the 
ground set,
this means that the densities $\rho^{(0)}$
are the same for vertices within the same 
$S^{(0)}_\ell$ and different for $\ell \neq \ell'$.
Moreover, since the density $\rho^{(1)}$
 of vertices in $S^{(1)}_\ell$ decreases as $\ell$ increases,
this implies that the sequence
$\{\widehat{S}^{(0)}_k\}_{k \geq 1}$
is the reverse of
$\{{S}^{(0)}_\ell\}_{\ell \geq 1}$.

Finally,
observe that proportional response
has the property that $\alpha^{(0)}(ij) > 0$
\emph{iff} $\alpha^{(1)}(ij) > 0$.
From Fact~\ref{fact:maximin_decomp},
$\widehat{S}^{(1)}_k$ is exactly
the collection of vertices that send
positive payload
to $\widehat{S}^{(0)}_k$ in $\alpha^{(1)}$.
Therefore,
$S^{(0)}_\ell = \widehat{S}^{(0)}_k$
\emph{iff} $S^{(1)}_\ell = \widehat{S}^{(1)}_k$.
This finishes the proof.
\end{proof}

\section{Universal Refinement Matching Problem}
\label{sec:univeral_matching}

As mentioned in Section~\ref{sec:intro},
an instance in Definition~\ref{defn:input}
can be interpreted as inputs in various problems.
As a precursor to the universally closest distribution
refinements problem studied in Section~\ref{sec:dist_refine},
we consider a matching problem.
Here, an instance 
is interpreted as the input of
the vertex-constrained fractional matching problem,
where the vertex constraints on each side
is scaled by the same multiplicative factor.

\begin{definition}[Vertex-Constrained Matching Instance]
\label{defn:matching}
Given an instance $G = (\Izero, \Ione; \mcal{F}; w)$
as in Definition~\ref{defn:input}
and two non-negative weights $\vec{c} = (c^{(0)}, c^{(1)}) \in \R^{\B}$, 
we use $G^{(\vec{c})}$ to denote the instance
of (fractional) matching on the bipartite graph
$(\Izero, \Ione; \mcal{F})$ with vertex capacity constraints.
For side~$\side$ and $i \in \Ib$,
the capacity constraint at vertex~$i$ is $c^{(\side)} \cdot w(i)$.
\end{definition}

\noindent \textbf{LP Interpretation.}
The maximum (fractional) matching problem instance $G^{(\vec{c})}$
is captured by the following linear program.

	\begin{align*}
	\LP(G^{(\vec{c})}) \quad \quad \quad	\max \quad \sum_{i \sim j} x_{ij}\\
		\text{s.t.} \quad \sum_{j':i \sim j'} x_{ij'} \leq c^{(0)} \cdot w(i), &\quad \forall i \in \Izero\\
		\sum_{i': i' \sim j} x_{i'j} \leq c^{(1)} \cdot w(j), &\quad \forall j \in \Ione\\
		x_{ij} \geq 0, &\quad \forall i \in \Izero \sim j \in \Ione
	\end{align*}

We will also consider the dual in our analysis:

	\begin{align*}
	\DP(G^{(\vec{c})}) \quad \quad \quad	 \quad \min \quad c^{(0)}\sum_{i \in \Izero} w(i) \cdot \lambda_i &+ c^{(1)} \sum_{j \in \Ione} w(j) \cdot \lambda_j \\
	\text{s.t.} \quad \lambda_i + \lambda_j \geq 1, &\quad \forall i \in \Izero \sim j \in \Ione\\
	\lambda_i \geq 0, &\quad \forall i \in \Izero \cup \Ione
	\end{align*}

\begin{definition}[Fractional Matching Induced by Refinement Pair]
Given an instance $G = (\Izero, \Ione; \mcal{F}; w)$
and a non-negative weight pair $\vec{c} = (c^{(0)}, c^{(1)})$, 
a refinement pair $\vec{\alpha} = (\alpha^{(0)}, \alpha^{(1)})$,
where $\alpha^{(\side)}$ is a refinement on vertex weights in $\Ib$,
induces a feasible fractional matching 
for $G^{(\vec{c})}$ that is denoted as follows:

$\vec{c} \bullet \vec{\alpha}(ij) := \min \{c^{(0)} \cdot \alpha^{(0)}(ij), 
c^{(1)} \cdot \alpha^{(1)}(ij)\}$,
for all $i \sim j$.
\end{definition}

\begin{remark}
The feasibility of the fractional matching $\vec{c} \bullet \vec{\alpha}$ to
the instance $G^{(\vec{c})}$ follows directly because
each $\alpha^{(\side)}$ is a refinement of the vertex weights in $\Ib$.
\end{remark}

\begin{definition}[Universal Refinement Pair for Matching]
Given an instance $G = (\Izero, \Ione; \mcal{F}; w)$
and $\tau \in [0, 1]$,
a refinement pair $\vec{\alpha}$
is $(1 - \tau)$-universal for matching in $G$,
if for any non-negative weight pair $\vec{c}$,
the solution $\vec{c} \bullet \vec{\alpha}$
achieves $(1 - \tau)$-approximation for the 
maximum matching problem on $G^{(\vec{c})}$.

For $\tau = 0$, we simply say $\vec{\alpha}$
is universal (for matching in $G$).
\end{definition}

\noindent \textbf{Intuition for Universal Refinement Pair.}
Besides being a technical tool for later sections,
the notion of a universal refinement pair is interesting
in terms of understanding the structure of vertex-constrained
bipartite matching.  Suppose it is known that the vertex weights
on each side are correlated and may vary but follow some fixed 
relative ratios among themselves.  The idea of a universal refinement
pair is that some pre-computation can be performed such that 
when the vertex weights on each side change proportionally,
the matching can be updated quickly to get a (nearly) optimal solution
 without computing from scratch.

Observe that the density decomposition in Definition~\ref{def:decomposition}
does not change if the vertex weights on one side are multiplied by
the same scalar.  Informally, the decomposition identifies
different levels of ``congestion'' when a matching is performed
between vertices on the two sides.  The following result formalizes
the connection between the density decomposition and universal refinement pairs.

\begin{lemma}[Local Maximin Achieves Universal for Matching]
\label{lem:matching-opt}
Suppose in an instance $G$, for side~$\side \in \B$,
the refinement $\alpha^{(\side)}$ on vertex weights in $\Ib$
is locally maximin, and the refinement $\alpha^{(\ob)}$ is the (unique)
proportional response to $\alpha^{(\side)}$.
Then, the refinement pair $(\alpha^{(\side)}, \alpha^{(\ob)})$
is universal for matching in $G$.

Moreover, suppose in the density decomposition, for side~$\side$,
$\rho^{(\side)}_*$ is the payload density vector
and $p^{(\side)}_*$ is the payload vector.
Then, for any $c = (c^{(\side)}, c^{(\ob)})$,
the weight of the maximum matching
in $G^{(\vec{c})}$ is:

$\sum_{j \in \Ib} \min\{\frac{c^{(\side)}}{\rho^{(\side)}_*(j)}, c^{(\ob)} \} \cdot p^{(\side)}_*(j).$
\end{lemma}

\begin{proof}
By Lemma~\ref{lem:local_maximin_sym},
$\alpha^{(\side)}$ and $\alpha^{(\ob)}$ are proportional responses
to each other and both are locally maximin.
Fix some $\vec{c} = (c^{(0)}, c^{(1)})$
and we show that $\vec{c} \bullet \vec{\alpha}$
is an optimal solution to $G^{(\vec{c})}$.

Let $\{(S^{(0)}_\ell, S^{(1)}_\ell\}_{\ell \geq 1}$
be the density decomposition with $\Ione$ being the ground set.
(Recall that the density decomposition is symmetric between $\Izero$ and
$\Ione$, where the forward or backward order of the sequence
depends on which side is the ground set.)

From Fact~\ref{fact:maximin_decomp},
recall that $\alpha^{(0)}(ij)$ and
$\alpha^{(1)}(ij)$ are non-zero only if there exists 
some $\ell$ such that $(i,j) \in \mcal{F}_\ell :=
\mcal{F} \cap (S^{(0)}_\ell \times S^{(1)}_\ell)$.

Let $L^{(1)} := \{\ell: c^{(1)} \cdot w(S^{(1)}_\ell) \leq 
c^{(0)} \cdot w(S^{(0)}_\ell)\}$
and $L^{(0)} := 
\{\ell: c^{(1)} \cdot w(S^{(1)}_\ell) >
c^{(0)} \cdot w(S^{(0)}_\ell)\}$.

Observe that 
for $i \in S^{(1)}_\ell$,
$\rho^{(1)}_*(i) = \frac{w(S^{(0)}_\ell)}{w(S^{(1)}_\ell)}$.
Therefore, for all $\ell \in L^{(1)}$ and $\ell' \in L^{(0)}$,
$\ell < \ell'$.

Moreover, for all $\ell < \ell'$,
for all $i \in S^{(0)}_\ell$ and $j \in S^{(1)}_{\ell'}$,
$i$ and $j$ are not neighbors in $\mcal{F}$,
from the definition of the density decomposition.

Hence, we can construct a feasible dual solution $\lambda$
to $\DP(G^{(\vec{c})})$ as follows.
\begin{compactitem}
\item For $\ell \in L^{(1)}$: for $i \in S^{(0)}_\ell$, $\lambda_i = 0$;
for $j \in S^{(1)}_\ell$, $\lambda_j = 1$.

\item For $\ell \in L^{(0)}$: for $i \in S^{(0)}_\ell$, $\lambda_i = 1$;
for $j \in S^{(1)}_\ell$, $\lambda_j = 0$.
\end{compactitem}

The dual feasibility of $\lambda$ holds because
if $i \in \Izero$ and $j \in \Ione$ such
that $\lambda_i = \lambda_j = 0$,
then it must be the case that $i$ and $j$ are not neighbors in $\mcal{F}$.

Next, we check that the dual solution $\lambda$
has the same objective value as the primal
solution $\vec{c} \bullet \vec{\alpha}$ to $\LP(G^{\vec{c}})$.

Observe that for $\ell \in L^{(1)}$,
for $i \in S^{(0)}_\ell \sim j \in S^{(1)}_\ell$,
we have 

$c^{(1)} \cdot \alpha^{(1)}(ij) = c^{(1)} \cdot \alpha^{(0)}(ij) \cdot \frac{w(S^{(1)}_\ell)}{w(S^{(0)}_\ell)} \leq c^{(0)} \cdot \alpha^{(0)}(ij)$.

Hence, $c^{(1)}\sum_{j \in S^{(1)}_\ell}  w(j) \cdot \lambda_j 
= c^{(1)} \sum_{(i,j) \in \mcal{F}_\ell}  \alpha^{(1)}(ij)
=  \sum_{(i,j) \in \mcal{F}_\ell}  \vec{c} \bullet \vec{\alpha}(ij)$,
where the first equality holds because 
$\alpha^{(1)}$ restricted to $\mcal{F}_\ell$ is a refinement
of vertex weights in $S^{(1)}_\ell$.

Similarly, 
for $\ell \in L^{(0)}$,
we have 
$c^{(0)}\sum_{i \in S^{(0)}_\ell}  w(i) \cdot \lambda_i
= \sum_{(i,j) \in \mcal{F}_\ell}  \vec{c} \bullet \vec{\alpha}(ij)$.

Therefore, the feasible primal solution $\vec{c} \bullet \vec{\alpha}$
and the feasible dual solution $\lambda$ have the same objective value,
which implies that both are optimal.

Finally, we compute the weight of the maximum matching in $G^{(\vec{c})}$.
For simpler notation, it is understood that in each sum below,
we include only pairs $(i,j)$ such that both
$\alpha^{(0)}(ij)$ and $\alpha^{(1)}(ij)$ are non-zero.

\begin{align*}
 & \sum_{(i,j) \in (\Iob \times \Ib) \cap \mcal{F}: \alpha(ij) > 0} \min \{c^{(\ob)} \cdot \alpha^{(\ob)}(ij), 
c^{(\side)} \cdot \alpha^{(\side)}(ij)\} \\
= & \sum_{(i,j)}  \min \{c^{(\ob)}, 
c^{(\side)} \cdot \frac{\alpha^{(\side)}(ij)}{\alpha^{(\ob)}(ij)}\} \cdot
  \alpha^{(\ob)}(ij) \\
= & \sum_{(i,j)}  \min \{c^{(\ob)}, 
c^{(\side)} \cdot \frac{1}{\rho^{(\side)}_*(j)} \} \cdot
  \alpha^{(\ob)}(ij) \\
= & \sum_{j \in \Ib} \min \{c^{(\ob)}, 
 \frac{c^{(\side)}}{\rho^{(\side)}_*(j)} \} \sum_{i: i \sim j} \alpha^{(\ob)}(ij) \\
= & \sum_{j \in \Ib} \min\{\frac{c^{(\side)}}{\rho^{(\side)}_*(j)}, c^{(\ob)} \} \cdot p^{(\side)}_*(j).
\end{align*}
\end{proof}

\subsection{Approximate Universal Matching}

Although locally maximin refinements (on one side~$\side$) may not be unique, Fact~\ref{fact:maximin_decomp} states that all such refinements correspond to
the same payload density vector (on the other side~$\ob$).
Hence, it is natural to define the notion
of approximation for refinements based on the
induced payload density vector.

\begin{definition}[$\tau$-Approximate Refinement]
\label{defn:approx_refine}
For $\tau \geq 0$,
a refinement $\alpha^{(\side)}$ on vertex weights in $\Ib$
achieves $\tau$-multiplicative error
if the induced payload density vector $\rho^{(\ob)}$
achieves $\tau$-multiplicative error
with respect to the density vector $\rho^{(\ob)}_*$
in the density decomposition.

In other words, for each $i \in \Iob$,
we have $|{\rho}^{(\ob)}(i) - \rho^{(\ob)}_*(i)| \leq \tau \cdot \rho^{(\ob)}_*(i)$.
\end{definition}

We now demonstrate how to translate the density vector approximation into approximations for the universal refinement matching problem. 

\begin{theorem}[Approximate Refinements Give Approximate
Universal Matching]\label{th:approx-flow}
Suppose given an instance $G$, for $0 \leq \tau < 1$,
for side $\ob \in \B$,
$\alpha^{(\ob)}$ is a refinement
on vertex weights in $\Iob$ that achieves $\tau$-multiplicative error,
and $\alpha^{(\side)}$ is the proportional response
to $\alpha^{(\ob)}$.
Then, the pair $\vec{\alpha} = (\alpha^{(\ob)}, \alpha^{(\side)})$
is $\frac{1-\tau}{1+ \tau}$-universal for matching in $G$.

\end{theorem}

\begin{proof}
Without loss of generality,
we prove the result for $\ob = 0$.
Fix positive weights $\vec{c} = (c^{(0)}, c^{(1)})$,
and we will show
that $\vec{c} \bullet \vec{\alpha}$
achieves $\frac{1-\tau}{1+ \tau}$-approximation
for the maximum fractional matching instance $G^{(\vec{c})}$.

Since $\alpha^{(0)}$ achieves $\tau$-multiplicative error,
the induced payload density $\rho^{(1)}$ satisfies,
for all $j \in \Ione$,

\begin{equation} \label{eq:density}
(1 - \tau) \cdot \rho^{(1)}_*(j) \leq
\rho^{(1)}(j) \leq  (1 + \tau) \cdot \rho^{(1)}_*(j).
\end{equation}

Equivalently, the induced payload $p^{(1)}(j) = w(j) \cdot \rho^{(1)}(j)$ also satisfies:

\begin{equation}  \label{eq:payload}
(1 - \tau) \cdot p^{(1)}_*(j) \leq
p^{(1)}(j) \leq  (1 + \tau) \cdot p^{(1)}_*(j).
\end{equation}

Moreover, because $\tau < 1$,
the proportional response $\alpha^{(1)}$ to $\alpha^{(0)}$
is unique, and we have for $i \in \Izero$,
$\alpha^{(1)}(ij) = \frac{\alpha^{(0)}(ij)}{\rho^{(1)}(j)}$.
By considering only pairs such that
$\alpha^{(0)}(ij) > 0$, the weight of the solution
$\vec{c} \bullet \vec{\alpha}$ is:

\begin{align*}
 & \sum_{(i,j) \in (\Izero \times \Ione) \cap \mcal{F}: \alpha^{(0)}(ij) > 0} \min \{c^{(0)} \cdot \alpha^{(0)}(ij), 
c^{(1)} \cdot \alpha^{(1)}(ij)\} \\
= & \sum_{(i,j)}  \min \{c^{(0)}, 
c^{(1)} \cdot \frac{\alpha^{(1)}(ij)}{\alpha^{(0)}(ij)}\} \cdot
  \alpha^{(0)}(ij) \\
= & \sum_{(i,j)}  \min \{c^{(0)}, 
c^{(1)} \cdot \frac{1}{\rho^{(1)}(j)} \} \cdot
  \alpha^{(0)}(ij) \\
\geq & \sum_{j \in \Ione} \min \{c^{(0)}, 
 \frac{c^{(1)}}{(1 + \tau) \cdot \rho^{(1)}_*(j)} \} \cdot \sum_{i \in \Izero: i \sim j} \alpha^{(0)}(ij)  & (\textrm{upper bound in (\ref{eq:density})}) \\
\geq & \frac{1}{1 + \tau} \sum_{j \in \Ione} \min\{\frac{c^{(1)}}{\rho^{(1)}_*(j)}, c^{(0)} \} \cdot p^{(1)}(j) \\
\geq & \frac{1-\tau}{1 + \tau} \sum_{j \in \Ione} \min\{\frac{c^{(1)}}{\rho^{(1)}_*(j)}, c^{(0)} \} \cdot p^{(1)}_*(j) & \textrm{(lower bound in (\ref{eq:payload}))} \\
= & \frac{1-\tau}{1 + \tau} \cdot \LP(G^{(\vec{c})}),
\end{align*}

where the last equality comes from Lemma~\ref{lem:matching-opt}.

\end{proof}

\subsection{Negative Result: Approximation Non-Transferable to Proportional Response}

Contrary to Lemma~\ref{th:approx-flow},
we show that given a refinement $\alpha^{(\ob)}$
to vertex weights on side~$\ob$ that achieves a certain multiplicative error,
there is not necessarily any guarantee on its proportional response.

\begin{lemma}[Approximation Guarantees Non-Transferable to Proportional Response]
For any $\tau \in (0, \frac{1}{2})$,
there exists an instance $G$
such that there exists a refinement $\alpha^{(0)}$ 
for vertex weights on $\Izero$
achieving $\tau$-multiplicative error,
but its proportional response $\alpha^{(1)}$ 
has a multiplicative error of at least $1 - \tau$.
\end{lemma}

\begin{proof}

Our counterexample is given in Figure~\ref{fig:counterexample}.
Observe that there are 3 vertices on each side.  Moreover,
the only dependence on $\tau$ are the weights of $i_2 \in \Izero$
and $j_3 \in \Ione$.

The density decomposition is: $(\{i_1\}, \{j_1, j_2\})$ and $(\{i_2, i_3\}, \{j_3\})$.
Moreover, for vertices in $\Ione$,
the correct payload densities
are $\rho^{(1)}_*(j_1) = \rho^{(1)}_*(j_2) = 1$
and $\rho^{(1)}_*(j_3) = \tau(1 + \tau)$.

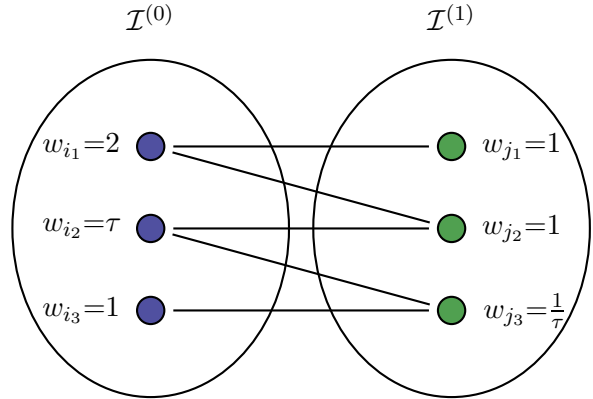
\begin{figure}[H]
	\begin{center}
		\definecolor{myblue}{RGB}{80,80,160}
		\definecolor{mygreen}{RGB}{80,160,80}
		
		\begin{tikzpicture}[thick,
			every node/.style={draw,circle},
			fsnode/.style={fill=myblue},
			ssnode/.style={fill=mygreen},
			every fit/.style={ellipse,draw,inner sep=-2pt,text width=2cm},
			->,shorten >= 3pt,shorten <= 3pt
			]
			
			\begin{scope}[xshift=0cm,yshift=0.5cm,start chain=going below,node distance=7mm]
				
				\node[fsnode, on chain] (u1) [label=left: $w_{i_1}{=}2$] {};
				\node[fsnode, on chain] (u2) [label=left: $w_{i_2}{=}\tau$] {};
				\node[fsnode, on chain] (u3) [label=left: $w_{i_3}{=}1$] {};

			\end{scope}
			\begin{scope}[xshift=4cm,yshift=0.5cm,start chain=going below,node distance=7mm]
				
				\node[ssnode,on chain] (v1) [label=right: $w_{j_1}{=}1$] {};
				\node[ssnode,on chain] (v2) [label=right: $w_{j_2}{=}1$] {};
				\node[ssnode,on chain] (v3) [label=right: $w_{j_3}{=}\frac{1}{\tau}$] {};
				
			\end{scope}
			
			\node [draw, circle, fit=(u1) (u2) (u3), inner sep=3mm] (circle) [label=above: $\Izero$]{};
			\node [draw, circle, fit=(v1) (v2) (v3), inner sep=3mm] (circle) [label=above: $\Ione$]{};
			
			\draw[-] (u1) -- (v1);
			\draw[-] (u1) -- (v2);
			\draw[-] (u2) -- (v2);
			\draw[-] (u2) -- (v3);
			\draw[-] (u3) -- (v3);
			
		\end{tikzpicture}
	\end{center}
	\caption{Counterexample}
	\label{fig:counterexample}
\end{figure}

Now, let us consider
the following refinement $\alpha^{(0)}$ on vertex weights in $\Izero$:

$\alpha^{(0)}(i_1 j_1) = \alpha^{(0)}(i_1 j_2) = 1$;

$\alpha^{(0)}(i_2 j_2) = \tau$ and $\alpha^{(0)}(i_2 j_3) = 0$;

$\alpha^{(0)}(i_3 j_3) = 1$.

The intuition is that the mistake is made by vertex~$i_2$.
Instead of sending its weight to $j_3$, its weight
is sent to $j_2$.

However, the resulting payload density on $\Ione$ is only affected slightly:

$\rho^{(1)}(j_1) = 1$, $\rho^{(1)}(j_2) = 1 + \tau$
and $\rho^{(1)}(j_3) = \tau$.
Hence, the density vector $\rho^{(1)}$
achieves $\tau$-multiplicative error with respect to $\rho^{(1)}_*$.

Consider the payload density for $i_2$ due to the proportional
response~$\alpha^{(1)}$ to $\alpha^{(0)}$.
The payload received by $i_2$ is: $1 \cdot \frac{\tau}{1 + \tau} = \frac{\tau}{1+\tau}$.  Hence,
the payload density from the proportional response
is: $\rho^{(0)}(i_2) = \frac{1}{1+\tau}$.

However,
in the density decomposition,
the correct payload density
should be $\rho^{(0)}_*(i_2)  = \frac{1}{\tau(1 + \tau)}$.

Therefore,
the multiplicative error is at least: $1 - \frac{\rho^{(0)}(i_2)}{\rho^{(0)}_*(i_2)} = 1 - \tau$.

\end{proof}

\section{Universally Closest Distribution Refinements Problem}
\label{sec:dist_refine}

In this section,
we show that when an instance is interpreted
as two distributions on two sample spaces imposed
with a neighboring relation,
a refinement pair satisfying Goal~\ref{goal}
is an optimal solution to the
\emph{universally closest distribution
refinements problem}.

\noindent \textbf{Distribution Instance.}
In Definition~\ref{defn:input},
an input is a \emph{distribution instance},
if, in addition, we require that
for each side~$\side \in \B$,
the vertex weights satisfies $\sum_{i \in \Ib} w^{(\side)}(i) = 1$,
i.e., $w^{(\side)}$ is a probability distribution on $\Ib$.

Observe that given a distribution instance
$G = (\Izero, \Ione; \mcal{F}; w)$,
for each side~$\side \in \B$,
a refinement $\alpha^{(\side)}$ of the vertex weights $w^{(\side)}$
is a probability distribution on $\mcal{F}$.
A high-level problem is to find a refinement pair
$\vec{\alpha} = (\alpha^{(0)}, \alpha^{(1)})$
such that the distributions
$\alpha^{(0)}$ and $\alpha^{(1)}$ on $\mathcal{F}$ are as ``close''
to each other as possible.

\noindent \textbf{Divergence Notion.} 
Given two distributions $P$ and $Q$ on the same sample space $\Omega$,
a notion of \emph{divergence} $\mathsf{D}(P \| Q)$
measures how close two distributions are,
where a smaller value typically means that the two distributions are closer.
An example is the total variation distance
$\mathsf{D}_{\textrm{TV}}(P \| Q) := \sup_{E \subseteq \Omega} |P(E) - Q(E)|$.
However, note that in general, a divergence needs not be symmetric,
i.e., it is possible that $\mathsf{D}(P \| Q) \neq \mathsf{D}(Q \| P)$.  An example of a non-symmetric divergence notion
is $\mathsf{D}_{\textnormal{KL}}(P \| Q) = \sum_{\omega \in \Omega} p_\omega \ln \frac{p_\omega}{q_\omega}$.

\begin{definition}[Closest Refinement Pair]
\label{defn:closest_refine}
Given a distribution instance $G = (\Izero, \Ione; \mcal{F}; w)$
and a divergence notion $\mathsf{D}$,
a refinement pair $\vec{\alpha} = (\alpha^{(0)}, \alpha^{(1)})$
(as in Definition~\ref{defn:refinement})
is closest (with respect to $\mathsf{D}$)
if, the pair attains the minimum $\mathsf{D}(\alpha^{(0)}||\alpha^{(1)})$ among all refinement pairs, whose value we denote as 
$\mathsf{D}^*(G)$.
%
\end{definition}

\begin{remark}
Since a divergence is typically continuous
and the collection of refinement pairs forms a compact set,
the minimum is attainable in such a case.

However, note that since $\mathsf{D}$ is not necessarily symmetric,
from Definition~\ref{defn:closest_refine} alone,
it is not clear if a pair 
minimizing $\mathsf{D}(\alpha^{(0)}||\alpha^{(1)})$ would
also minimize $\mathsf{D}(\alpha^{(1)}||\alpha^{(0)})$.
\end{remark}

\begin{definition}[Universally Closest Refinement Pair]
\label{defn:universal_refine}
Given a distribution instance $G = (\Izero, \Ione; \mcal{F}; w)$
and a collection $\mcal{D}$ of divergence notions,
a refinement pair $\vec{\alpha} = (\alpha^{(0)}, \alpha^{(1)})$
is universally closest (with respect to $\mathcal{D}$)
if, for all $\mathsf{D} \in \mcal{D}$, 
$\mathsf{D}(\alpha^{(0)}||\alpha^{(1)}) = \mathsf{D}^*(G)$.
\end{definition}

\subsection{Class of Hockey-Stick Divergences}
\label{sec:hockey-stick}

We first consider the class of hockey-stick divergences.

\begin{definition}[Hockey-Stick Divergence]
Given distributions $P$ and $Q$ on the same
sample space~$\Omega$ and $\gamma \geq 0$,
the (reversed\footnote{In the literature,
the hockey-stick divergence $\widetilde{\mathsf{D}}_\gamma$ is defined
with the roles of $P$ and $Q$ interchanged:
$\widetilde{\mathsf{D}}_\gamma(P \| Q) = 
\mathsf{D}_{\gamma}(Q \| P)$.  We change the convention to make
the notation consistent when we consider power functions later.}) hockey-stick divergence
is defined as 

$\mathsf{D}_{\gamma}(P \| Q) := \sup_{S \subseteq \Omega} Q(S) - \gamma \cdot P(S).$
\end{definition}

\begin{remark}
The hockey-stick divergence is related to the
well-known $(\epsilon, \delta)$-differential privacy inequality.
Specifically, $\mathsf{D}_{e^\epsilon}(P \| Q) \leq \delta$ \emph{iff}
for all subsets $S \subseteq \Omega$,

$Q(S) \leq e^\epsilon \cdot P(S) + \delta.$

To express the other direction of the inequality, i.e., 
$\mathsf{D}_{e^\epsilon}(Q\| P) \leq \delta$,
we could consider another parameter $\widehat{\gamma} \geq 0$
and keep the order of arguments in $\mathsf{D}_{\widehat{\gamma}}(P \|Q )$.
Specifically,
by considering the complement of events~$S$ in the above inequality,
we can deduce that:
$\mathsf{D}_{e^\epsilon}(Q \| P) \leq \delta$
\emph{iff} 
$\mathsf{D}_{e^{-\epsilon}}(P \| Q) \leq 1 - e^{-\epsilon} (1 - \delta)$.
\end{remark}

We consider universal closest distribution refinements
with respect to the class of hockey-stick divergences $\mcal{D}_{\textrm{HS}} := \{\mathsf{D}_\gamma: \gamma \geq 0\}$.

The following lemma shows that the hockey-stick divergence
is closely related to the matching instance defined
in Definition~\ref{defn:matching}.

\begin{lemma}[Hockey-Stick Divergence and Matching]
\label{lemma:hs_matching}
Suppose $G = (\Izero, \Ione; \mcal{F}; w)$
is a distribution instance
and 
$\vec{\alpha} = (\alpha^{(0)}, \alpha^{(1)})$
is a refinement pair.
For $\gamma \geq 0$, consider
the pair of weights $\vec{c} := (c^{(0)} = \gamma, c^{(1)} = 1)$.

Then, the weight of the fractional matching 
$\vec{c} \bullet \vec{\alpha}$ in the 
matching instance $G^{(\vec{c})}$ is exactly $1 - \mathsf{D}_\gamma(\alpha^{(0)} \| \alpha^{(1)})$.

In order words,
finding a closest refinement pair for the divergence $\msf{D}_\gamma$
is equivalent to finding a maximum matching in $G^{(\vec{c})}$.

\end{lemma}

\begin{proof}
Let $S := \{f \in \mcal{F}: \alpha^{(1)}(f) \geq  \gamma \cdot 
\alpha^{(0)}(f)\}$.

Observe that
$\mathsf{D}_\gamma(\alpha^{(0)} \| \alpha^{(1)})
= \sum_{f \in S} (\alpha^{(1)}(f) - \gamma \cdot 
\alpha^{(0)}(f) )$.

Finally, the weight of the matching $\vec{c} \bullet \vec{\alpha}$ is:

\begin{align*}
 & \sum_{f \in \mcal{F}} \min\{\gamma \cdot \alpha^{(0)}(f),
\alpha^{(1)}(f)\} 
=  \sum_{f \in S} \gamma \cdot \alpha^{(0)}(f)
+ \sum_{f \in \mcal{F} \setminus S} \alpha^{(1)}(f) \\
= & 1 - \sum_{f \in S} (\alpha^{(1)}(f) - \gamma \cdot 
\alpha^{(0)}(f) ) = 1 - \mathsf{D}_\gamma(\alpha^{(0)} \| \alpha^{(1)}),
\end{align*}

as required.
\end{proof}

\begin{theorem}[Locally Maximin Pair is Closest for 
$\mcal{D}_{\textnormal{HS}}$]
\label{th:local_hs}
Suppose in a distribution instance $G$, 
$\vec{\alpha} = (\alpha^{(0)}, \alpha^{(1)})$ is
a distribution refinement pair.
Then, we have the following.

\begin{compactitem}

\item Suppose the pair $\vec{\alpha}$
is locally maximin and 
and proportional response to each other.
Then,
$\vec{\alpha}$ is universally closest
with respect to the class 
$\mcal{D}_{\textnormal{HS}}$ of hockey-stick divergences.

\item Suppose for some $\side \in \B$ and $0 \leq \tau < 1$,
${\alpha}^{(\side)}$ achieves $\tau$-multiplicative 
error and ${\alpha}^{(\ob)}$ is the proportional response to 
${\alpha}^{(\side)}$.

Then, for any $\mathsf{D} \in \mcal{D}_{\textnormal{HS}}$,
$\mathsf{D}(\alpha^{(0)} \| \alpha^{(1)})
\leq  \frac{1-\tau}{1 + \tau} \cdot \mathsf{D}^*(G) + \frac{2 \tau}{1+\tau}$.
\end{compactitem}
\end{theorem}

\begin{proof}
The first statement is a direct consequence of Lemmas~\ref{lem:matching-opt} and~\ref{lemma:hs_matching}.
The second statement follows from Lemma~\ref{th:approx-flow};
specifically, we have
$(1 - \msf{D}) \geq (\frac{1-\tau}{1 + \tau}) \cdot (1 - \msf{D}^*)$,
which is equivalent to the required result after rearranging.
\end{proof}

\subsection{Class of Data Processing Divergences}
\label{sec:dp_div}

Extending Theorem~\ref{th:local_hs},
we show that a locally maximin pair is actually the closest
with respect to a wide class of divergences.
In fact, for any reasonable usage of divergence,
it is intuitive that taking a partial view of 
a pair of objects should not increase the divergence between them.
This is formally captured by the \emph{data processing inequality}.

\begin{definition}[Data Processing Inequality]
\label{defn:dpi}
	A divergence notion $\mathsf{D}$ satisfies the data processing inequality if given any two 
	distributions $P$ and $Q$ on the sample space~$\Omega$
	and any function $\varphi: \Omega \to \widehat{\Omega}$,
	the induced distributions $\varphi(P)$ and $\varphi(Q)$ on $\widehat{\Omega}$ satisfy
	the monotone property:

	\[
	\mathsf{D}(\varphi(P) \| \varphi(Q)) \leq \mathsf{D}(P \| Q).
	\]
	
	We use $\mcal{D}_{\textnormal{DPI}}$ to denote
	the class of divergences satisfying the data processing inequality.
\end{definition}

It is known that all divergences in 
$\mcal{D}_{\textnormal{DPI}}$ between two distributions $P$ and $Q$ can be recovered
from the power function between them.
Recall that in this work, we focus on finite sample spaces;
see Figure~\ref{fig:power-func} for an example.
It is known that power functions have
a fractional knapsack interpretation~\cite{kadane1968discrete}.

\begin{definition}[Power Function]
\label{defn:power}
Suppose $P$ and $Q$ are distributions on the same sample space~$\Omega$.

\begin{compactitem}
\item 
For a discrete sample space~$\Omega$, the power function $\Pow(P \| Q): [0,1] \rightarrow [0,1]$ can be
defined in terms of the \emph{fractional knapsack problem}.

Given a collection $\Omega$ of items,
suppose $\omega \in \Omega$ has weight $P(\omega)$ and value $Q(\omega)$.
Then, given $x \in [0,1]$,
$\Pow(P \| Q)(x)$ is the maximum value attained with weight capacity
constraint~$x$, where items may be taken fractionally.

\item For a continuous sample space~$\Omega$,
given $x \in [0,1]$, $\Pow(P \| Q)(x)$ is the supremum
of $Q(S)$ over measurable subsets $S \subseteq \Omega$
satisfying $P(S) = x$.

\end{compactitem}
\end{definition}

\begin{figure}[H]
	\centering
	\begin{subfigure}[b]{0.45\textwidth}
	\includegraphics[width = \textwidth]{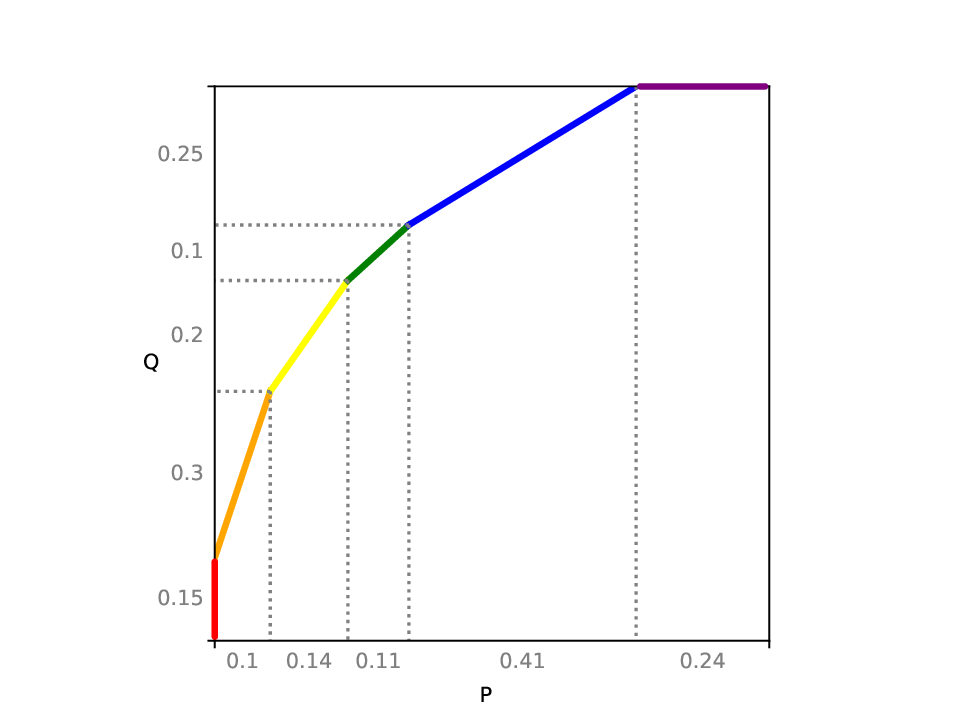}
	\caption{$g = \Pow(P||Q)$ }
	\end{subfigure}
	\hfill
	\begin{subfigure}{0.45\textwidth}
	\includegraphics[width = \textwidth]{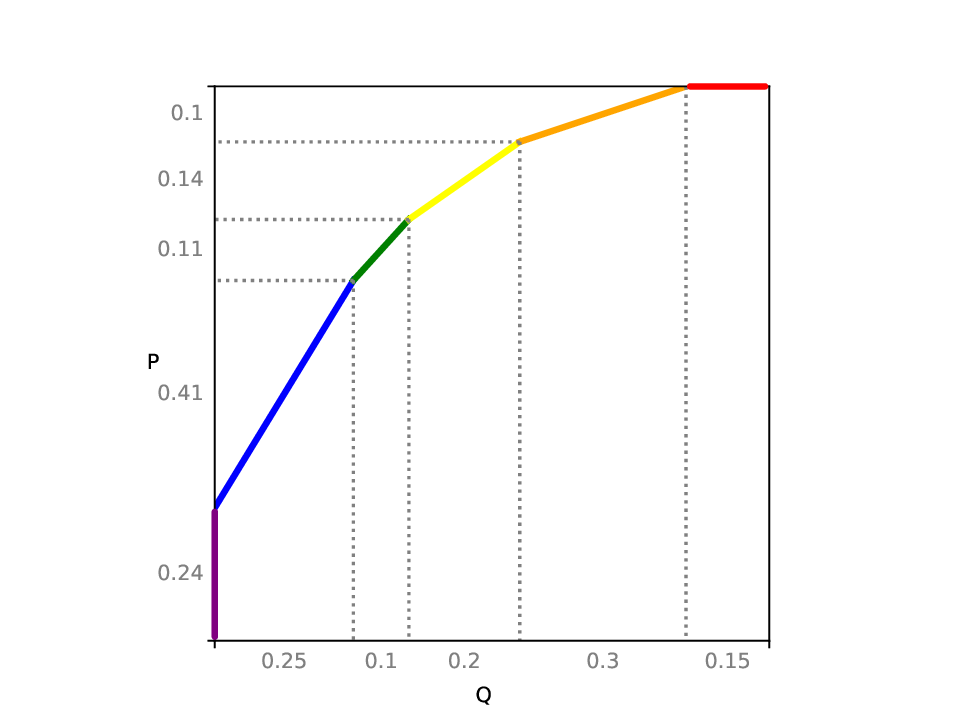}
	\caption{$\widehat{g} = \Pow(Q \| P) = \mf{R}(\Pow(P \|Q))$ }
	\label{fig:reflect}
	\end{subfigure}
	\caption{Consider a sample space $\Omega$ with
	6 elements indicated with different rainbow colors.
	Two distributions on $\Omega$ are given by the arrays:
	$P = [0, 0.1, 0.14, 0.11, 0.41, 0.24]$ and 
	$Q = [0.15, 0.3, 0.2, 0.1, 0.25, 0]$,
		where the elements~$\omega \in \Omega$
		are sorted in increasing order of
		the ratio $\frac{Q(\omega)}{P(\omega)}$.
		The left figure shows the power curve  $g = \Pow(P||Q)$, and the right figure shows $\widehat{g} = \Pow(Q \| P)$ which	is the reflection of $g$ about the line $y = 1 - x$.}
	\label{fig:power-func}
\end{figure}

In the literature (see~\cite{dong2022gaussian} for instance),
the tradeoff function $\mathsf{T}(P \| Q) := 1 - \Pow(P \| Q)$,
which has essentially equivalent mathematical properties,
has also been considered.
However, power functions have the same notation interpretation as divergences
in the sense that ``smaller'' means ``closer'' distributions.
In fact, Vadhan and Zhang~\cite{DBLP:conf/stoc/Vadhan023}
defined an abstract concept of 
\emph{generalized probability distance} (with the tradeoff function as the main example) such that
the notation agrees with this interpretation. 
We recap some well-known properties of power functions.

\begin{fact}[{Valid power functions~\cite[Proposition 2.2]{dong2022gaussian}}]
\label{fact:power}
Suppose $g : [0,1] \rightarrow [0,1]$ is a function.  Then, there exist
distributions $X$ and $Y$ such that $\Pow(X \| Y) = g$ \emph{iff}
$g$ is continuous, concave and $g(x) \geq x$ for all $x \in [0,1]$.
\end{fact}

\noindent \textbf{Partial Order on Power Functions.}
We consider a partial order on 
power functions.
\ignore{
non-negative functions
with domain $[0,1]$ (which include functions that are not power functions).
}
Given such functions $g_1$ and $g_2$,
we denote $g_1 \preceq g_2$
\emph{iff} for all $x \in [0,1]$, $g_1(x) \leq g_2(x)$.

We next state formally how a divergence
satisfying Definition~\ref{defn:dpi}
can be recovered from the power function.

\begin{fact}[Recovering Divergence From Power Function~{\cite[Proposition B.1]{dong2022gaussian}}]
\label{fact:div_pow}
For any divergence $\mathsf{D} \in \mcal{D}_{\textnormal{DPI}}$,
there exists a functional
$\ell_\mathsf{D}: [0,1]^{[0,1]} \rightarrow \R$ 
such that for any distributions $P$ and $Q$ on
the same sample space, 
$\mathsf{D}(P \| Q) = \ell_\mathsf{D}(\Pow(P \|Q))$.

Moreover, if $g_1 \preceq g_2$, then 
$\ell_\mathsf{D}(g_1) \leq \ell_\mathsf{D}(g_2)$.
\end{fact}

The following is the main result of this section,
which states that a locally maximin pair essentially
solves the closest distribution refinement problem optimally
in the most general sense.

\begin{theorem}[Locally Maximin Pair Gives Minimal Power Function]
\label{th:local_power}
Suppose in a distribution instance $G$, 
$\vec{\alpha} = (\alpha^{(0)}, \alpha^{(1)})$ is
a distribution refinement pair
that
is locally maximin 
and proportional response to each other.
Then, the power function~$g_\sdag$  between any other distribution refinement pair of $G$ satisfies
$\Pow(\alpha^{(0)} \| \alpha^{(1)}) \preceq g_\sdag$.
Consequently, any such locally maximin pair gives rise
to the same power function, and
we denote $\Pow^*(G) = \Pow(\alpha^{(0)} \| \alpha^{(1)})$.

From Fact~\ref{fact:div_pow},
this implies that 
$\vec{\alpha}$ is a universally closest
refinement pair with respect to $\mcal{D}_{\textnormal{DPI}}$.
\end{theorem}

\begin{remark}
\label{rem:approx_div}
Observe that given any $\mathsf{D} \in \mcal{D}_{\textnormal{DPI}}$
and any strictly increasing function $\phi: \R \rightarrow \R$, we still have
$\phi(\mathsf{D}( \cdot \| \cdot)) \in \mcal{D}_{\textnormal{DPI}}$.
Since any deviation from the correct answer can be
arbitrarily magnified by $\phi$, it is difficult to give any meaningful approximation
guarantee on the value of a general divergence in  $\mcal{D}_{\textnormal{DPI}}$.
\end{remark}

We prove Theorem~\ref{th:local_power} by analyzing the relationship
between hockey-stick divergence and power function,
and extend the result from Theorem~\ref{th:local_hs}.
Similar to convex functions, we also consider the notion of subgradient
for power functions.

\begin{definition}[Subgradient]
Suppose $g: [0,1] \rightarrow [0,1]$ is a power function
and $x \in [0,1]$.
Then, the subgradient of $g$ at $x$
is defined as the collection $\partial g(x)$ of real numbers satisfying: 

$\gamma \in \partial g(x)$ \emph{iff} for all $y \in [0,1]$, $g(x) + \gamma \cdot (y - x) \geq g(y)$.

In other words, the line segment with slope~$\gamma$ touching the curve $g$ at $x$ never goes below the curve.
\end{definition}

The following is a special case
of Fact~\ref{fact:div_pow},
and describes explicitly how a hockey-stick divergence
is recovered from a power function. See Figure~\ref{fig:recover} for an illustration.

\begin{fact}[Recovering Hockey-Stick Divergence From Power Function]
\label{fact:hs_pow}
Suppose $P$ and $Q$ are distributions on the same sample space~$\Omega$, where
$g = \Pow(P \| Q)$ is the power function between them.
Then, for any $\gamma \geq 0$,
there exists $x \in [0,1]$ such that $\gamma \in \partial g(x)$
and $\mathsf{D}_\gamma(P \| Q) = g(x) - \gamma x$.

In other words, one can find a line segment with slope $\gamma$
touching the curve $g$ and the $y$-intercept
of the line will give $\mathsf{D}_\gamma(P \| Q)$.
\end{fact}

\begin{proof}
The result should be standard, but we include a short proof for completeness.
One can imagine a line segment with slope~$\gamma$ starting high above
and gradually being lowered until touching the curve~$g$ at some $x \in [0,1]$.
Observe that the line segment may touch the curve at more than 1 point,
in which case we take $x$ to be the infimum.

We partition the sample space~$\Omega$
into $S_1 := \{\omega \in \Omega: Q(\omega) > \gamma \cdot P(\omega)\}$
and $S_2 := \Omega \setminus S_1$.
Observe that $P(S_1) = x$ and $g(x) = Q(S_1)$.

Then, it follows that 
$\mathsf{D}_\gamma(P \| Q) = Q(S_1) - \gamma \cdot P(S_1)
= g(x) - \gamma x$, as required.
\end{proof}

\begin{figure}[H]
	\centering
	\includegraphics[width=0.5\textwidth]{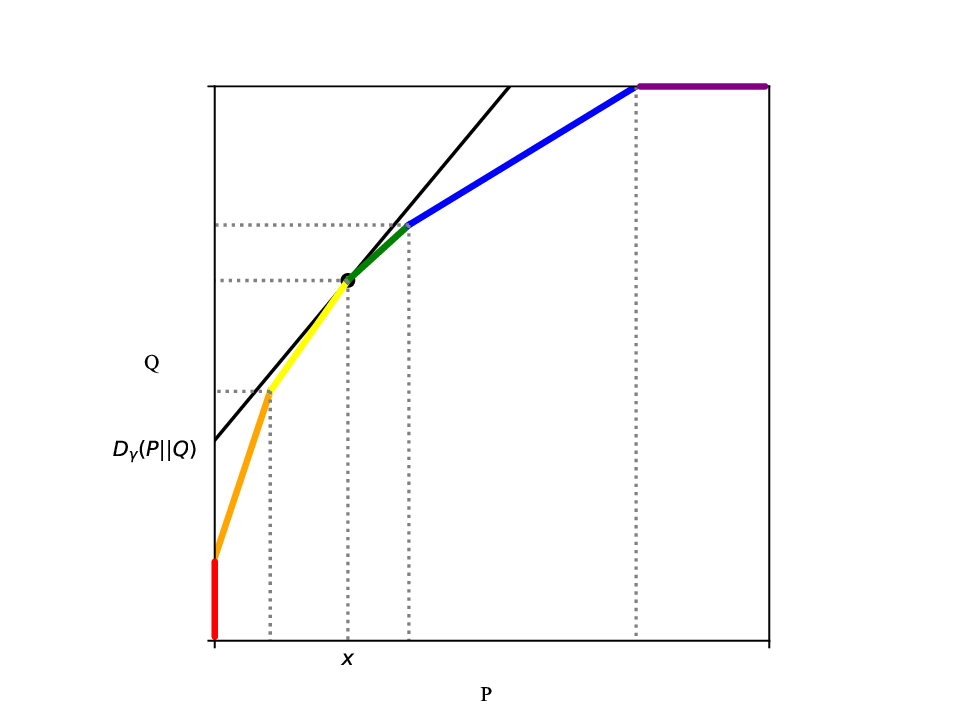}
	
	\caption{
	Using the same example in Figure~\ref{fig:power-func},
	consider a line segment
	with slope $\gamma = 1.2$ touching
	the curve $g$ at $x$.
	In this case, the red, orange, and yellow elements~$\omega$ to the left of point $x$ satisfy
	$\frac{Q(\omega)}{P(\omega)} > 1.2$.
	}
\label{fig:recover}
\end{figure}

\begin{proofof}{Theorem~\ref{th:local_power}}
Suppose $g = \Pow(\alpha^{(0)} \| \alpha^{(1)})$
is the power function between the
given locally maximin pair $\vec{\alpha}$.
For the sake of contradiction,
suppose some other refinement pair 
$\vec{\alpha}_\sdag$ produces the power function $g_\sdag$
such that for some $x \in [0,1]$,
$g(x) > g_\sdag(x)$.

Fix some $\gamma \in \partial g_\sdag(x)$
and consider the line segment $\ell_\sdag$ with slope $\gamma$ touching
curve~$g_\sdag$ at $x$.  By Fact~\ref{fact:hs_pow},
the hockey-stick divergence $\mathsf{D}_\gamma(\vec{\alpha}_\sdag)$
is the $y$-intercept of $\ell_\sdag$.

However, since $g(x) > g_\sdag(x)$,
to find a line segment with slope $\gamma$ touching curve $g$,
we must move $\ell_\sdag$ strictly upwards.
This implies that $\mathsf{D}_\gamma(\vec{\alpha})
> \mathsf{D}_\gamma(\vec{\alpha}_\sdag)$,
contradicting Theorem~\ref{th:local_hs}
which states that 
$\vec{\alpha}$ is universally closest
with respect to $\mcal{D}_{\textnormal{HS}}$.
\end{proofof}

\subsection{Approximation to Minimal Power Function}
\label{sec:approx_power}

Even though in Remark~\ref{rem:approx_div},
we mention that an approximate refinement
does not necessarily give any guarantee to the value of a general divergence notion,
we show that it is possible to give some guarantee
for the resulting power function.

\noindent \textbf{Notation.}  In order to describe 
an approximation notion for power functions,
we introduce some precise notation that facilitates
the description of power function transformation.

\noindent \emph{Power Curve.}
When we plot a power function $g$ in the $xy$-plane,
if $g(0) > 0$, we assume that there
is a vertical line segment from $(0,0)$ to $(0, g(0))$.
We also use $g$ to denote the power curve.

\noindent \textbf{Reflection.}
We use $\mf{R}(u,v) := (1-v, 1-u)$
to denote the reflection transformation of the
line $y = 1 - x$.
Observe that for any two distributions $P$ and $Q$,
$\Pow(Q \| P) = \mf{R}(\Pow(P \|Q))$;
see Figure~\ref{fig:reflect}.
Moreover, $g_1 \preceq g_2$ \emph{iff}
$\mf{R}(g_1) \preceq \mf{R}(g_2)$.

\noindent \emph{Convention.}
Observe that a power function~$g: [0,1] \to [0,1]$ is almost injective
except that multiple inputs might map to 1.
We use the convention that $g^{-1}(y) = \inf \{x \in [0,1]: g(x) = y\}$,
where the infimum is needed only for the case $y=1$.
Then, a reflected power function
can be expressed as $\mf{R}(g)(x) = 1 - g^{-1}(1-x)$.

\noindent \textbf{Approximation to Power Functions.}
Given a power function $g$ and $\gamma \geq 1$,
our goal is to define a notion of $\gamma$-approximation
to~$g$ that corresponds to a \emph{stretching} transformation $\mf{S}_\gamma$,
where $\gamma = 1$ is the identity transformation,
and $\gamma_1 < \gamma_2$ implies that
for all power functions~$g$, $\mf{S}_{\gamma_1}(g) \preceq \mf{S}_{\gamma_2}(g)$.

Since we have shown that the optimal power function
for the closest distribution refinement problem
can be achieved by the symmetric density decomposition,
ideally, the stretching transformation should also be symmetric,
i.e., $\mf{R} \circ \mf{S}_\gamma = \mf{S}_\gamma \circ \mf{R}$.

\noindent \textbf{Orientation of Power Curves.}
As aforementioned, a power curve
starts at $(0,0)$ and ends at $(1,1)$.
Given $\side \in \{0,1\}$,
we use $\side$ to denote the orientation of the curve
starting at $(\side, \side)$.
For $\side = 0$, the motion is clockwise;
for $\side = 1$, the motion is anti-clockwise.
In either case, the curve is like an arc moving
around the center $(1,0)$.  Since we need to describe
stretching with respect to the two curve orientations,
instead of using the terms vertical vs horizontal
(which would be exchanged after reflection),
we will use the terms \emph{radial} vs \emph{tangential}
in the sense as follows.

\begin{compactitem}
\item \emph{Radial Axis.}
For the curve orientation $\side = 0$ that
starts at $(0,0)$,
the radial direction refers to the $x$-axis
in the increasing direction.
For the curve orientation $\side = 1$
that starts at $(1,1)$,
the radial direction refers to the $y$-axis
in the decreasing direction.

\item \emph{Tangential Axis.}
For the curve orientation $\side = 0$,
the tangential direction refers to the
$y$-axis in the increasing direction.
For the curve orientation~$\side = 1$.
the tangential direction
refers to the $x$-axis in the decreasing direction.

\end{compactitem}

We define radial and tangential stretching as follows;
see Figure~\ref{fig:stretch}.

\begin{definition}[Power Curve Stretching]
\label{defn:curve_stretching}
Given $\gamma \geq 1$, orientation $\side \in \{0,1\}$
and direction $\msf{d} \in \{\msf{tan}, \msf{rad}\}$,
we define a stretching transformation 
$\mf{S}^{(\side, \msf{d})}_\gamma$.
We describe the case $\side = 0$
and use symmetry $\mf{S}^{(1, \msf{d})}_\gamma := \mf{R} \circ \mf{S}^{(0, \msf{d})}_\gamma \circ \mf{R}$ to define the case $\side = 1$.
Two directions of stretching can be defined for an
oriented power curve~$g$ as follows.

	\begin{itemize}
		\item 	\emph{Radial Stretching.}
		For $x \in [0,1]$, $\mf{S}^{(0, \msf{rad})}_\gamma(g)(x) := g(\min\{\gamma x, 1\})$; see Figure~\ref{fig:0rad}.
		
		In the orientation $\side = 0$, we can view this as pushing
		the curve horizontally towards the line $x=0$.

		\item 	\emph{Tangential Stretching.}
		For $x \in [0,1]$, $\mf{S}^{(0, \msf{tan})}_\gamma(g)(x) := 
		\min\{ \gamma \cdot (g(x) - g(0)) + g(0),  1\}$;
		see Figure~\ref{fig:0tan}.
		
		In the orientation $\side = 0$, we can view this as pushing
		the curve vertically towards the line $y=1$.

	\end{itemize}

From Fact~\ref{fact:power},
it is easy to check that the result of stretching a power function 
is still a power function.

\end{definition}

\begin{figure}[!ht]
	\centering
	\begin{subfigure}{0.45\textwidth}
		\includegraphics[width = \textwidth]{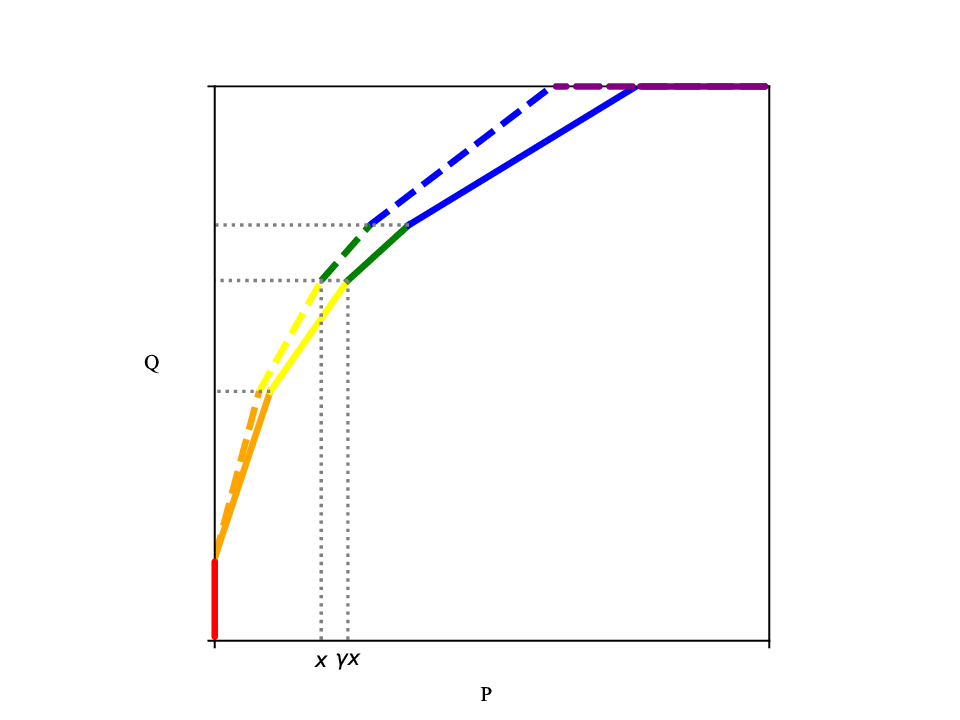}
	\begin{tikzpicture}[overlay, remember picture]
		\draw[-stealth, ultra thick, black] (2.6,0.5) -- (1.6,0.5);
	\end{tikzpicture}
		\caption{$\mf{S}^{(0, \msf{rad})}_\gamma(g)(x)$}
		\label{fig:0rad}
		\vspace{1cm}
	\end{subfigure}	
	\hfill
	\begin{subfigure}[b]{0.45\textwidth}
		\includegraphics[width = \textwidth]{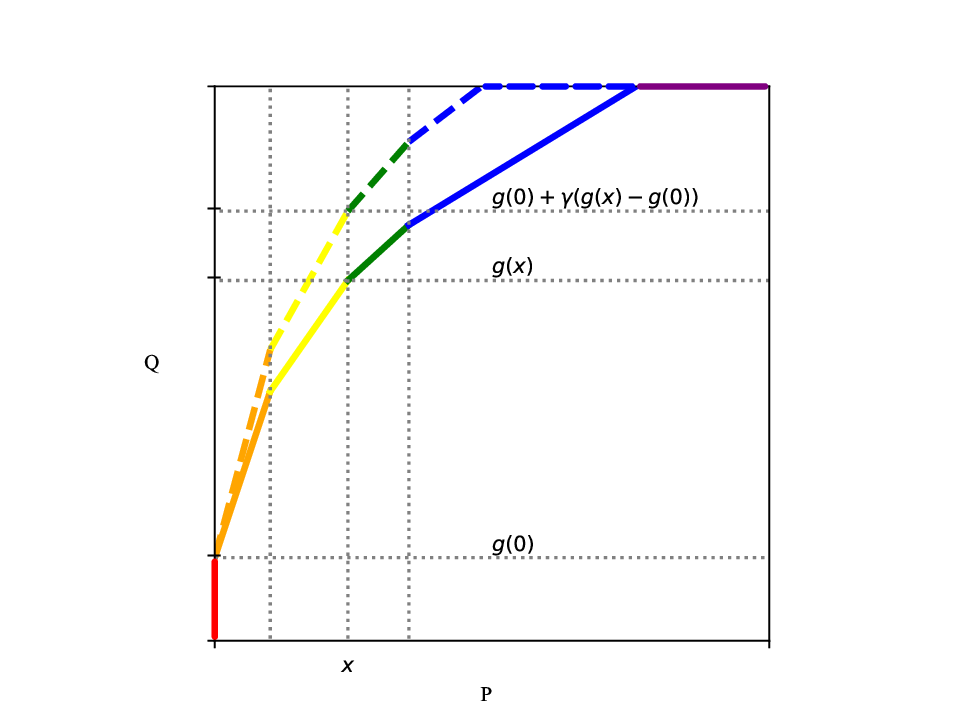}
		 \begin{tikzpicture}[overlay, remember picture]
			\draw [-stealth, ultra thick, black] (1,1) -- (1,2);
		\end{tikzpicture}
		\caption{$\mf{S}^{(0, \msf{tan})}_\gamma(g)(x)$}
		\label{fig:0tan}
		\vspace{1cm}
	\end{subfigure}
	\vspace{1cm}
	
	\begin{subfigure}{0.45\textwidth}
		\includegraphics[width = \textwidth]{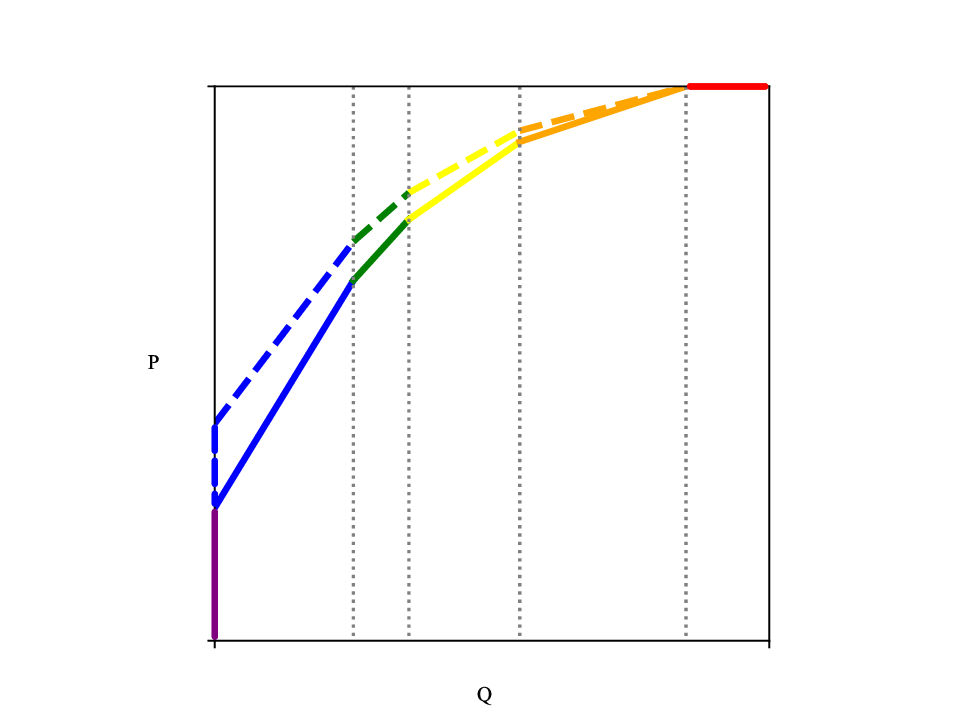}
		\begin{tikzpicture}[overlay, remember picture]
			\draw[-stealth, ultra thick, black] (6.3,4.4) -- (6.3,5.4);
		\end{tikzpicture}
		\caption{$\mf{S}^{(1, \msf{rad})}_\gamma(\widehat{g})(x)$}
		\label{fig:1rad}
	\end{subfigure}
	\hfill
		\begin{subfigure}[b]{0.45\textwidth}
		\includegraphics[width = \textwidth]{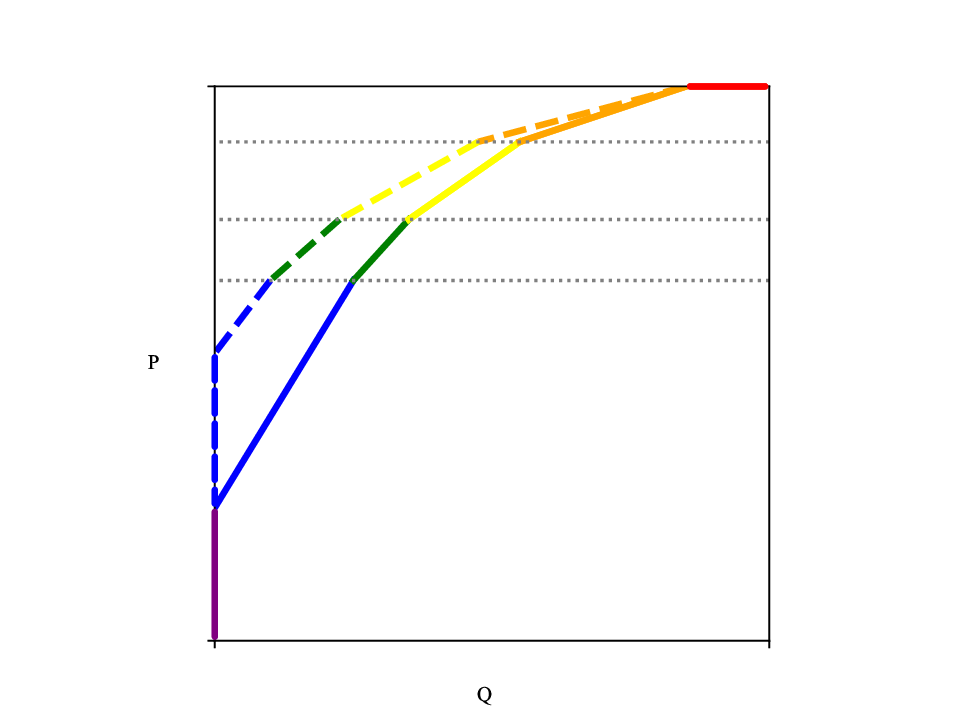}
		\begin{tikzpicture}[overlay, remember picture]
			\draw[stealth-, ultra thick, black] (5,5.8) -- (6,5.8);
			
		\end{tikzpicture}
		\caption{$\mf{S}^{(1, \msf{tan})}_\gamma(\widehat{g})(x)$}
		\label{fig:1tan}
	\end{subfigure}
	\caption{We use the same example
	in Figure~\ref{fig:power-func} with the power
	curves $g$ and $\widehat{g} = \mf{R}(g)$,
	with stretching factor $\gamma = 1.25$.
	In each subfigure, the location of the black arrow
	denotes the orientation, where an arrow
	next to the point $(\side, \side)$ indicates
	orientation~$\side \in \{0,1\}$.
	The direction of the arrow indicates the direction of
			the stretching.  The solid curve indicates the original
			power function, and the dotted curve indicates
			the stretched power function.
			Observe that (a) and (c) are reflections of each other,
			and (b) and (d) are reflections of each other.
				}
	\label{fig:stretch}
\end{figure}
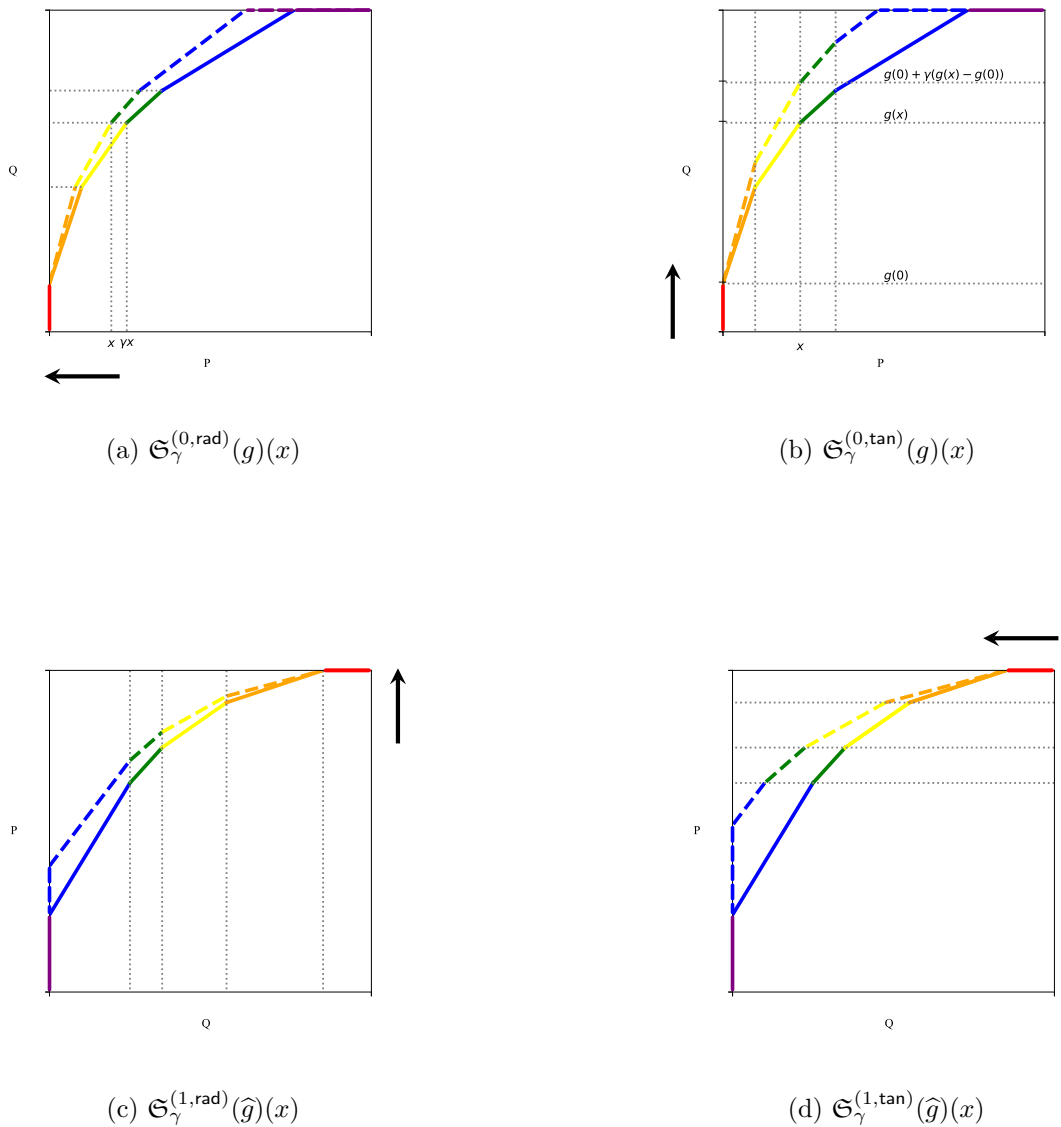

Using the geometric interpretation,
one can verify the expressions of
the stretching transformations
for the orientation $\side = 1$.

\begin{fact}[Stretching Transformations for Orientation $\side = 1$]
\label{fact:stretch_side}
For any power function~$g$ and $x \in [0,1]$,
we have the following expressions.

\begin{compactitem}

\item $\mf{S}^{(1, \msf{rad})}_\gamma(g)(x) = \frac{1}{\gamma} \cdot g(x) + 1 - \frac{1}{\gamma}$;
see Figure~\ref{fig:1rad}.

\item $\mf{S}^{(1, \msf{tan})}_\gamma(g)(x) 
= g(\frac{1}{\gamma} \cdot x + (1 -  \frac{1}{\gamma}) \cdot g^{-1}(1))$;
see Figure~\ref{fig:1tan}.

\end{compactitem}

\end{fact}


\noindent \textbf{Relating Power Curve Stretching
to Density Perturbation in Fractional Knapsack Instance.}
Recall that the power function $g = \Pow(P \| Q)$
between distributions $P$ and $Q$ on sample space
has a fractional knapsack interpretation as
in Definition~\ref{defn:power},
where each item $\omega \in \Omega$ has weight
$P(\omega)$ and value~$Q(\omega)$.

In the curve orientation $\iota = 0$
where we start from $(0,0)$,
given $x \in [0,1]$,
recall that $g(x)$ is the maximum achievable
value by (fractionally) including items with total weight at most~$x$.

On the other hand, the curve orientation~$\iota = 1$ (starting
at $(1,1)$)
can be interpreted as the \emph{reversed} knapsack
problem.
Given $z \in [0,1]$, the quantity $1 - g(1-z)$
can be interpreted as the minimum achievable
value by (fractionally) including items with total weight 
at least~$z$.

\noindent \emph{Density Perturbation and Power Curve Stretching.}
Suppose for the items in $\Omega$,
the weights $P$ stay the same, but
the values $Q$ are perturbed to some new $\widehat{Q}$
(which is still a distribution)
within a multiplicative factor of $\gamma \geq 1$.\footnote{Because
of the application to the closest distribution refinements problem,
we include the extra condition that if $P(\omega) = 0$,
then the weight $\widehat{Q}(\omega) = Q(\omega)$ stays the same.}
Intuitively, for the fractional knapsack problem,
the optimal value can increase by a factor of at most~$\gamma$,
while the optimal value for the reversed
knapsack problem can decrease by a factor
of at most~$\gamma$.
When this is interpreted by the power curve, note that even though stretching happens in the values ($y$-direction)
in both cases, because of reflection,
tangential stretching happens for the orientation $\side =0$ (original knapsack),
while radial stretching happens for orientation $\side = 1$ (reversed knapsack).
This intuition is formalized in the following lemma.

\begin{lemma}
	\label{lem:approx-knapsack}
	Suppose $P$ and $Q$ are distributions on sample space~$\Omega$
	with the power function $g := \Pow(P \| Q)$.
	Moreover, $\widehat{Q}$ is a perturbed distribution of $Q$ for which
	there exist $\gamma_1, \gamma_2 \geq 1$
	such that for all $\omega \in \Omega$ satisfying
	$P(\omega) > 0$,  $\frac{1}{\gamma_2} \cdot Q(\omega) \leq \widehat{Q}(\omega) 
	\leq \gamma_1 \cdot Q(\omega)$;
	if $P(\omega) = 0$, then $\widehat{Q}(\omega) = Q(\omega)$.
	
	Then, the power function between $P$ and $\widehat{Q}$
	satisfies:
	
	$$\widehat{g} := \Pow(P \| \widehat{Q})
	\preceq \min\{\mf{S}^{(0, \msf{tan})}_{\gamma_1}(g),
				\mf{S}^{(1, \msf{rad})}_{\gamma_2}(g)\}.$$
	
\end{lemma}

\begin{proof}
Without loss of generality,
items~$\omega$ with $P(\omega) = Q(\omega) = 0$ may
first be removed.
We consider how the new
power function $\widehat{g} = \Pow(P \| \widehat{Q})$
compared with the original~$g$.
Since we want to bound the worst case when
$P$ and $\widehat{Q}$ become even further apart,
intuitively this happens if
the items $\omega$ with higher densities $\frac{Q(\omega)}{P(\omega)}$
have their values $\widehat{Q}(\omega)$ increased further,
while the opposite holds for items with lower densities.

Using the interpretation of the power function
via the fractional knapsack problem as in Definition~\ref{defn:power},
observe that the fractional knapsack problem
can be solved optimally by using the densities of items
as the greedy heuristic.  Hence,
in the orientation $\side = 0$ starting at $(0,0)$,
we show that
the maximum value achievable in the knapsack problem
increases by a factor of at most~$\gamma_1$;
this corresponds to the stretching transformation
$\mf{S}^{(0, \msf{tan})}_{\gamma_1}$.
On the other hand,
when we consider the other orientation $\side = 1$ starting at $(1,1)$,
this can be interpreted as the reversed knapsack problem,
in which we show that the minimum value achievable decreases by a factor
of at most~$\gamma_2$;
this corresponds to the stretching
transformation $\mf{S}^{(1, \msf{rad})}_{\gamma_2}$.

\noindent (1) We prove the result formally for the
orientation $\side = 0$.
Since $\widehat{Q}(\omega) = Q(\omega)$ stays the same
if $P(\omega) = 0$,
$\widehat{g}(0) = g(0) = \sum_{\omega: P(\omega) = 0} Q(\omega)$.

Suppose for the sake of contradiction,
there exists $0 < x \leq 1$
such that

$\widehat{g}(x) > 
\mf{S}^{(0, \msf{tan})}_{\gamma_1}(g)(x) = 
\min\{ g(0) + \gamma_1 \cdot (g(x) - g(0)) ,  1\}$.

If $\mf{S}^{(0, \msf{tan})}_{\gamma_1}(g)(x) = 1$,
we immediately get a contradiction, because
a power function cannot take values strictly greater than 1.
Hence, we may assume
$\mf{S}^{(0, \msf{tan})}_{\gamma_1}(g)(x) = g(0) + \gamma_1 \cdot (g(x) - g(0))$
from now on.

Recall that $\widehat{g}(x)$ is the optimal value achieved
in the fractional knapsack problem (with item weights
$P$ and values $\widehat{Q}$) where the total weight constraint at most~$x > 0$.
Observe that in the knapsack solution, one can always include items~$\omega$ with 
zero weight $P(\omega) = 0$ (whose aggregate value is $g(0)$).
Suppose $S$ is the corresponding collection
of (fractional) items with non-zero weights
whose aggregate value (with respect to $\widehat{Q}$) attains 
$\widehat{g}(x) - g(0)$.

Next, observe that if $P(\omega) > 0$,
then $Q(\omega) \geq \frac{1}{\gamma_1} \cdot \widehat{Q}(\omega)$.
Hence, in the original instance of the knapsack problem
with values~$Q$, the solution $S$ together with zero-weighted items
gives an aggregate value of
at least $g(0) + \frac{1}{\gamma_1} \cdot (\widehat{g}(x) - g(0))
> g(x)$, which gives the required contradiction.

\noindent (2) We next consider the other curve orientation
$\side = 1$ starting at $(1,1)$.
Recall that with respect to the
original values~$Q$ in the reversed knapsack problem,
for $z \in [0,1]$,
$1 - g(1-z)$ is the minimum value achieved 
by fractionally selecting items with aggregate weight at least~$z$.

Let $N := \{\omega \in \Omega: Q(\omega) = 0\}$. Observe that
$Q(\omega) = 0$ \emph{iff} $\widehat{Q}(\omega) = 0$.
Hence, it follows that for $0 \leq z \leq z_0 := P(N)$,
$g(1-z) = \widehat{g}(1-z) = 1$,
because in the reversed knapsack problem,
when the required aggregate weight is less than $P(N)$,
it is possible to fractionally select items in $N$ to achieve zero value.

Next, we consider
the stretched curve $\mf{S}^{(1, \msf{rad})}_{\gamma_2}(g)$.
For $x \in [0, 1]$, from Fact~\ref{fact:stretch_side}, we have
$\mf{S}^{(1, \msf{rad})}_{\gamma_2}(g)(x) = \frac{1}{\gamma_2} \cdot g(x) + 1 - \frac{1}{\gamma_2}$.

Observe that for  $x \in [1 - z_0, 1]$,
$\mf{S}^{(1, \msf{rad})}_{\gamma_2}(g)(x) = \widehat{g}(x) = 1$.
For the sake of contradiction,
suppose that there exists some $x \in [0, 1 - z_0]$
such that $\widehat{g}(x) > \frac{1}{\gamma_2} \cdot g(x) + 1 - \frac{1}{\gamma_2}$.
This implies that in the reversed knapsack problem with values~$\widehat{Q}$
and weight requirement~$z = 1 - x \geq z_0$,
there is a fractional collection $S$ of items with 
aggregate value $1 - \widehat{g}(x) < \frac{1}{\gamma_2} \cdot (1 - g(x))$.

Since in the reversed knapsack problem,
one would never need to pick any item~$\omega$ with $P(\omega) = 0$,
it follows that any selected item~$\omega$ must
satisfy $Q(\omega) \leq \gamma_2 \cdot \widehat{Q}(\omega)$.
Therefore, it follows that in the reversed knapsack problem
with original values~$Q$ and weight requirement $z$, the collection~$S$
is a feasible solution with aggregate value
at most $\gamma_2 \cdot (1 - \widehat{g}(x)) < 1 - g(x) = 1 - g(1-z)$,
thereby reaching a contradiction.

Combining both parts, we have the required result.
\end{proof}

Lemma~\ref{lem:approx-knapsack} suggests an approximation notion
for power curves that involves both radial and tangential stretching
transformations. Hence, it would not readily give a symmetric approximation notion.
The next lemma shows that radial and tangential stretching transformations are
comparable, which means we can pick one of them to define an approximation
notion on power curves.

\begin{lemma}[Comparing Radial and Tangential Stretchings]
\label{lemma:compare_stretch}
For each orientation~$\side \in \{0,1\}$, $\gamma \geq 1$
and any power function~$g$,

$\mf{S}^{(\side, \msf{rad})}_\gamma(g) \preceq \mf{S}^{(\side, \msf{tan})}_\gamma(g).$
\end{lemma}

\begin{proof}
To simplify the notation, 
we extend the domain of a power function to $g:[0, +\infty) \to [0,1]$,
where $g(x) = 1$ for $x \geq 1$.  Observe that $g$ is still concave.

We show that result for $\side = 0$,
and use the expressions in Definition~\ref{defn:curve_stretching}.
It suffices to show that for all $x \in [0,1]$,

\begin{equation} \label{eq:compare}
g(\gamma x) \leq \gamma \cdot (g(x) - g(0)) + g(0).
\end{equation}

This is because we still have $g(x) \leq 1$ for all $x \geq 0$.
Hence, when the right-hand side of (\ref{eq:compare}) is larger than 1,
the left-hand side is still at most 1.

Rearranging (\ref{eq:compare}) is equivalent to:

$$\frac{1}{\gamma} \cdot g(\gamma x) + (1 - \frac{1}{\gamma}) \cdot g(0)
\leq g(\frac{1}{\gamma} \cdot \gamma x + (1 - \frac{1}{\gamma}) \cdot 0),$$

which holds because $g$ is concave.

The result for the orientation $\side = 1$ follows immediately by symmetry.
\end{proof}

From Lemma~\ref{lemma:compare_stretch},
it follows that using tangential stretching
will lead to a more convenient notion of power curve approximation;
see Figure~\ref{fig:approx}.

\begin{definition}[$\gamma$-Approximation for Power Function]
\label{defn:approx_power}
Given $\gamma \geq 1$ and a power function~$g$,
the $\gamma$-approximation of $g$ is defined as:

$\mf{S}_\gamma(g) := \min \{\mf{S}^{(0, \msf{tan})}_\gamma(g),
\mf{S}^{(1, \msf{tan})}_\gamma(g)\}.$
\end{definition}

From Fact~\ref{fact:power},
it is easy to check that 
 applying the operator $\mf{S}_\gamma$
to a power function~$g$
will return a power function 
$\mf{S}_\gamma(g)$,
because the minimum of two concave functions is still concave.

\begin{figure}[H]
	\centering
	\begin{subfigure}[b]{0.45\textwidth}
		\includegraphics[width = \textwidth]{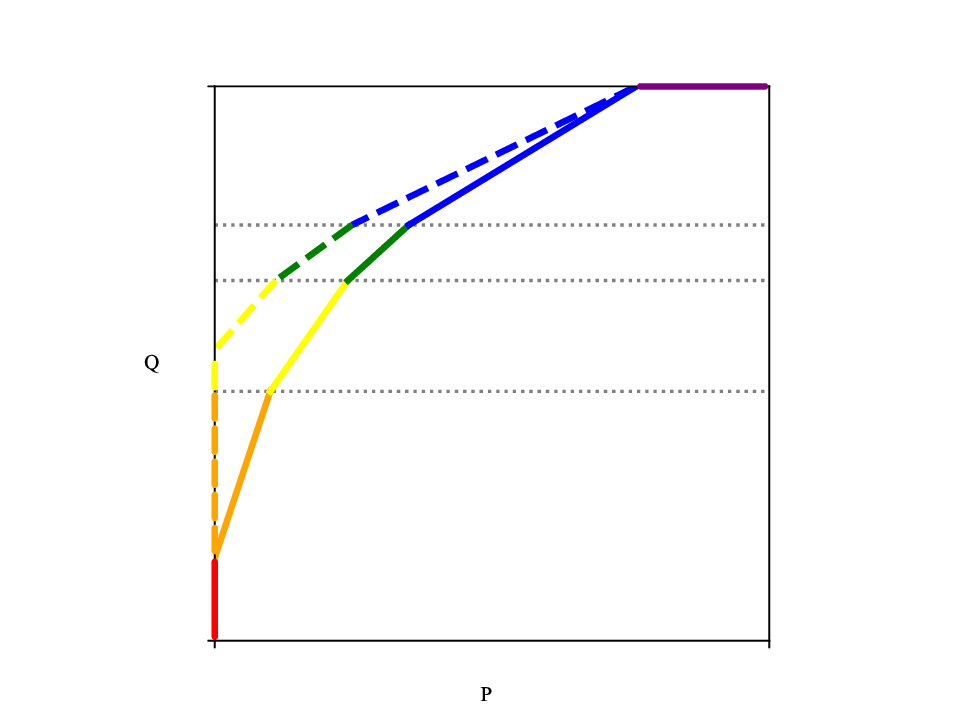}
		\caption{$\mf{S}^{(1, \msf{tan})}_\gamma(g)$}
	\end{subfigure}
	\hfill
	\begin{subfigure}{0.45\textwidth}
		\includegraphics[width = \textwidth]{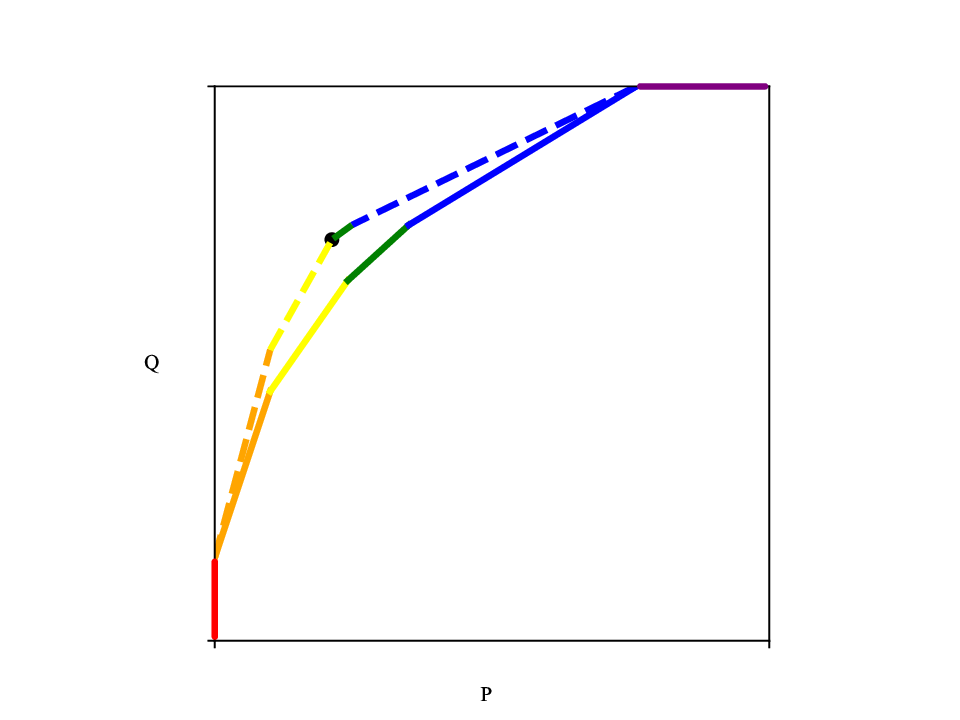}
		\caption{$\mf{S}_\gamma(g)$}
	\end{subfigure}
	\caption{We consider the same $g = \Pow(P||Q)$
	from Figure~\ref{fig:power-func}.
		The dotted curve in left figure shows $\mf{S}^{(1, \msf{tan})}_\gamma(g)$. Taking the minimum with
		$\mf{S}^{(0, \msf{tan})}$ from Figure~\ref{fig:0tan},
		the dotted curve on the right figure shows
		shows $\mf{S}_\gamma(g)$.}
	\label{fig:approx}
\end{figure}

Now we have all the technical tools to describe
the approximation result for power function.

\begin{theorem}[Approximation to Minimal Power Function]
\label{th:min_power}
Suppose in a distribution instance $G$, 
for $0 \leq \tau \leq \frac{1}{2}$,
$\vec{\alpha} = (\alpha^{(0)}, \alpha^{(1)})$ is
a distribution refinement pair
such that
$\alpha^{(1)}$ achieves $\tau$-multiplicative error,
and $\alpha^{(0)}$ is the proportional response to $\alpha^{(1)}$.
Then, the corresponding power function
satisfies

$$\Pow(\alpha^{(0)} \| \alpha^{(1)}) \preceq \mf{S}_{1+2\tau}(\Pow^*(G)).$$
\end{theorem}

\begin{proof}
Consider the density decomposition
as in Definition~\ref{def:decomposition}
with $\Izero$ as the ground set,
in which element $i \in \Izero$
has density $\rho_*(i)$.

In view of Fact~\ref{fact:maximin_decomp}
that relates the density decomposition
to locally maximin refinements
and the interpretation
of power functions in Definition~\ref{defn:power}
as the fractional knapsack problem,
the power function $g^* := \Pow^*(G)$ is related to
the following instance of the fractional knapsack
problem.

Given some locally maximin
refinement pair $(\alpha^{(0)}_*, \alpha^{(1)}_*)$
that are proportional responses to each other,
the items in the knapsack problem 
are pairs in $\mcal{F}$ such that
$(i,j) \in \mcal{F}$
has weight
$\alpha^{(0)}_*(ij)$ and value~$\alpha^{(1)}_*(ij)$.
However,
if we fix some~$i \in \Izero$,
then any $(i,j) \in \mcal{F}$ with
$\alpha^{(0)}_*(ij) > 0$
must have value-to-weight density
$\rho_*(i) = \frac{\alpha^{(1)}_*(ij)}{\alpha^{(0)}_*(ij)}$.
Since we may select items fractionally,
it is equivalent to consolidate 
all items with the same density as a single item (with the
corresponding aggregated weight).

Therefore, we can interpret
$\Pow^*(G)$ via an alternative instance of the fractional knapsack problem
in which $\Izero$ is the collection
of items, where each $i \in \Izero$ has weight $w(i)$
and value-to-weight density $\rho_*(i)$.
Recall that for $x \in [0,1]$,
$\Pow^*(G)(x)$ is the maximum achieved value by
fractionally picking items whose aggregate weight is at most $x$.


Next,  we consider a perturbation of the item values (but keeping the same item weights). If $\alpha^{(1)}$ is a refinement on 
vertex weights on $\Ione$ and has $\tau$-multiplicative error,
this means that this refinement will induce
a density vector $\rho$ for $\Izero$
such that for each $i \in \Izero$,
$(1-\tau) \rho_*(i) \leq \rho(i) \leq (1 +\tau) \rho_*(i)$.

Moreover, if $\alpha^{(0)}$ is the proportional response to
$\alpha^{(1)}$, then
the power function $g := \Pow(\alpha^{(0)} \| \alpha^{(1)})$
corresponds to an instance
of the fractional knapsack problem
such that each item $i \in \Izero$ still has weight~$w(i)$,
but its value is $\rho(i) \cdot w(i)$.

Hence, we can apply Lemma~\ref{lem:approx-knapsack}
with $\gamma_1 = 1 + \tau$
and $\gamma_2 = \frac{1}{1 - \tau} \leq 1 + 2 \tau$
for $\tau \in [0, \frac{1}{2}]$.
Then,
we have
$g \preceq
\min\{\mf{S}^{(0, \msf{tan})}_{\gamma_1}(g^*),
				\mf{S}^{(1, \msf{rad})}_{\gamma_2}(g^*)\}.$
				
Finally, applying Lemma~\ref{lemma:compare_stretch}
to compare tangential and radial stretching
and using Definition~\ref{defn:approx_power}
for power curve approximation gives the required result.
\end{proof}

\section{Characterizing Market Equilibrium via Local Maximin Condition}
\label{sec:market}

In this section, 
we show how an instance in Definition~\ref{defn:input}
can be interpreted as a special symmetric case
of a linear Fisher market, which is itself a special case
of a linear Arrow-Debreu market.
We show that in Fisher markets,
a market equilibrium can be equivalently
characterized by a simpler local maximin condition.
Therefore, existing algorithms for finding approximate market equilibria
may be used to find approximate refinements as in Definition~\ref{defn:approx_refine}.

\subsection{Linear Arrow-Debreu and Fisher Markets}

\begin{definition}[Linear Arrow-Debreu Market]
\label{defn:arrow-debreu}
There is a collection $\mcal{I}$ of agents,
where each~$i \in \mcal{I}$ has one unit of divisible
good of type~$i$.
From agent~$i$'s perspective,
per unit of another good~$j$ has value~$w_i(j)$.
\ignore{
A typical assumption is that
for any non-empty proper  subset $\emptyset \neq S \subsetneq \mcal{I}$,
there exist $i \in S$ and $j \notin S$ such that $w_i(j) > 0$.
}

An allocation~$\vec{x}$ specifies how each good is \emph{completely} allocated
to the agents. Specifically, $x_{j \to i} \geq 0$
is the fraction of good~$j$ assigned to agent~$i$, i.e.,
for each $j$, $\sum_{i \in \mcal{I}} x_{j \to i} = 1$.

Under allocation $\vec{x}$, the utility of agent~$i$ is $u_i(\vec{x}) := \sum_{j \in \mcal{I}} w_i(j) \cdot
x_{j \to i}$.

\noindent \textbf{Special Bipartite Case.}
A market is \emph{bipartite} if the agents can be partitioned into two sides
such that $w_i(j) > 0$ implies that $i$ and $j$ are from different sides.
Moreover, an allocation~$\vec{x}$ is bipartite (with
respect to the same bipartition) if
$x_{i \to j} > 0$ implies that $i$ and $j$ are from different sides.
\end{definition}

Equilibrium for Arrow-Debreu markets is defined by introducing
the notion of good \emph{prices}, which gives an evaluation for each 
good.  An intuitive view
is that there is some platform that buys the good from each agent
at this price, after which each agent can use this earned money
to buy other goods from the platform at their corresponding prices.

\begin{definition}[Arrow-Debreu Market Equilibrium]
\label{defn:arrow-debreu-equi}
Given an instance of a market in Definition~\ref{defn:arrow-debreu},
an allocation $\vec{x}$ is an equilibrium if there
exist positive good prices $\vec{p} \in \R^{\mcal{I}}$
such that the following holds.

\begin{compactitem}
\item Agent~$i$'s good is considered to worth money~$p_i$ per unit, and
agent~$i$ must use all this amount to buy goods to realize the
allocation $\vec{x}$, i.e.,
for each $i$, we have $p_i = \sum_{j \in \mcal{I}} x_{j \to i} \cdot p_j$.

\item Agent~$i$ is going to spend money
on goods that have the best value-to-price ratio,
i.e., $x_{j \to i} > 0$
only if $j \in \arg \max_{k \in \mcal{I}} \frac{w_i(k)}{p_k}$.
\end{compactitem}
\end{definition}

With the introduction of money, observe that an agent~$i$ receiving
good~$j$ in an equilibrium allocation~$\vec{x}$ does not mean that agent~$j$ will receive good~$i$.
In other words, it is possible that
$x_{j \to i} > 0$, but $x_{i \to j} = 0$.

\noindent \textbf{Non-Uniqueness for Equilibrium Prices.}
Note that equilibrium prices are not unique, 
because any positive scalar multiple of an equilibrium price
vector is still an equilibrium.
However, equilibrium prices are not unique even after normalization
with some scalar.  An extreme example is when there are two independent
ecosystems, where equilibrium prices in each system may be scaled
arbitrarily.  This example can be modified to an irreducible market,
where an agent from one system has non-zero, but minimal, valuations
for goods from the other system.
On the other hand, the utilities of agents are unique at equilibrium.
Therefore, an equilibrium characterization based on a local maximin condition
in terms of utilities is more direct without going through the notion of prices.

A Fisher market can be viewed as a special case
of bipartite Arrow-Debreu markets,
where the view of one side is \emph{objective}
in the sense defined as follows.

\begin{definition}[Fisher Market]
A Fisher market can be viewed as a special case of
Definition~\ref{defn:arrow-debreu} in which
the agents are partitioned into buyers
$\mcal{B}$ and sellers $\mcal{S}$ such that
the valuations of agents satisfy the following.

\begin{compactitem}

\item Each buyer has non-zero valuations
only for sellers' goods.
The valuation of buyer~$i$ on
per unit of seller~$j$'s goods can be
any $w_i(j) \geq 0$.
In this case, we say that the view of the buyers may be
\emph{subjective}.

\item Each seller has non-zero valuations
only for buyers' goods.  Moreover,
for each buyer~$i \in \mcal{B}$,
there exists $B_i$ (\emph{aka} budget) such that
if buyer~$i$ has a positive valuation $w_i(j) > 0$
on seller~$j$'s good,
then seller~$j$ derives value $w_j(i) = B_i$
from per unit of buyer~$i$'s good.

As we shall see, if a buyer~$i$ has zero $w_i(j) = 0$ interest
on good~$j$, then in an equilibrium, buyer~$i$ will never
receive any positive fraction of good~$j$.
Hence, the value $w_j(i)$ of buyer~$i$'s good from the perspective
of seller~$j$ is irrelevant.
The convention is that we assume either $w_j(i) = B_i$ (in which case
the view of the sellers is \emph{objective}),
or $w_j(i) = 0$ (i.e., sellers' view is partially objective).

\end{compactitem}

As before, an allocation~$\vec{x}$ specifies how each agent
distributes its good completely to other agents.
For buyer~$i$ and seller~$j$,
the notation $b_{i \to j} := B_i \cdot x_{i \to j}$
is also used.
\end{definition}

\begin{remark}[Fisher Market Equilibrium]
\label{rem:fisher_equi}
Because the Fisher market has a special structure,
the prices $\vec{p}$ for an equilibrium allocation $\vec{x}$ in
Definition~\ref{defn:arrow-debreu-equi}
are chosen to have the following special form in the literature.

\begin{compactitem}

\item As aforementioned, equilibrium prices are not unique in an Arrow-Debreu
equilibrium.  However, it is possible to choose,
for each buyer~$i$, $p_i = B_i$.

This means that the good for each buyer~$i$ becomes a
form of currency from the perspective of sellers.
Hence, sellers view that every buyer's good has the same value-to-price ratio.

Observe that when viewed as an Arrow-Debreu market,
the equilibrium prices for buyers' goods need 
not be proportional to the budgets.

\item For each seller~$j$, the price
of its good is naturally induced by
the amount it receives from the buyers: $p_j = \sum_{i \in \mcal{B}} b_{i \to j}$,
where $b_{i \to j} = B_i \cdot x_{i \to j}$.
Since we can view that seller~$j$ has utility $B_i$ for buyer~$i$'s total budget,
sometimes it is more convenient to view $p_j$ as the utility of
seller~$j$.

In the literature, for Fisher markets, $\vec{b}$ typically refers to 
the actual amount $b_{i \to j}$ sent from a buyer~$i$ to a seller~$j$,
and $\vec{x}$ refers to the fraction $x_{j \to i}$
of good~$j$ allocated to buyer~$i$.  However,
we will continue to use the notation $x_{i \to j} = \frac{b_{i \to j}}{B_i}$.

Even though equilibrium allocations may not be unique,
it is known that the equilibrium utilities are unique.
After we fix $p_i = B_i$ for all buyers~$i$ as aforementioned,
for each seller~$j$, $p_j$ represents both the corresponding price and utility,
which is therefore unique under equilibrium.

\item In the literature, an equilibrium allocation in a Fisher market
also satisfies a proportional response condition.  Specifically,
the allocation of sellers' goods is a proportional response to the allocation
of the buyers' budgets.  This is captured by the equation that
for each buyer~$i$ and seller~$j$:

$B_i \cdot x_{i \to j} = b_{i \to j} = x_{j \to i} \cdot p_j$.

Contrary to the more general Arrow-Debreu equilibrium,
it must hold that $x_{i \to j} > 0$ \emph{iff} $x_{j \to i} > 0$.
\end{compactitem}

The above discussion implies that it is straightforward
to check whether a budget allocation $\vec{b}$ from buyers to sellers
is an equilibrium.  The price $p_j$ of seller~$j$'s good
is simply the total amount of money received by~$j$.
Then, it suffices to check that for every buyer~$i$,
$b_{i \to j} > 0$ only if $j$ attains the minimum
price to value ratio $\frac{p_j}{w_i(j)}$.
The formal argument is given in Fact~\ref{fact:fisher_local}.

On the other hand, given an allocation $\vec{x}$ of
sellers' goods to buyers, it is possible to check whether it is
an equilibrium by first deriving
a budget allocation $b_{i \to j} := \frac{w_i(j) \cdot x_{j \to i}}{u_i(\vec{x})} \cdot B_i$ by proportional response, and then check whether
the resulting $\vec{b}$ is an equilibrium. (It is known
that if the derived $\vec{b}$ is an equilibrium,
then the proportional response to $\vec{b}$ will return the original $\vec{x}$.)

One of our discoveries is that it is possible to directly check whether $\vec{x}$
is an equilibrium by considering a similar local maximin condition at every seller.
\end{remark}

We next describe how an input instance in Definition~\ref{defn:input}
can be interpreted as a symmetric instance of Fisher markets,
in which both buyers and sellers have (partially) objective views.
Hence, there is no distinction between buyers and sellers, as both groups of agents have exactly the same role.

\begin{definition}[Symmetric Fisher Market]
\label{defn:sym_fisher}
Given an instance $(\Izero, \Ione; \mcal{F}; w)$,
for side~$\side \in \B$,
an agent~$i \in \Ib$ has a non-zero valuation
on some good $j \in \Iob$ only if $\{i, j\} \in \mcal{F}$,
in which case agent~$i$'s valuation on per unit of good~$j$ is
$w^{(\ob)}(j)$.
\end{definition}

\noindent \textbf{Equilibrium Characterization for Symmetric Markets.}
Observe that as mentioned in Remark~\ref{rem:fisher_equi},
the description of the Fisher equilibrium is not symmetric between the two sides,
because the vertex weights on one side act as budgets of buyers.
We will instead give an alternative equilibrium characterization
based on the local maximin condition that is symmetric for agents from both sides.

\subsection{Characterizing Fisher Equilibrium via Local Maximin Conditions}
\label{sec:local_equi}

As seen in Sections~\ref{sec:sym_decomp},
\ref{sec:univeral_matching} and~\ref{sec:dist_refine},
 the local maximin condition is a key concept in the various scenarios.
Hence, we explore if a similar notion is relevant
in markets.  However, since the value of a good may be subjective
among agents, instead of considering refinements (that specify
how the objective weight of an item is distributed),
we adapt the local maximin condition to allocations (that
specify how an item is fractionally distributed).

\begin{definition}[Local Maximin Condition in Arrow-Debreu Markets]
\label{defn:local_AD}
In an Arrow-Debreu market as in Definition~\ref{defn:arrow-debreu},
an allocation~$\vec{x}$ achieves the local maximin condition
at~$i \in \mcal{I}$ if, $x_{i \to j} > 0$ implies that

$j \in \arg \min_{k \in \mcal{I}} \frac{u_k(\vec{x})}{w_k(i) \cdot w_i(k)}$,

where we use the convention that for any $x \in \R$, $\frac{x}{0} = +\infty$.

Since $u_k(\vec{x})$ is the utility of agent~$k$,
the quantity~$\frac{u_k(\vec{x})}{w_k(i)}$ measures
the utility of~$k$ in terms of the number of units of good~$i$.
Finally, from agent~$i$'s perspective, the value of good~$k$
is $w_i(k)$; hence,
the quantity $\frac{u_k(\vec{x})}{w_k(i) \cdot w_i(k)}$
measures the ``utility density'' of agent~$k$ from $i$'s perspective.
The local maximin condition means that agent~$i$ will allocate non-zero
fraction of its good only to agents~$j$ that achieve minimum utility density
from its perspective.

If an allocation $\vec{x}$ achieves the local maximin condition
at every $i \in \mcal{I}$,
we simply say that $\vec{x}$ is locally maximin.

\noindent \textbf{Special Bipartite Case.}
In a bipartite market
with bipartition $\mcal{I} = \Izero \cupdot \Ione$
of agents, for side~$\side \in \B$,
we use 
$\vec{x}^{(\side)} = x_{\Ib \to \Iob}$
to denote the complete (fractional) allocation of
items from $\Ib$ to $\Iob$.  In this case,
$\vec{x}^{(\side)}$ is locally maximin if it is locally maximin
at every agent in $\Ib$.
\end{definition}

Observe that the Fisher market is a special case
of bipartite Arrow-Debreu markets.  Moreover,
as in Remark~\ref{rem:fisher_equi},
the local maximin condition on buyer budget allocation
is equivalent to the market equilibrium condition.

\begin{fact}[Fisher Equilibrium Equivalent to Locally Maximin Buyers]
\label{fact:fisher_local}
In a Fisher market,
a buyer budget allocation $\vec{b}$ is locally maximin 
\emph{iff} it corresponds to a market equilibrium.
\end{fact}

\begin{proof}
Observe that given a buyer budget allocation $\vec{b}$,
the utility of a seller~$j$ is exactly the 
price $p_j = \sum_{i \in \mcal{B}} b_{i \to j}$.
(Moreover, the good allocation $\vec{x}$ is naturally induced
by proportional response $x_{j \to i} := \frac{b_{i \to j}}{p_j}$.)

For the local maximin condition at buyer~$i$,
$b_{i \to j} > 0$ implies that seller~$j$
achieves the minimum utility density $\frac{p_j}{B_i \cdot w_i(j)}$
among all sellers.

However, since $B_i$ is the same for all sellers,
this is equivalent to $j$ achieving the maximum $\frac{w_i(j)}{p_j}$,
which is the value-to-price ratio of good~$j$ from buyer~$i$'s perspective,
as specified in the market equilibrium condition.
\end{proof}

We will show that the local maximin condition on the allocation of the goods
is also equivalent to a market equilibrium, which has not been
investigated before.  We have briefly mentioned
the concept of proportional response in Remark~\ref{rem:fisher_equi},
which we now formally describe in a bipartite market.

\begin{definition}[Proportional Response in Bipartite Markets]
	\label{defn:PR-bipartite-arrow-debreu}
Suppose in a bipartite Arrow-Debreu market with bipartition $\mcal{I} = \Izero \cupdot \Ione$
of agents, for side~$\side \in \B$,
$\vec{x}^{(\side)}$ is an allocation on items from $\Ib$ to $\Iob$.
Then, a proportional response to $\vec{x}^{(\side)}$ is an allocation $\vec{x}^{(\ob)}$ of items from $\mcal{I}^{(\ob)}$
to $\Ib$ satisfying the following:
	\begin{itemize}
		\item If $j \in \mcal{I}^{(\ob)}$ receives a positive utility $u_j(\vec{x}^{(\side)}) = \sum_{i \in \Ib} w_j(i) \cdot x_{i \to j} > 0$ from $\vec{x}^{(\side)}$, 
		then for each $i \in \Ib$, 
		\[
		x^{(\ob)}_{j \to i} = \frac{w_j(i) \cdot x^{(\side)}_{i \to j} }{u_j(\vec{x}^{(\side)})}.
		\]
		\item If $j \in \mcal{I}^{(\ob)}$ receives zero utility from $\vec{x}^{(\side)}$, then in $x^{(\ob)}$, $j$ may distribute its item arbitrarily among $\Ib$,
		i.e., we just need $\sum_{i \in \Ib} x^{(\ob)}_{j \to i} = 1$.
	\end{itemize}
\end{definition}

The following lemma is a generalization 
of Lemma~\ref{lem:local_maximin_sym}.
To avoid degenerate cases,
we assume that every agent~$i$
has another \emph{mutually interested} agent~$j$,
i.e., $w_i(j) \cdot w_j(i) > 0$.

\begin{lemma}[Proportional Response to Locally Maximin Allocation]
\label{lem:PR-maximin-Bipartite}
Consider a bipartite Arrow-Debreu market
with bipartition $\mcal{I} = \Izero \cupdot \Ione$ such that 
every agent has at least one mutually interested agent on the other side.
%
%
%

For side~$\side \in \B$, let
$\vec{x}^{(\side)}$ be an allocation of items from 
$\Ib$ to $\Iob$ that satisfies the local maximin condition. 
Suppose the allocation $\vec{x}^{(\ob)}$ on items in $\Iob$ is the (unique) proportional response to $\vec{x}^{(\side)}$. Then, the following holds:
	
	\begin{itemize}
		\item The allocation $\vec{x}^{(\ob)}$ also satisfies the local maximin condition
		at every agent in $\Iob$.
		\item We get back $\vec{x}^{(\side)}$ as the proportional response to $\vec{x}^{(\ob)}$. 
	\end{itemize}
	
\end{lemma}

\begin{proof}
Since $\vec{x}^{(\side)}$ is locally maximin and
every agent has a mutually interested agent
on the other side, it follows that under  $\vec{x}^{(\side)}$,
every agent in $\Iob$ has positive utility.
Therefore, the proportional response
$\vec{x}^{(\ob)}$ to $\vec{x}^{(\side)}$ is unique.
When there is no risk of ambiguity,
we omit the superscript and just write $\vec{x}$.
When $k$ and $l$ are agents from different sides,
we denote $\phi_l(k) := \frac{u_k(\vec{x})}{w_k(l) \cdot w_l(k)}$.

Define $\phi_i(k) = \frac{u_k(\vec{x})}{w_k(i) \cdot w_i(k)}$,
which is the ``utility density'' score of $k$
from agent~$i$'s perspective.

For $i \in \Ib$,
since $\vec{x}^{(\side)}$
is locally maximin at~$i$
and agent~$i$ has a mutually interested agent in $\Iob$,
it follows that
$\lambda_i := \min_{k \in \Iob} \phi_i(k)$ is positive and finite.
Moreover, for $i \in \Ib$ and $j \in \Iob$,
proportional response means that
$x^{(\ob)} _{j \to i} > 0$ \emph{iff} $x^{(\side)}_{i \to j} > 0$,
in which case we must have
$w_i(j) \cdot w_j(i) > 0$ and

$$x^{(\ob)} _{j \to i} = \frac{w_j(i) \cdot x^{(\side)}_{i \to j}}{u_j(\vec{x})} 
= \frac{x^{(\side)}_{i \to j}}{\phi_i(j) \cdot w_i(j)} = \frac{x^{(\side)}_{i \to j}}{\lambda_i \cdot w_i(j)}.$$

For any agent~$k$, let $\Gamma(k)$ denote
the collection of agents~$l$ on the other side
such that both $x_{k \to l}$ and $x_{l \to k}$ are positive.

For $i \in \Ib$,
its utility from $x^{(\ob)}$ is:

$$u_i(\vec{x}) = \sum_{l \in \Gamma(i)} w_i(l) \cdot x_{l \to i}
= \sum_{l \in \Gamma(i)} w_i(l) \cdot \frac{x_{i \to l}}{\lambda_i \cdot w_i(l)}
= \frac{1}{\lambda_i}.
$$

Therefore, $x_{i \to j} > 0$ and $x_{j \to i} > 0$
implies that $u_i(\vec{x}) \cdot u_j(\vec{x}) = w_i(j) \cdot w_j(i)$.

\noindent (1) We prove the first statement. 
For $j \in \Iob$ and $k \in \Ib$, we have
	
	$$\phi_j(k) = \frac{u_k(\vec{x})}{w_k(j) \cdot w_j(k)} = \frac{1}{\lambda_k \cdot w_k(j) \cdot w_j(k)}.$$

Hence, $x^{(\ob)}_{j \to i} > 0$
implies that $\phi_j(i) = \frac{1}{u_j(\vec{x})}$.
It suffices to show that
$x^{(\ob)}_{j \to i} = 0$ implies that
$\phi_j(i) \geq \frac{1}{u_j(\vec{x})}$.

	
	
Since $x_{j \to i} = 0$, we have $x_{i \to j} = 0$. Then,
because $\vec{x}^{(\side)}$
is locally maximin at~$i$,
we have  $\frac{u_j(\vec{x})}{w_i(j) \cdot w_j(i)} = \phi_i(j) \geq \lambda_i = \frac{1}{u_i(\vec{x})}$, i.e., $u_i(\vec{x}) \cdot u_j(\vec{x}) \geq
w_i(j) \cdot w_j(i)$. Therefore, $\phi_j(i) = \frac{u_i(\vec{x})}{w_i(j) \cdot w_j(i)} \geq \frac{1}{u_j(\vec{x})}$, as required.
	
\noindent (2)	
We prove the second statement. Suppose $\vec{z}^{(\side)}$ is the proportional response to $\vec{x}^{(\ob)}$. From Definition~\ref{defn:PR-bipartite-arrow-debreu} of proportional response, we have:
	
	$$z_{i \to j} = \frac{w_i(j) \cdot x_{j \to i}}{u_i(\vec{x})}.$$

	If $x_{i \to j} = 0$, then $x_{j \to i} = 0$, then $z_{i \to j} = 0$. 
	
	If $x_{i \to j} > 0$, then
	$x_{j \to i} > 0$ and $u_i(\vec{x}) \cdot u_j(\vec{x}) = w_i(j) \cdot w_j(i)$.
	Recalling that $x_{j \to i} = \frac{w_j(i) \cdot x_{i \to j}}{u_j(\vec{x})}$,
	we have:
	
	$$z_{i \to j} = \frac{w_i(j)}{u_i(\vec{x})} \cdot
	\frac{w_j(i)}{u_j(\vec{x})} \cdot x_{i \to j} = x_{i \to j},$$
	
	as required.
	
	This completes the proof. 
\end{proof}

With Lemma~\ref{lem:PR-maximin-Bipartite},
we can get an alternative characterization
of an equilibrium goods allocation in terms of the
local maximin condition, without the need
to consider good prices.

\begin{theorem}[Fisher Equilibrium Equivalent to Locally Maximin Sellers]
	\label{th:fisher_local_seller}
	In a Fisher market in which
	every buyer has a non-zero utility from at least one good,
	an allocation $\vec{x}$ of goods to buyers 
	corresponds to a market equilibrium
		\emph{iff} $\vec{x}$ 
	is locally maximin (with respect to sellers), i.e., for each
	seller~$j \in \mcal{S}$ and buyer~$i \in \mcal{B}$,
	$x_{j \to i} > 0$ only if
	
	$$i \in \arg \min_{k \in \mcal{B}} \frac{u_k(\vec{x})}{w_j(k) \cdot w_k(j)}.$$
\end{theorem}

\begin{proof}
We apply Lemma~\ref{lem:PR-maximin-Bipartite}
and start from the locally maximin goods allocation~$\vec{x}$ to buyers.
Then, the proportional response $\vec{b}$
to $\vec{x}$ is the buyers' budget allocation to the sellers that
is locally maximin.  Moreover, $\vec{b}$ and $\vec{x}$
are proportional response to each other.
From Fact~\ref{fact:fisher_local},
the budget allocation~$\vec{b}$ corresponds to an equilibrium,
and hence, $\vec{x}$ is also an equilibrium.

Conversely, suppose a goods allocation $\vec{x}$ is an equilibrium.
This means that there exists an equilibrium budget allocation~$\vec{b}$
such that $\vec{x}$ is the proportional response to $\vec{b}$.
Fact~\ref{fact:fisher_local} implies that
$\vec{b}$ is locally maximin, and applying Lemma~\ref{lem:PR-maximin-Bipartite}
again, we conclude that $\vec{x}$ is also locally maximin.
\end{proof}

\noindent \textbf{Counter Example for General Arrow-Debreu Markets.}
Contrary to the Fisher market,
we show that the local maximin condition is not 
a sufficient condition for market equilibrium in
general Arrow-Debreu markets.
Our counterexample is a bipartite market, but
the views of both sides are subjective.

\begin{theorem}[Local Maximin Allocation
Does Not Imply Equilibrium in Bipartite Arrow-Debreu Markets]
\label{th:counterexample}
There exists a bipartite Arrow-Debreu market
with bipartition $\mcal{I} = \Izero \cupdot \Ione$
and an allocation $\vec{x}$ of agents' items such that the following holds.

\begin{compactitem}

\item The allocation $\vec{x}$ consists of a pair $(\vec{x}^{(0)},
 \vec{x}^{(1)})$ that are proportional responses to each other by 
agents from both sides, and both are locally maximin.

\item The allocation~$\vec{x}$ is not a market equilibrium.
\end{compactitem}

\end{theorem}

\begin{proof}
There are 4 agents in total, where $\Izero := \{0, 2\}$ and
$\Ione := \{1, 3\}$.  When we perform addition or subtraction on the agent indices,
we use mod~4.

The valuations of the agents on goods are defined as follows.
For each agent~$i$, $w_i(i+1) = 2$ and $w_i(i-1) = 1$.
All other valuations are zero.  It can be checked that
these valuations correspond to a bipartite market,
and every agent has a non-zero utility on the goods
of both agents on the other side.  Also, the views of both sides are subjective,
and so this is not a Fisher market.

By symmetry, the equilibrium allocation $\vec{x}^*$ is unique, and each
agent~$i$ should receive 1 unit of good~$i+1$,
i.e., $x^*_{i+1 \to i} = 1$, and all goods have
some common positive price.  Moreover, every agent has utility 2 in the equilibrium.

We next construct a locally maximin $\vec{x}$
that is not an equilibrium.  It corresponds to agents~0 and 1 exchanging
their goods, and agents~2 and 3 exchanging their goods.
This clearly satisfies proportional response.
Moreover, each agent in $\Izero$ has utility 2,
while each agent in $\Ione$ has utility 1.  Hence, $\vec{x}$ is not an equilibrium.

However, observe that if $i$ and $j$ are from different sides,
then it must be the case that $w_i(j) \cdot w_j(i) = 2$.  
Then, from the perspective of agent~$i$,
it follows that both agents~$j$ on the other side will have the same
utility density $\frac{u_j(\vec{x})}{w_i(j) \cdot w_j(i)} = 
\frac{u_j(\vec{x})}{2}$.  Hence, $\vec{x}$ is locally maximin and we
have our counterexample.
\end{proof}

\section{Achieving Approximate Refinements with Multiplicative Error in Distributed Settings}
\label{sec:mult_error}

In this section,
we show how approximate refinements with 
multiplicative error as in 
Definition~\ref{defn:approx_refine} can be computed
via the iterative proportional response 
process~\cite{wu2007proportional,zhang2011proportional,birnbaum2011distributed}
that has been proposed for approximating
market equilibrium in distributed settings.

In Section~\ref{sec:sym_decomp},
we establish the connection
between symmetric density decompositions and locally maximin refinements,
while in Section~\ref{sec:market},
we show that the local maximin condition is equivalent
to the market equilibrium condition.
Therefore, existing algorithms to compute approximate market equilibrium
with multiplicative error and their analysis~\cite{zhang2011proportional,birnbaum2011distributed} can be readily applied
to achieve the same approximation notion for refinements.
To the best of our knowledge,
this is the first approach that achieves multiplicative error
to approximate the densities of vertices in the density decomposition.

\noindent \textbf{Iterative Proportional Response Process.}
This process has been extensively analyzed
in the context of computing an approximate equilibrium
for Fisher markets~\cite{zhang2011proportional,birnbaum2011distributed} through interaction between
agents in distributed settings.
Since we will only consider refinements $\alpha^{(\side)}$ in which
every (non-isolated) node on the receiving side~$\ob$ has a non-zero payload,
we use $\msf{PR}(\alpha^{(\side)})$ to denote the unique
proportional response to $\alpha^{(\side)}$
as defined in Definition~\ref{def:PR-refinement}.
In a distributed setting,
each node in $\Ib$ is responsible
for maintaining $\alpha^{(\side)}(f)$ for
edges $f \in \mcal{F}$ incident on~$i$;
it is clear that $\msf{PR}(\cdot)$ takes one round
of communication and a node needs to communicate
only with its neighbors in $\mcal{F}$.

Following the convention in~\cite{zhang2011proportional,birnbaum2011distributed},
we also assume that the input is a distribution instance, i.e.,
the vertex weights on each side sum to 1.
We paraphrase the iterative process for our input instance as follows.

\begin{definition}[Iterative Proportional Response Process]
\label{defn:pr-process}
Given an instance $(\Izero, \Ione; \mathcal{F}; w)$,
suppose for side~$\side \in \B$,
$\alpha^{(\side)}_0$ is a refinement
on vertex weights in~$\Ib$ such the every non-isolated
vertex in $\Iob$ receives non-zero payload from $\alpha^{(\side)}_0$.

Starting from $\alpha^{(\side)}_0$,
a sequence of refinements is produced as follows.
For $t \geq 1$, assuming that $\alpha^{(\side)}_{t-1}$ is already computed,
the following refinements are computed in the $t$-th iteration:
\begin{compactitem}
\item $\alpha^{(\ob)}_{t-1} \gets \msf{PR}(\alpha^{(\side)}_{t-1})$;
\item $\alpha^{(\side)}_{t} \gets \msf{PR}(\alpha^{(\ob)}_{t-1})$.
\end{compactitem}
\end{definition}

\begin{remark}[Comparison of Terminology in Fisher Markets]
Even though our instance can be interpreted as a symmetric Fisher market
where both sides play equal roles, the iterative process in
Definition~\ref{defn:pr-process} is not symmetric between
the two sides, because the process is initiated at
some refinement on vertex weights on one side~$\side$.

For Fisher market, the initiating side~$\side$ corresponds to the buyers,
while side~$\ob$ corresponds to the sellers.
After iteration~$t$,
the refinement $\alpha^{(\side)}_t$ on vertex weights on side~$\side$ corresponds to the budget allocation
from buyers to sellers,
and the resulting payloads $p^{(\ob)}_t$ received by
the other side~$\ob$ correspond to the prices of sellers' goods in the Fisher market,
but in a symmetric market, it is more convenient
to interpret the payloads as utilities.

After getting the proportional response 
$\alpha^{(\ob)}_t \gets \msf{PR}(\alpha^{(\side)}_t)$,
the payloads received by vertices on side~$\side$
from $\alpha^{(\ob)}_t$ correspond to the buyers' utilities.

For the rest of the description,
we will interpret the problem instance
in terms of density decomposition and locally maximin refinements. 
Hence, we will phrase the results in terms of payloads, instead of prices and utilities.
Observe that the multiplicative error in Definition~\ref{defn:approx_refine}
is stated for the density $\rho(i) = \frac{p(i)}{w(i)}$, for which
it is equivalent to consider payload $p(i)$ as far as multiplicative error is concerned. 
\end{remark}

\noindent \textbf{Notation Recap.}
Given a distribution instance $(\Izero, \Ione; \mathcal{F}; w)$,
suppose we fix some pair $\vec{\alpha}_* = (\alpha^{(0)}_*, \alpha^{(1)}_*)$
of locally maximin refinements that are
proportional response to each other.
Observe that even though $\vec{\alpha}_*$ 
may not be unique, by Fact~\ref{fact:maximin_decomp},
the induced payload vectors $(p^{(0)}_*, p^{(1)}_*)$ are unique.
Given some side~$\side \in \B$,
we denote $n := |\Ib|$ and $\overline{n} := |\Iob|$.

The following fact is paraphrased from Section 5.1 in~\cite{birnbaum2011distributed}.
We denote $w_{\min} := \min_{i \in \Ib} w^{(\side)}(i)$
and $\overline{w}_{\min} := \min_{j \in \Iob} w^{(\ob)}(j)$.
Their result considers a parameter $\overline{u}_{\min} := \min_{i \in \Ib \sim j \in \Iob}
\frac{w^{(\ob)}(j)}{\sum_{k: i \sim k} w^{(\ob)}(k)} \geq 
\overline{w}_{\min}$.

\begin{fact}[Convergence Rate for Payloads~\cite{birnbaum2011distributed}]
\label{fact:convergence}
Suppose the iterative process in Definition~\ref{defn:pr-process}
is initiated on side~$\side$ with refinement $\alpha^{(\side)}_0$.
Suppose after iteration~$t$,
the refinement $\alpha^{(\side)}_t$ on vertex weights
on side~$\side$ induces the payload vector~$p^{(\ob)}_t$
on side~$\ob$, and denote 
$\overline{\eta}_t := \max_{j \in \Iob} \frac{|p^{(\ob)}_*(j) - p^{(\ob)}_t(j)|}{p^{(\ob)}_*(j)}$.
Then, we have:

$$\overline{\eta}_t^2 \leq \frac{16 n}{\overline{u}_{\min}} \cdot 
\frac{\msf{D}_{\textrm{KL}}(\alpha^{(\side)}_*  \| \alpha^{(\side)}_0)}{t}.$$

\end{fact}

\noindent \textbf{Initial Refinement.}
As stated in Fact~\ref{fact:convergence},
the convergence rate depends on the choice of
the initial refinement $\alpha^{(\side)}_0$ on the vertex weights from side~$\side$.
Observe that if $\alpha^{(\side)}_0(ij) = 0$ for some $i \sim j$,
then for all $t \geq 1$, $\alpha^{(\side)}_t(ij)$ stays zero.
Moreover, if $\alpha^{(\side)}_*(ij) > 0$,
then $\msf{D}_{\textrm{KL}}(\alpha^{(\side)}_*  \| \alpha^{(\side)}_0) = + \infty$,
and the convergence result becomes degenerate.
In the literature~\cite{zhang2011proportional,birnbaum2011distributed},
a conservative choice is to distribute that weight $w^{(\side)}(i)$ evenly among
among all $j$'s on the other side.
Specifically, for all $i \in \Ib \sim j \in \Iob$,
we have:

\begin{equation}
\label{eq:init}
\alpha^{(\side)}_0(ij) \gets \frac{w^{(\side)}(i)}{\overline{n}}.
\end{equation}

\noindent \textbf{Technical Remark.}
Observe that in (\ref{eq:init}),
we did not specify $\alpha^{(\side)}_0(ij)$ when
$i$ and $j$ are not neighbors in $\mcal{F}$.
In order for $\alpha^{(\side)}_0$ to be a distribution,
we could define $\alpha^{(\side)}_0(ij) = \frac{w^{(\side)}(i)}{\overline{n}}$
for such cases.  Hence, 
we refer to $\alpha^{(\side)}_0$ as a super-refinement,
because the support
of $\alpha^{(\side)}_0$ may be a proper superset of $\mcal{F}$.

However, note that in Definition~\ref{def:PR-refinement}
of proportional response to $\alpha^{(\side)}_0$,
any weights on pairs outside $\mcal{F}$ will be ignored.
Hence, for $t \geq 1$, 
$\alpha^{(\side)}_t$ will have support contained in $\mcal{F}$.

Note that in $\alpha^{(\side)}_0(ij)$,
we could have distributed the weight $w(i)$ evenly only
among neighbors of $i$.  However,
this would lead to a weaker result in the later
Lemma~\ref{lemma:div_other}.

\ignore{

However, actually  it suffices to distribute the weight among the neighbors~$j$ of $i$
in $\mcal{F}$.
Specifically, for $i \in \Ib \sim j \in \Iob$,
we denote its degree $\Delta^{(\side)}_i := |\{j \in \Iob: i \sim j\}|$.
Then, we can pick an initial refinement:

\begin{equation}
\label{eq:init}
\alpha^{(\side)}_0(ij) \gets \frac{w^{(\side)}(i)}{\Delta^{(\side)}_i}.
\end{equation}

We denote the maximum degree on each side
as $\Delta_{\max} := \max_{i \in \Ib} \Delta^{(\side)}_i$
and $\overline{\Delta}_{\max} := 
\max_{j \in \Iob} \Delta^{(\ob)}_j$.
}

The following fact is proved in~\cite[Lemma 13]{birnbaum2011distributed}.

\begin{fact}[Divergence between Initial and Optimal Refinements]
\label{fact:div_initial}
Using the initial (super-)refinement~$\alpha^{(\side)}_0$ in (\ref{eq:init}),
for any refinement $\alpha$ (not necessarily locally maximin) on vertex weights on side~$\side$,
the divergence between the two refinements has the following bound:

$$\msf{D}_{\textrm{KL}}(\alpha  \| \alpha^{(\side)}_0) \leq \ln (n \cdot \overline{n}).$$
\end{fact}

\noindent \textbf{Achieving Multiplicative Error for Payloads on Both Sides.}
Observe that Facts~\ref{fact:convergence} and~\ref{fact:div_initial}
together can give a bound on the multiplicative error for the payloads
on side~$\ob$, when the process is initiated on the other side~$\side$.
Observe that multiplicative error on one side's payloads
is sufficient for the applications to
universal maximum matching in Theorem~\ref{th:approx-flow}
and universal closest refinements
in Theorem~\ref{th:min_power}.

However, it is interesting to see if the iterative process can achieve
multiplicative error for both sides' payloads simultaneously.
A careful study of the proofs in~\cite{zhang2011proportional,birnbaum2011distributed} reveals
that it is also possible to derive an upper bound
for the multiplicative error for payloads on side~$\side$.  However,
since our instance is symmetric on both sides,
a better and simpler bound can be achieved by deriving
a variant of Fact~\ref{fact:div_initial}
for the proportional response $\alpha^{(\ob)}_0$ to $\alpha^{(\side)}_0$.
Then, it suffices to consider initiating the process
on side~$\ob$ with the refinement~$\alpha^{(\ob)}_0$.

\begin{lemma}[Divergence for the Other Side]
\label{lemma:div_other}
Suppose the (super-)refinement $\alpha^{(\side)}_0$
on vertex weights on side~$\side$ is defined as in~(\ref{eq:init}),
and $\alpha^{(\ob)}_0 = \msf{PR}(\alpha^{(\side)}_0)$ is its proportional response.
Then, for any refinement~$\overline{\alpha}$ of vertex weights on side~$\ob$,
we have:

$$\msf{D}_{\textrm{KL}}(\overline{\alpha}  \| \alpha^{(\ob)}_0) \leq \ln \frac{\overline{n}}{u_{\min}},$$

where $u_{\min} := 
\min_{i \in \Ib \sim j \in \Iob} \frac{w^{(\side)}(i)}{\sum_{k: k \sim j} w^{(\side)}(k) }.$

\end{lemma}

\begin{proof}
By definition
of proportional response, for $i \in \Ib \sim j \in \Iob$,

$${\alpha}^{(\ob)}_0(ij) = 
\frac{{\alpha}^{(\side)}_0(ij)}{\sum_{k: k \sim j} {\alpha}^{(\side)}_0(kj)} \cdot w^{(\ob)}(j)
= 
\frac{w^{(\side)}(i)/\overline{n}}{\sum_{k: k \sim j} w^{(\side)}(k)/\overline{n}} \cdot w^{(\ob)}(j) \geq u_{\min} \cdot w^{(\ob)}(j).$$ 

Next, we give a bound on the divergence.
The first inequality below follows because entropy
of a distribution is non-negative.

\begin{align*}
\msf{D}_{\textrm{KL}}(\overline{\alpha}  \| \alpha^{(\ob)}_0) 
&= \sum_{i \sim j} \overline{\alpha}(ij)
 \ln \frac{\overline{\alpha}(ij)}{\alpha^{(\ob)}_0(ij)}
\leq
\sum_{i \sim j} \overline{\alpha}(ij)
 \ln \frac{1}{\alpha^{(\ob)}_0(ij)} \\
& \leq
\sum_{i \sim j} \overline{\alpha}(ij)
 \ln \frac{1}{u_{\min} \cdot w^{(\ob)}(j)} 
= \sum_{i \sim j} \overline{\alpha}(ij)
 \ln \frac{1}{u_{\min}} 
+
\sum_{j \in \Iob} w^{(\ob)}(j)
 \ln \frac{1}{w^{(\ob)}(j)} 
\\
& \leq
 \ln \frac{1}{u_{\min}} 
+
 \ln \overline{n},
\end{align*}
as required.
\end{proof}

Facts~\ref{fact:convergence},~\ref{fact:div_initial}
and Lemma~\ref{lemma:div_other} readily give the main conclusion
of this section.

\begin{theorem}[Multiplicative Error for Both Sides' Payloads]
\label{th:both_sides}
Suppose the iterative process in Definition~\ref{defn:pr-process}
is initiated on side~$\side$ with (super-)refinement $\alpha^{(\side)}_0$
as in (\ref{eq:init}).
Suppose after iteration~$t$,
the refinement $\alpha^{(\side)}_t$ on vertex weights
on side~$\side$ induces the payload vector~$p^{(\ob)}_t$
on side~$\ob$
and the refinement $\alpha^{(\ob)}_t \gets \msf{PR}(\alpha^{(\side)}_t)$
on side~$\ob$ induces the payload vector
~$p^{(\side)}_t$ on side~$\side$.

Denote 
${\eta}_t := \max_{i \in \Ib} \frac{|p^{(\side)}_*(i) - p^{(\side)}_t(i)|}{p^{(\side)}_*(i)}$
and
$\overline{\eta}_t := \max_{j \in \Iob} \frac{|p^{(\ob)}_*(j) - p^{(\ob)}_t(j)|}{p^{(\ob)}_*(j)}$.
Then, we have:

$${\eta}_t^2 \leq \frac{1}{t} \cdot \frac{16 \overline{n}}{u_{\min}} \cdot 
\ln \frac{\overline{n}}{u_{\min}};$$

$$\overline{\eta}_t^2 \leq \frac{1}{t} \cdot \frac{16 n}{\overline{u}_{\min}} \cdot 
\ln (n \cdot \overline{n}).$$

In other words, for $0 < \tau \leq \frac{1}{2}$, 
to achieve $\tau$-multiplicative error, we have the following:

\begin{compactitem}

\item $T \geq 
\frac{1}{\tau^2} \cdot \frac{16 \overline{n}}{u_{\min}} \cdot 
\ln \frac{\overline{n}}{u_{\min}} \implies
\eta_T \leq \tau$;

\item $T \geq \frac{1}{\tau^2} \cdot \frac{16 n}{\overline{u}_{\min}} \cdot 
\ln (n \cdot \overline{n})   \implies 
\overline{\eta}_T \leq \tau$.
\end{compactitem}

\end{theorem}

\begin{remark}
As aforementioned, using the proofs in~\cite{zhang2011proportional,birnbaum2011distributed},
to achieve $\eta_T \leq \tau$,
we would need
$T \geq \Omega(\frac{1}{\tau^2} \cdot \frac{n}{\overline{u}_{\min} \cdot w^2_{\min}} 
\cdot  \log(n \cdot \overline{n}))
\geq \Omega(\frac{1}{\tau^2} \cdot \frac{n^3}{\overline{u}_{\min}} 
\cdot  \log(n \cdot \overline{n}))$.

\end{remark}

\section{Review: Achieving Absolute Error for Density Vectors in Previous Works}
\label{sec:additive_error}

For completeness,
we review how absolute error can be achieved
in previous works~\cite{DBLP:conf/www/DanischCS17,harb2022faster}.
In these approaches, the hypergraph interpretation of the instances is used,
and the two sides do not play the same role.
In previous presentations, the hyperedges can have arbitrarily positive weights,
but the nodes have uniform weights.  
We show that these approaches can easily be adapted such that nodes can have arbitrary
positive weights as well.

\noindent \textbf{Simplified Notation.}
Since the hypergraph interpretation is used,
we will denote hypergraph $H = (V, E)$,
where edges $E = \Izero$ and nodes $V = \Ione$ may both
have arbitrary
positive weights $w : V \cup E \to \R_{>0}$.
Recall that for $e \in E$ and $i \in V$,
$\{e, i\} \in \mcal{F}$ represents a pair neighbors
in the instance \emph{iff} $i \in e$ holds in the hypergraph $H$.
Moreover, we use $\alpha: \mcal{F} \to \R_{\geq 0}$
to denote a refinement of edge weights,
where $\alpha_{e \to i}$ is the payload node~$i$ received from
edge~$e$.  Consequently,
for each node~$i \in V$,
we denote its totally received payload $p_i = \sum_{e \in E: i \in e} \alpha_{e \to i}$
and its density $\rho_i = \frac{p_i}{w_i}$.
\ignore{
Moreover, since we no longer need to use
the superscript~$\side$ to indicate the side,
we will use the superscript to indicate
an optimal solution such as $\alpha_*$ or $\rho_*$ in
this section.
}

\noindent \textbf{Quadratic Program.}
Both works~\cite{DBLP:conf/www/DanischCS17,harb2022faster}
considered the same convex program which can be readily extended
to hypergraphs with general edge and node weights as follows.
The feasible set $\mcal{K}(H) :=
\left\{\alpha \in \mathbb{R}^{|\mathcal{F}|}_{\geq 0} \, \middle| \, \forall e \in E, \\ \sum_{i \in V: i \in e} \alpha_{e \to i} = w_e \right\}$
consists of refinements of hyperedge weights.
The objective function $Q$ is given in the following quadratic program.

\begin{align*}
	\mathsf{CP}(H): \quad \min \quad  Q(\alpha) & = \sum_{i \in V}\frac{p_i^2}{w_i} = \sum_{i \in V} w_i \cdot \rho_i^2 \\
	\text{s.t.} 
	 \quad p_i & = \sum_{e \in E: i \in e} \alpha_{e \to i}
	= w_i \cdot \rho_i, \quad \forall i \in V \\
		\alpha & \in \mcal{K}(H) \\
\end{align*}

A variant of the following fact was given
as early as in~\cite{fujishige1980lexicographically},
and a proof for uniform node weights
is given in~\cite[Corollary 4.4]{DBLP:conf/www/DanischCS17},
which is easily extended to general weights.

\begin{fact}
Every optimal solution~$\alpha_*$
to $\mathsf{CP}(H)$ induces the
density vector $\rho_*$ on nodes in the density decomposition
in Definition~\ref{def:decomposition}.
\end{fact}

As mentioned in Definition~\ref{defn:error},
absolute error is defined with respect to a norm.
We consider the following norm for density vectors on nodes in~$V$.

\begin{definition}[Norm $\| \cdot \|_w$]
\label{defn:norm_w}
We define a norm $\| \cdot \|_w$
on the space of density vectors in~$V$ as follows.
For $\rho \in \R^V$,
	$\|\rho\|_w^2 = \sum_{i \in V} w_i \cdot \rho_i^2$. 
\end{definition}

	%
	%

\begin{remark}
Comparing with the standard Euclidean norm,
we have: $w_{\min} \cdot \| \rho \|_2^2 \leq \| \rho \|_w \leq
w_{\max} \cdot \| \rho \|_2^2$.
Hence, $\epsilon$-absolute error with respect to norm $\| \cdot \|_w$ implies 
$\frac{\epsilon}{\sqrt{w_{\min}}}$-absolute error with respect to the
Euclidean norm.

Furthermore, observe that given density vector $\rho$ and
the corresponding payload vector~$p$,
we have $\| \rho \|^2_w = \sum_{i \in V} \frac{p_i^2}{w_i}$.  Therefore,
we would need to define another norm if we wish to express our results
in terms of the payload vector.  For the rest of the section,
we will focus on approximating the density vector~$\rho_*$.
\end{remark}

The following lemma
is generalized from~\cite{DBLP:conf/www/DanischCS17}.
It shows that to get an absolute error for $\rho$,
it suffices to consider an additive error for the objective
function in $\mathsf{CP}(H)$.

\begin{lemma}[Additive Error for Quadratic Program]
\label{lem:error-QP}
Suppose
an edge refinement $\alpha \in \mcal{K}(H)$ induces
the density vector $\rho \in \R^V$.
Then, comparing with the optimal~$\rho_*$ induced by an optimal~$\alpha_*$,
we have:

$\|\rho - \rho_* \|_w^2 \leq Q(\alpha) - Q(\alpha_*)$.
\end{lemma}

\begin{proof}
Denote $\epsilon := \|\rho - \rho_*\|_w$.
Since $\mcal{K}(H)$ is convex, 
we denote $\alpha(\theta) := \alpha_* + \theta(\alpha - \alpha_*) \in \mcal{K}(H)$.
We define a function $\gamma : [0,1] \to \R$
by $\gamma(\theta) := Q(\alpha(\theta))
= 
\sum_{i \in V} w_i \cdot \left( \rho_*(i) + \theta (\rho_i - \rho_*(i)) \right)^2
$, where the last equality follows because
the density vector is a linear function of the hyperedge refinement.

Observe that since $\alpha_*$ is an optimal solution, we have $\gamma'(0) \geq 0$. 
Moreover, 
$\gamma(\theta)$ is a second-degree polynomial in~$\theta$,
it can be easily verified that $\gamma''(\theta) = 2 \epsilon^2$ for $\theta \in [0,1]$.

Integrating this twice over $\theta \in [0,1]$, we have
$\gamma(1) - \gamma(0) = Q(\alpha) - Q(\alpha_*) = \epsilon^2 + \gamma'(0) \geq \epsilon^2$,
as required.
\end{proof}

\noindent \textbf{First-Order Iterative Methods to Solve Convex Program.}
Both approaches~\cite{DBLP:conf/www/DanischCS17,harb2022faster}
consider iterative gradient methods to tackle $\mathsf{CP}(H)$.
The difference is that the standard Frank-Wolfe method~\cite{DBLP:conf/icml/Jaggi13}
is used in~\cite{DBLP:conf/www/DanischCS17},
while the projected gradient descent with Nesterov momentum~\cite{Nesterov1983AMF} method 
(such as accelerated FISTA~\cite{beck2009fast})
is used in~\cite{harb2022faster}.
For completeness, we recap both methods.
In first-order methods,
the convergence rate depends
on a smoothness parameter $L_Q := \sup_{x \neq y \in \mcal{K}(H)} \frac{\|\nabla Q(x) - \nabla Q(y) \|_2}{\|x - y\|_2}$.
If $Q$ is twice differentiable,
this is equivalent to the supremum of the spectral norm $\|\nabla^2 Q(\alpha)\|$ of the Hessian
over $\alpha \in \mcal{K}(H)$.

\begin{algorithm}[H]
	\SetAlgoLined
	\SetKwInput{Input}{Input}
	\SetKwInput{Output}{Output}
	\SetKwFunction{FrankWolfe}{Frank-Wolfe}
	\SetKwProg{Function}{function}{:}{}
	\Input{Objective function $Q$, feasible set $\mcal{K}(H)$, number~$T$ of iterations.}
	\Output{$\alpha_{T}$}
	
		Set initial $\alpha_{0} \in \mcal{K}(H)$ arbitrarily\;
		\For{$t = 1$ to $T$}{
			$\gamma_t \gets \frac{2}{{t+2}}$\;
			$\hat{\alpha} \gets \arg\min_{\alpha \in \mcal{K}(H)} \langle \alpha, \nabla Q(\alpha_{t-1}) \rangle$\;
			$\alpha_{t} \gets (1 - \gamma_t) \cdot \alpha_{t-1} + \gamma_t \cdot \hat{\alpha}$\;
		}
		\Return $\alpha_{T}$\;
		
	\caption{Frank-Wolfe Algorithm}
	\label{algo:FW-basic}
\end{algorithm}

\begin{algorithm}[H]
	\SetKwInput{Input}{Input}
	\SetKwInput{Output}{Output}
	\SetKwInput{Initialization}{Initialization}
	\Input{Objective function $Q$, feasible set $\mcal{K}(H)$, number~$T$ of iterations.}
	Set learning rate $\gamma \gets \frac{1}{L_Q}$\;
	Set initial $\alpha_{0} = \beta_0 \in \mcal{K}(H)$ arbitrarily\;
	\For{$t = 1$ to $T$}{\
		\Comment{Projection operator $\prod_{\mcal{K}(H)}$ to feasible set}
		\(\alpha_{t} = \prod_{\mcal{K}(H)}(\beta_{t-1} - \gamma \cdot \nabla Q(\beta_{t-1}))\)\; 
		$\beta_t = \alpha_t + \frac{t-1}{t+2} \cdot (\alpha_t - \alpha_{t-1})$;
	}
	\Output{\(\alpha_{T}\)}
	\caption{Accelerated FISTA}
	\label{algo:FISTA}
\end{algorithm}

\begin{fact}[Running Time]
\label{fact:time}
When applied to $\mathsf{CP}(H)$,
each iteration of Algorithm~\ref{algo:FW-basic}
takes $O(|\mcal{F}|)$ time
and each iteration of Algorithm~\ref{algo:FISTA}
takes $O(|\mcal{F}| \log r_H)$ time,
where $|\mcal{F}| := \sum_{e \in E} |e|$
and $r_H := \max_{e \in E} |e|$ 
is the maximum number of nodes contained in a hyperedge.
\end{fact}

\begin{proof}
Observe that for $i \in e$, $\nabla Q(\alpha)_{e \to i} = \frac{2 p_i}{w_i}$.  Hence, each iteration
of Algorithm~\ref{algo:FW-basic} takes $O(|\mcal{F}|)$ time.

We next analyze the running time for each iteration of Algorithm~\ref{algo:FISTA}, which is dominated by the projection step.
Observe that given a vector \(y \in \mathbb{R}^{|\mathcal{F}|}_+\), we can project it on $\mcal{K}(H)$ by finding \(\alpha\) that minimizes \(\|\alpha - y\|_2^2\) subject to \(\alpha \in \mcal{K}(H)\). This is known as the simplex projection and it has a simple closed form solution; see \cite{wang2013projection} for the basic algorithm, and \cite{iutzeler2018distributed} for a recent distributed variant of the algorithm. In any case, the projection will take $O(|e| \log |e|)$ time for each 
$e \in E$ to do the projection.
Summing over $e \in E$, and hence overall the projection step takes time $O(|\mcal{F}| \log r_H)$. 
\end{proof}

\begin{fact}[Convergence Rates of First-Order Methods]
\label{fact:conv_first}
Suppose $\alpha_*$ is an optimal solution to $\mathsf{CP}(H)$.
Then, after $T \geq 1$ iterations for each of the following
algorithms, we have the corresponding upper
bound on the additive error $Q(\alpha_T) - Q(\alpha_*)$:

\begin{compactitem}

\item \emph{Frank-Wolfe} (Algorithm~\ref{algo:FW-basic})~\textnormal{\cite{DBLP:conf/icml/Jaggi13}}:

$$\frac{1}{T+2} \cdot \mathsf{Diam}(\mcal{K}(H))^2 \cdot L_Q.$$

\item \emph{Accelerated FISTA} (Algorithm~\ref{algo:FISTA})~\textnormal{\cite{beck2009fast}}:

$$\frac{2}{T^2} \cdot \mathsf{Diam}(\mcal{K}(H))^2 \cdot L_Q.$$
\end{compactitem}
\end{fact}

We next give upper bounds on 
$\mathsf{Diam}(\mcal{K}(H))$ and $L_Q$.

\begin{fact}[Upper Bound on Diameter]
\label{fact:diam}
We have $\mathsf{Diam}(\mcal{K}(H))^2 \leq 2 \sum_{e \in E} w_e^2$.
\end{fact}

\begin{proof}
	For any $\alpha, \beta \in \mcal{K}(H)$, 
	
	$\|\alpha - \beta\|^2 = \sum_{e \in E} \sum_{i \in V} (\alpha_{e \to i} - \beta_{e \to i})^2 \leq 
	\sum_{e \in E} \sum_{i \in V} (\alpha_{e \to i}^2 + \beta_{e \to i}^2)$ 
	
	$\leq \sum_{e \in E} \{\left( \sum_{i \in V} \alpha_{e \to i} \right)^2 + \left( \sum_{i \in V} \beta_{e \to i} \right)^2 \}= 2 \sum_{e \in E} w_e^2$. 
\end{proof}

\begin{fact}[Upper Bound on $L_Q$]
\label{fact:L_Q}
We have $L_Q \leq 2 \Delta_w(H)$,
where $\Delta_w(H) := \max_{i \in V} \frac{d_i}{w_i}$ 
and $d_i := |\{e \in E: i \in e\}|$ is the number
of hyperedges containing~$i$.
\end{fact}

\begin{proof}
For any $\alpha, \beta \in \mathbb{R}^{|\mathcal{F}|}_+$, 
for $i \in e$,
denote $\nabla Q(\alpha)_{e \to i} = \frac{2 p_i}{w_i}$
and $\nabla Q(\beta)_{e \to i} = \frac{2 q_i}{w_i}$.
Then, we have:
	
	\begin{align*}
		\|\nabla Q(\alpha) - \nabla Q(\beta)\|^2 &= \sum_{i \in V} \sum_{e \in E: i \in e}\left(\frac{2 p_i}{w_i} - \frac{2q_i}{w_i}\right)^2  =
		\sum_{i \in V} \frac{4 d_i}{{w_i}^2} \cdot (p_i - q_i)^2 \\
		&= \sum_{i \in V} \frac{4 d_i}{{w_i}^2} \left(\sum_{e \in E: i \in e} (\alpha_{e \to i} - \beta_{e \to i}) \right)^2 \\ 
		&\leq \sum_{i \in V} \frac{4 d_i^2}{{w_i}^2} \sum_{e \in E: i \in e} (\alpha_{e \to i} - \beta_{e \to i})^2 \\
		&\leq (2\Delta_w(H))^2 \sum_{i \in V} \sum_{e \in E} (\alpha_{e \to i} - \beta_{e \to i})^2 \\&= (2\Delta_w(H))^2 \cdot \|\alpha - \beta\|^2, 
	\end{align*}
	where the first inequality follows from Cauchy-Schwarz inequality. 
\end{proof}

From Lemma~\ref{lem:error-QP},
Facts~\ref{fact:conv_first},~\ref{fact:diam}
and~\ref{fact:L_Q},
we have the following bounds on the absolute error.

\begin{corollary}[Absolute Error on Density Vector]
\label{cor:abs_error}
Suppose $\rho_*$ is density vector induced
by an optimal solution to $\mathsf{CP}(H)$.
Then, after $T \geq 1$ iterations for each of the following
algorithms, we have the corresponding upper
bound on the absolute error 
$\|\rho_T - \rho_* \|_w$ on the density vector:

\begin{compactitem}

\item \emph{Frank-Wolfe} (Algorithm~\ref{algo:FW-basic}):

$$2\sqrt{\frac{\Delta_w(H) \sum_{e \in E} w_e^2}{T+2}}.$$

\item \emph{Accelerated FISTA} (Algorithm~\ref{algo:FISTA}):

$$\frac{\sqrt{8 \Delta_w(H) \sum_{e \in E} w_e^2 }}{T}.$$
\end{compactitem}
\end{corollary}

\bibliographystyle{alpha}
\bibliography{ref}

\end{document}